\documentclass[aps,pra,twocolumn,notitlepage,nofootinbib,superscriptaddress,longbibliography]{revtex4-2}

\usepackage{amsthm}
\usepackage{amsmath}
\usepackage{amssymb}
\usepackage{braket}
\usepackage{tikz}
\usepackage{quantikz}

\usepackage{graphicx}
\usepackage{dcolumn}

\usepackage[hidelinks]{hyperref}
\usepackage{thmtools}
\usepackage{thm-restate}

\declaretheorem[name=Theorem]{theorem}
\declaretheorem[name=Lemma,sibling=theorem]{lemma}
\declaretheorem[name=Proposition,sibling=theorem]{proposition}
\declaretheorem[name=Corollary,sibling=theorem]{corollary}
\declaretheorem[name=Definition,sibling=theorem]{definition}
\declaretheorem[name=Remark,sibling=theorem]{remark}

\newcommand{\opr}{\operatorname}

\newcommand{\tr}{\opr{Tr}}
\newcommand{\cC}{\mathcal{C}}
\begin{document}
\title{Structure, Optimality, and Symmetry in Shadow Unitary Inversion}

\author{Guocheng Zhen}
\thanks{Guocheng Zhen and Yu-Ao Chen contributed equally to this work.}
\affiliation{Thrust of Artificial Intelligence, Information Hub,\\
The Hong Kong University of Science and Technology (Guangzhou), Guangdong 511453, China}
\author{Yu-Ao Chen}
\thanks{Guocheng Zhen and Yu-Ao Chen contributed equally to this work.}
\affiliation{Thrust of Artificial Intelligence, Information Hub,\\
The Hong Kong University of Science and Technology (Guangzhou), Guangdong 511453, China}
\author{Mingrui Jing}
\affiliation{Thrust of Artificial Intelligence, Information Hub,\\
The Hong Kong University of Science and Technology (Guangzhou), Guangdong 511453, China}
\author{Jingu Xie}
\affiliation{Thrust of Artificial Intelligence, Information Hub,\\
The Hong Kong University of Science and Technology (Guangzhou), Guangdong 511453, China}
\author{Ranyiliu Chen}
\email{chenranyiliu@quantumsc.cn}
\affiliation{Quantum Science Center of Guangdong-Hong Kong-Macao Greater Bay Area, Shenzhen 518045, China}
\author{Xin Wang}
\email{felixxinwang@hkust-gz.edu.cn}
\affiliation{Thrust of Artificial Intelligence, Information Hub,\\
The Hong Kong University of Science and Technology (Guangzhou), Guangdong 511453, China}

\date{\today}

\begin{abstract}
Reversing unitary operations is a key task in quantum computing and quantum control. In this work, we introduce and develop the framework of shadow unitary inversion, a relaxed variant of unitary inversion in which the goal is to reproduce the action of the inverse unitary only at the level of the expectation value of a fixed observable. This task captures an operational setting in which only shadow information is required and allows query complexities significantly below those of full unitary inversion. We establish a dimension-dependent lower bound showing that any $t$-query scheme requires $t$ to scale at least linearly with the system dimension, with the constant determined by the spectral properties of the target observable. In the qubit case, we construct a deterministic three-query sequential protocol that achieves exact shadow inversion, and we provide a complete characterization of all admissible qubit channels satisfying the shadow constraint. Numerical evidence suggests that three queries are optimal. For higher-dimensional systems, we develop a semidefinite-programming formulation for optimizing shadow-inversion combs and introduce a representation-theoretic symmetry reduction that decomposes the problem into invariant blocks, substantially reducing the problem size. These results provide the first systematic study for shadow unitary inversion and establish its resource requirements and symmetry structure across dimensions.
\end{abstract}
                              
\maketitle

\section{Introduction}

In quantum information science, the ability to reverse an unknown unitary transformation represents a fundamental challenge that lies at the heart of quantum control, error correction, and information recovery~\cite{Nielson-chuang,Reverse-2019}. 
A unitary operation $U$ describes an isolated quantum evolution, and its inverse $U^{-1}$ corresponds to effectively undoing the associated dynamical process, thereby restoring the system to its previous state. {Moreover, the ability to efficiently perform the inversion of unitaries has been proven to boost certain quantum information processing tasks \cite{Tang:2025owa}.}
When $U$ is completely characterized, constructing $U^{-1}$ is straightforward through physical operations governed by the inverse Hamiltonian~\cite{QC-2002,Khaneja2001}. 
However, in many realistic situations, such as when dealing with black-box quantum devices or untrusted quantum channels, the exact description of $U$ is not available \emph{a priori}. Hence determining and implementing $U^{-1}$ without explicit knowledge of $U$ is of profound importance. 

{For exact and deterministic unitary reversion of unknown unitaries, the work~\cite{QURA} solved this problem by developing the quantum unitary reversal algorithm for quantum systems with arbitrary dimension $d$, requiring $\mathcal{O}(d^2)$ queries of the unknown unitary, which has been proven to be the optimal scaling \cite{Odake2024analytical,Chen2025tight}. For the qubit case, simpler algorithms have also been developed \cite{Rever-unknown,mo2025parameterized}. As this query complexity remains costly for near-term devices, several relaxed but operationally meaningful variants of unitary reversion have been explored. These include virtual unitary reversion \cite{zhu2024reversing}, implemented via nonphysical HPTP maps followed by classical post-processing, as well as probabilistic \cite{Storage-Retrieval,Prob-exact, Invert-general} and approximate \cite{Haah2023query,Chen2025tight} inversion schemes. Recently, the work \cite{mo2025} proposed methods of inversion of unitaries with structured Hamiltonians.}

{In this work we introduce a different relaxation of the unitary reversion problem, which we call \emph{shadow inversion} (Fig. \ref{fig:diagram}). Instead of requiring full reversal of the unitary, shadow inversion demands correctness only under a fixed measurement.} This relaxation is meaningful, as in many quantum information tasks only shadow information is relevant \cite{Shadow-Tomo}. The framework of shadow information was formalized in \cite{Shadow-Tomo} and later extended in the theory of classical shadows \cite{Huang2020}, which provides both theoretical and practical scalability. More recently, the concept has been further generalized to study information recoverability in noisy quantum environments \cite{recover-shadow}.
Naturally, one would ask what is the minimum query complexity for implementing the shadow inversion of a unitary under a given observable, and whether it can be fundamentally lower than that of the complete inversion process. 

This work establishes the first systematic framework for shadow unitary inversion and provides following three main contributions. First, we present an exact three-query construction for qubit systems and provide a complete structural characterization of all qubit channels that satisfy the shadow inversion condition. Second, we establish a general lower bound showing that any $t$-query shadow unitary inversion protocol in dimension $d$ must satisfy $t = O(d)$. Third, we develop a semidefinite-programming framework for arbitrary dimension and introduce a symmetry-based reduction that decomposes the problem into invariant blocks, substantially reducing the number of free variables. 

\begin{figure*}
    \centering
    \includegraphics[width=0.70\linewidth]{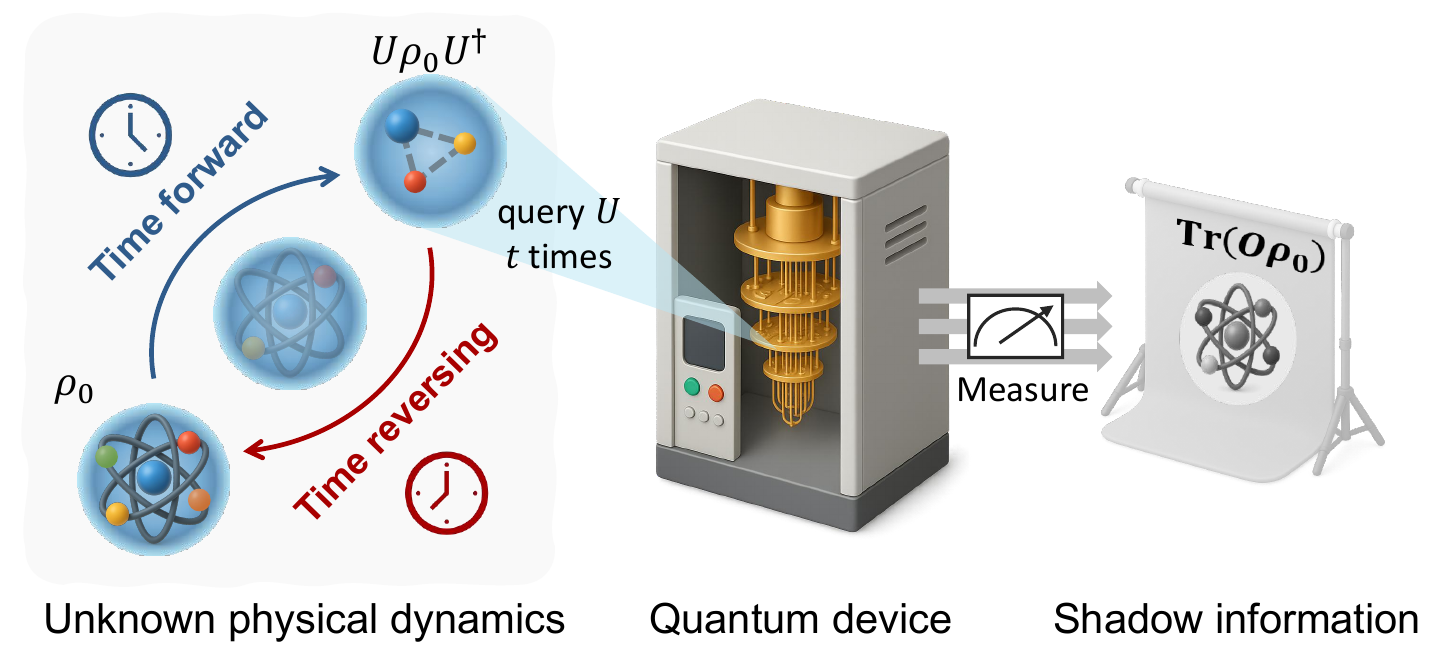}
    \caption{The general framework of shadow inversion problem with respect to the unknown physical dynamics. With the access of $t$-query of the unknown unitary evolution $U$, our quantum device target to reproduce the shadow information of the initial state $\rho_0$ with respect to the observable $O$.}
    \label{fig:diagram}
\end{figure*}

\section{General formulation and fundamental limit}
\label{prob-formu}
With the motivation and conceptual picture in place, we now introduce the general formalism that governs shadow unitary inversion, which aims to reproduce, for all input states, the expectation value of a fixed observable after the inverse of an unknown unitary evolution.

We consider a system of dimension $d$ and a fixed observable $O$. A circuit $\mathcal{N}$ is allowed to adaptively (sequentially) query an unknown unitary $U$ exactly $t$ times. Its action depends on the choices made by the circuit between queries and on any internal memory system. For every $U$, we denote the effective channel implemented by the circuit as $\mathcal{N}_U$.

\begin{definition}[Shadow unitary inversion]

{For $d,t \in \mathbb{N}^+$, let $O$ be a $d$-dimensional observable. A quantum circuit $\mathcal{N}$ is said to be a \emph{$t$-query shadow inversion of $d$-dimensional unitaries under $O$}, if for any unitary $U \in \opr{U}(d)$, $\mathcal{N}$ query $U$ exactly $t$ times, and the output circuit $\mathcal{N}_U$ satisfies}
\begin{equation}\label{shadow-inversion}
\operatorname{Tr}\!\bigl[\mathcal{N}_{U}(\rho)\, O \bigr]
    = \operatorname{Tr}\!\bigl[U^{\dagger} \rho U O \bigr]
\end{equation}
for all density operators $\rho\in\mathsf{D}(\mathbb{C}^d)$.
\end{definition}

An equivalent Heisenberg-picture statement of the shadow inversion condition is
\begin{equation*}
    \tr \bigl[\rho \mathcal{N}_{U}^{\dagger}(O)\bigr] = \tr \bigl[\rho UOU^{\dagger}\bigr],~\forall \rho\Longleftrightarrow{} \mathcal{N}_{U}^{\dagger}(O) = UOU^{\dagger}
\end{equation*}
which requires the adjoint channel to reproduce the action of $U$ on the single operator $O$ for all unitary $U$. This formulation allows us to quantify the resources required for exact shadow inversion in arbitrary dimension. The following theorem establishes a $O(d)$ lower bound on the number of queries needed by any protocol. The bound, not surprisingly, depends on the spectral structure of $O$.

\begin{restatable}{theorem}{lbthm}\label{lb-thm}
For any $d,t \in \mathbb{N+}$, if $\mathcal{N}$ is a $t$-query shadow inversion of $d$-dimensional unitaries under some fixed non-trivial $d$-dimensional Hermitian observable $O \not\propto I_{d}$, then we have the following lower-bound
    \begin{equation}\label{lb-query}
        t \geq d \cdot \max_{k = 1, \cdots, d}\Big\{ \frac{\max_{1 \leq i \leq d}|\lambda_{i} - \lambda_{k}|}{\sum_{i=1}^{d}|\lambda_{i} - \lambda_{k}|} \Big\} - 1
    \end{equation}
    where $\lambda_{1} \geq \cdots \geq \lambda_{d}$ are the eigenvalues of $O$.
\end{restatable}

    

The factor $\max_{k = 1, \cdots, d}\Big\{ \frac{\max_{1 \leq i \leq d}|\lambda_{i} - \lambda_{k}|}{\sum_{i=1}^{d}|\lambda_{i} - \lambda_{k}|} \Big\}$ reflects the level of concentration of the spectrum of $O$. For example, when $O$ corresponds to a $\{P,I-P\}$ measurement where $P$ is a rank-$k$ ($k\le d/2$) projection, it give a lower-bound $t\ge d/k-1$. Overall, it contrasts with the $\Theta(d^2)$ query requirement for full inversion, showing that shadow inversion admits substantially lower complexity while still exhibiting a meaningful information-theoretic limit. We refer the proof of theorem \ref{lb-thm} to Appendix \ref{sec:lb}. 

\section{Qubit shadow unitary inversion}

The general framework above establishes the constraints that every shadow inversion protocol must satisfy, together with a dimension-dependent lower bound on query complexity. We now focus on the qubit case, which is the simplest nontrivial setting and allows a complete characterization. In this regime we construct an explicit protocol achieving exact shadow inversion using three queries and then show that all exact solutions necessarily take a very restricted form. Numerical evidence further suggests that three queries are indeed optimal, distinguishing the qubit case as a fully understood instance of shadow inversion.

Consider the spectrum decomposition $O = V \Sigma V^\dagger$, where $\Sigma$ is real diagonal. In qubit case, simply appending $V$ at the end of the circuit reduces the problem to the computational basis scenario, therefore it suffices to study the case $O = Z$, the Pauli-Z operator. Throughout this section we reduce the problem to reproducing the statistics of $U^\dagger\rho U$ with respect to the $Z$ measurement.

Our first result establishes that three queries always suffice.

\begin{restatable}{theorem}{maintheorem}\label{main-thm}
{For any fixed $2$-dimensional observable $O$, there exists a $3$-query shadow inversion of $2$-dimensional unitaries under $O$.} 
\end{restatable}

The construction arises from an optimization of quantum circuits with oracle accesses to $U$ in the language of quantum combs and link product \cite{Chiribella2008quantum}, subject to the shadow inversion constraint. Specifically, we consider the semidefinite program
\begin{equation}\label{SDP-setting}
\begin{aligned}
    \min_{C} \int_{\opr{U}(d)} \|\tr_{2}\big[(C \ast |U\rangle\!\rangle\!\langle\!\langle U|^{\otimes t})^{\intercal}&(I_d \otimes O^{\intercal})\big]\\ &- UOU^{\dagger}\|_Fd\mu_{H},        
\end{aligned}
\end{equation}
where the minimization is over all $t$-slot quantum combs $C$. Here $\mu_{H}$ is the Haar measure on the Unitary group $\opr{U}(d)$ and $\|\cdot\|_F$ is Frobenius norm. A solution achieving zero objective value of \eqref{SDP-setting} provides a comb that implements exact shadow inversion. For $t = 3$ we obtain a feasible solution, from which an explicit circuit can be extracted. We defer the precise construction as well as the proof of correctness to Appendix \ref{proof-of-thm}. 

Numerical evidence (Table \ref{tab:seq-vs-paral}) indicates that the lower-bound of $t$ for $t$-query circuit achieving shadow inversion of $2$-dimensional unitaries under fixed $2$-dimensional observable $O$ is $3$ which suggests the construction in Theorem~\ref{main-thm} may be optimal. 
The code for this optimization problem is available at ref \cite{ZhenGuocheng_shadow_unitary_inversion_numerics}. We also remark that, parallel quantum combs may not be able to achieve shadow qubit-unitary inversion within 3 slots, constituting an example where adaptive processing is essential.

\begin{table*}
\caption{Comparison of sequential and parallel quantum combs for shadow inversion of $2$-dimensional unitary under any fixed $2$-dimensional observable. Reported values are the solutions of SDP problem \eqref{SDP-setting}. The Haar integral is approximated via Monte Carlo with 2000 uniformly sampled unitary matrices.}
\setlength{\tabcolsep}{10pt} 
\renewcommand{\arraystretch}{1.3} 
\begin{tabular}{c|ccc|ccc}
\hline
 & \multicolumn{3}{c|}{Sequential} & \multicolumn{3}{c}{Parallel} \\
\cline{2-7}
$t$ & 1 & 2 & 3 & 1 & 2 & 3 \\
\hline
min values of \eqref{SDP-setting} & 0.7058 & 0.1894 & 0 & 0.7058 & 0.4707 & 0.3536 \\
\hline
\end{tabular}
\label{tab:seq-vs-paral}
\end{table*}



To further understand the mechanism behind the three-query threshold, we derive a complete characterization of all channels that can constitute exact shadow inversion under $Z$. This result identifies the precise operator structure available to any $t$-query protocol, while regardless of its internal comb structure.

\begin{restatable}{proposition}{mainprop}\label{form-des}

For any $t \in \mathbb{N}^+$, $\mathcal{N}$ is a $t$-query shadow inversion of $2$-dimensional unitaries under Pauli-$Z$ if and only if
\begin{equation}\label{eq:form-of-CPTP}
\begin{aligned}
    \mathcal{N}_{U}(\rho) = p(U) U^\dagger \rho U &+ (1 - p(U)) Z U^\dagger \rho U Z\\ &+ r(U) (U^\dagger \rho U Z - Z U^\dagger \rho U)
\end{aligned}
\end{equation}
holds for any $U \in \opr{U}(2)$ and density operator $\rho \in \mathsf{D}(\mathbb{C}^{2})$, where $p,r$ are functions of $U$ satisfying the following:

\begin{align*}
&0 \leq p(U) \leq 1,  \\
&\operatorname{Re}(r(U)) = 0,  \\
&|r(U)|^2 \leq p(U)(1 - p(U)). 
\end{align*}
\end{restatable}

Although a full analytical proof of the lower bound for qubit case remains open, Proposition \ref{form-des} offers a step in this direction and may ultimately lead to either a proof or a counterexample.

\section{Symmetry-based SDP simplification for general dimension}

The qubit case provides a complete and tractable example of shadow inversion. For higher dimensions, however, the size of the optimization problem increases rapidly.

but higher-dimensional systems exhibit significantly richer structure. In a $t$-slot comb the size of the Choi matrix $C$ in the optimization \eqref{SDP-setting} is $d^{2(t+1)} \times d^{2(t+1)}$, resulting in a total variable number of $d^{4t+4}$. Direct semidefinite optimization is therefore computationally infeasible for even moderate values of $d$ and $t$. To address this, we introduce a symmetry-based reduction that substantially decreases the number of free parameters.




We show that after an appropriate permutation of tensor factors, without loss of generality one can assume that Choi operators of feasible shadow inversion combs commute with the group action generated by $U^{\otimes(t+1)}$ and by actions of $O'\cap$ on the system registers:
\begin{equation}\label{symmerty-prop}
\begin{aligned}
    &[P_{\pi}CP_{\pi}, U^{\otimes t+1} \otimes V^{\otimes t} \otimes W] = 0\\
    &\forall U \in \opr{U}(d), V,W \in O'\cap\opr{U}(d),
\end{aligned}
\end{equation}
where $O'$ is the commutant of the observable $O$ and $P_{\pi}$ is some fixed permutation (see Corollary \ref{cor-sym} for details).

An important observation about \eqref{symmerty-prop} is that, in the traditional task of deterministic and exact unitary inversion \cite{Invert-general}, the Choi operator enjoys a general unitary covariance. In our exact shadow inversion task, by contrast, the Choi operator is not only generally unitary-covariant, but also exhibits an additional unitary covariance associated with the observable~$O$. This structure is similar to a related work that studies the task with a fixed input
state~$\rho$ \cite{murao2025higherorderquantumcomputingknown}: in that setting, the unitary covariance of the Choi operator likewise combines a general unitary covariance with an extra covariance determined by~$\rho$.

Since operators commuting with a group representation decompose into blocks corresponding to irreducible components, we can apply standard representation-theoretic techniques to construct a basis in which all constraints become block diagonal. Schur–Weyl duality then allow us to separate the Hilbert space into invariant subspaces labeled by irreducible representations of the intersection of $U(d)$ and the commutant $O'$. The Choi operator then decomposes into a direct sum of blocks, and each block carries a small number of independent degrees of freedom. The SDP constraints restrict these blocks independently, enabling the full optimization to be replaced by several smaller SDPs running in parallel.

This leads to a significant reduction in the combinatorial size of the optimization problem. Explicit counting of the block structure yields the following bound on the number of free variables.

\begin{restatable}{theorem}{upperboundvariables}\label{upper-bound-variables}
For any $d,t \in \mathbb{N}^{+}$, let $O$ be some fixed $d$-dimensional observable and $N_{O,t}$ be the number of variables in the simplified SDP for general $t$-query shadow inversion of $d$-dimensional unitary under $O$. Then for any $O$, it holds that
    \begin{equation*}
        N_{O,t} \leq (t+1)! \, t! \, d^{t+1}
    \end{equation*}
\end{restatable}

This reduction replaces the original exponential factor $d^{3t+3}$ with $(t+1)!t!$, a poly-factorial one in $t$. In practice, this structure enables numerical exploration of shadow inversion protocols at dimensionalities that would otherwise be far beyond reach. The full details of the reduction procedure and the construction of the symmetry-adapted basis are provided in Appendix \ref{subsec:simplification}.

\section{Concluding remarks}

The shadow inversion technique provides a direct operational method for estimating the quantity $\tr[U^\dagger\rho UO]$ using only oracle queries to the unknown unitary $U$. An alternative route, explored in earlier work \cite{mo2025parameterized}, is to virtually construct an inversion comb and then recover the same expectation value through classical post-processing of measurement statistics. This virtual strategy can reproduce many properties of an inverse transformation while avoiding the need to implement a physical reversal. 

We show how the statistical performance of these two approaches compares when both are used to estimate the same observable expectation. Specifically, for the qubit case, using a total of $N$ queries to the unknown unitary, an estimator derived from shadow inversion yields an unbiased estimate with variance $\frac{3}{N}\left(1-\tr\!\left[U^{\dagger}\rho UZ\right]^2\right)$. In contrast, the virtual-comb estimator exhibits variance at least $\frac{1}{N}\left(9-\tr\!\left[U^{\dagger}\rho UZ\right]^2\right)$, which is strictly larger for all inputs. This comparison indicates that shadow inversion  provides a more query-efficient route to observable estimation. We refer the detailed calculation to Appendix \ref{app:virtualcomb}.

A related but operationally different line of work, developed concurrently in Ref.~\cite{murao2025higherorderquantumcomputingknown}, studies a dual version of the shadow inversion task in which the goal is to reverse an unknown unitary with respect to a fixed input state, rather than a fixed observable. This scenario leads to a mathematically analogous contraction of the unitary action, but it differs in the operational setting and in the structure of the corresponding constraints. Their analysis also yields a three-query construction for qubit unitaries. Although both approaches involve partial information about the unitary and share similar structural features, a formal duality between the two formulations has not yet been established. Understanding whether these tasks can be related through an explicit dual transformation or a common variational framework remains an interesting direction for future work, motivated by the goal of fully understanding the fundamental limits of reversing quantum dynamics.

\section*{Acknowledgment}
This work was partially supported by the National Key R\&D Program of China (Grant No.~2024YFB4504004), the Guangdong Provincial Quantum Science Strategic Initiative (Grant No.~GDZX2403008, GDZX2403001), the Guangdong Provincial Key Lab of Integrated Communication, Sensing and Computation for Ubiquitous Internet of Things (Grant No.~2023B1212010007), the Quantum Science Center of Guangdong-Hong Kong-Macao Greater Bay Area, and the Education Bureau of Guangzhou Municipality.

\bibliography{ref}

\begin{thebibliography}{29}%
\makeatletter
\providecommand \@ifxundefined [1]{%
 \@ifx{#1\undefined}
}%
\providecommand \@ifnum [1]{%
 \ifnum #1\expandafter \@firstoftwo
 \else \expandafter \@secondoftwo
 \fi
}%
\providecommand \@ifx [1]{%
 \ifx #1\expandafter \@firstoftwo
 \else \expandafter \@secondoftwo
 \fi
}%
\providecommand \natexlab [1]{#1}%
\providecommand \enquote  [1]{``#1''}%
\providecommand \bibnamefont  [1]{#1}%
\providecommand \bibfnamefont [1]{#1}%
\providecommand \citenamefont [1]{#1}%
\providecommand \href@noop [0]{\@secondoftwo}%
\providecommand \href [0]{\begingroup \@sanitize@url \@href}%
\providecommand \@href[1]{\@@startlink{#1}\@@href}%
\providecommand \@@href[1]{\endgroup#1\@@endlink}%
\providecommand \@sanitize@url [0]{\catcode `\\12\catcode `\$12\catcode `\&12\catcode `\#12\catcode `\^12\catcode `\_12\catcode `\%12\relax}%
\providecommand \@@startlink[1]{}%
\providecommand \@@endlink[0]{}%
\providecommand \url  [0]{\begingroup\@sanitize@url \@url }%
\providecommand \@url [1]{\endgroup\@href {#1}{\urlprefix }}%
\providecommand \urlprefix  [0]{URL }%
\providecommand \Eprint [0]{\href }%
\providecommand \doibase [0]{https://doi.org/}%
\providecommand \selectlanguage [0]{\@gobble}%
\providecommand \bibinfo  [0]{\@secondoftwo}%
\providecommand \bibfield  [0]{\@secondoftwo}%
\providecommand \translation [1]{[#1]}%
\providecommand \BibitemOpen [0]{}%
\providecommand \bibitemStop [0]{}%
\providecommand \bibitemNoStop [0]{.\EOS\space}%
\providecommand \EOS [0]{\spacefactor3000\relax}%
\providecommand \BibitemShut  [1]{\csname bibitem#1\endcsname}%
\let\auto@bib@innerbib\@empty
\bibitem [{\citenamefont {Nielsen}\ and\ \citenamefont {Chuang}(2011)}]{Nielson-chuang}%
  \BibitemOpen
  \bibfield  {author} {\bibinfo {author} {\bibfnamefont {M.~A.}\ \bibnamefont {Nielsen}}\ and\ \bibinfo {author} {\bibfnamefont {I.~L.}\ \bibnamefont {Chuang}},\ }\href@noop {} {\emph {\bibinfo {title} {Quantum Computation and Quantum Information: 10th Anniversary Edition}}},\ \bibinfo {edition} {10th}\ ed.\ (\bibinfo  {publisher} {Cambridge University Press},\ \bibinfo {address} {USA},\ \bibinfo {year} {2011})\BibitemShut {NoStop}%
\bibitem [{\citenamefont {Quintino}\ \emph {et~al.}(2019{\natexlab{a}})\citenamefont {Quintino}, \citenamefont {Dong}, \citenamefont {Shimbo}, \citenamefont {Soeda},\ and\ \citenamefont {Murao}}]{Reverse-2019}%
  \BibitemOpen
  \bibfield  {author} {\bibinfo {author} {\bibfnamefont {M.~T.}\ \bibnamefont {Quintino}}, \bibinfo {author} {\bibfnamefont {Q.}~\bibnamefont {Dong}}, \bibinfo {author} {\bibfnamefont {A.}~\bibnamefont {Shimbo}}, \bibinfo {author} {\bibfnamefont {A.}~\bibnamefont {Soeda}},\ and\ \bibinfo {author} {\bibfnamefont {M.}~\bibnamefont {Murao}},\ }\bibfield  {title} {\bibinfo {title} {Reversing unknown quantum transformations: Universal quantum circuit for inverting general unitary operations},\ }\href {https://doi.org/10.1103/PhysRevLett.123.210502} {\bibfield  {journal} {\bibinfo  {journal} {Phys. Rev. Lett.}\ }\textbf {\bibinfo {volume} {123}},\ \bibinfo {pages} {210502} (\bibinfo {year} {2019}{\natexlab{a}})}\BibitemShut {NoStop}%
\bibitem [{\citenamefont {Tang}\ and\ \citenamefont {Wright}(2025)}]{Tang:2025owa}%
  \BibitemOpen
  \bibfield  {author} {\bibinfo {author} {\bibfnamefont {E.}~\bibnamefont {Tang}}\ and\ \bibinfo {author} {\bibfnamefont {J.}~\bibnamefont {Wright}},\ }\href@noop {} {\bibinfo {title} {{Amplitude amplification and estimation require inverses}}} (\bibinfo {year} {2025}),\ \Eprint {https://arxiv.org/abs/2507.23787} {arXiv:2507.23787 [quant-ph]} \BibitemShut {NoStop}%
\bibitem [{\citenamefont {Palao}\ and\ \citenamefont {Kosloff}(2002)}]{QC-2002}%
  \BibitemOpen
  \bibfield  {author} {\bibinfo {author} {\bibfnamefont {J.~P.}\ \bibnamefont {Palao}}\ and\ \bibinfo {author} {\bibfnamefont {R.}~\bibnamefont {Kosloff}},\ }\bibfield  {title} {\bibinfo {title} {Quantum computing by an optimal control algorithm for unitary transformations},\ }\href {https://doi.org/10.1103/PhysRevLett.89.188301} {\bibfield  {journal} {\bibinfo  {journal} {Phys. Rev. Lett.}\ }\textbf {\bibinfo {volume} {89}},\ \bibinfo {pages} {188301} (\bibinfo {year} {2002})}\BibitemShut {NoStop}%
\bibitem [{\citenamefont {Khaneja}\ and\ \citenamefont {Glaser}(2001)}]{Khaneja2001}%
  \BibitemOpen
  \bibfield  {author} {\bibinfo {author} {\bibfnamefont {N.}~\bibnamefont {Khaneja}}\ and\ \bibinfo {author} {\bibfnamefont {S.~J.}\ \bibnamefont {Glaser}},\ }\bibfield  {title} {\bibinfo {title} {Efficient transfer of coherence through ising spin chains},\ }\href@noop {} {\bibfield  {journal} {\bibinfo  {journal} {Physical Review A}\ }\textbf {\bibinfo {volume} {63}},\ \bibinfo {pages} {032308} (\bibinfo {year} {2001})}\BibitemShut {NoStop}%
\bibitem [{\citenamefont {Chen}\ \emph {et~al.}(2025{\natexlab{a}})\citenamefont {Chen}, \citenamefont {Mo}, \citenamefont {Liu}, \citenamefont {Zhang},\ and\ \citenamefont {Wang}}]{QURA}%
  \BibitemOpen
  \bibfield  {author} {\bibinfo {author} {\bibfnamefont {Y.-A.}\ \bibnamefont {Chen}}, \bibinfo {author} {\bibfnamefont {Y.}~\bibnamefont {Mo}}, \bibinfo {author} {\bibfnamefont {Y.}~\bibnamefont {Liu}}, \bibinfo {author} {\bibfnamefont {L.}~\bibnamefont {Zhang}},\ and\ \bibinfo {author} {\bibfnamefont {X.}~\bibnamefont {Wang}},\ }\href {https://arxiv.org/abs/2403.04704} {\bibinfo {title} {Quantum algorithm for reversing unknown unitary evolutions}} (\bibinfo {year} {2025}{\natexlab{a}}),\ \Eprint {https://arxiv.org/abs/2403.04704} {arXiv:2403.04704 [quant-ph]} \BibitemShut {NoStop}%
\bibitem [{\citenamefont {Odake}\ \emph {et~al.}(2025)\citenamefont {Odake}, \citenamefont {Yoshida},\ and\ \citenamefont {Murao}}]{Odake2024analytical}%
  \BibitemOpen
  \bibfield  {author} {\bibinfo {author} {\bibfnamefont {T.}~\bibnamefont {Odake}}, \bibinfo {author} {\bibfnamefont {S.}~\bibnamefont {Yoshida}},\ and\ \bibinfo {author} {\bibfnamefont {M.}~\bibnamefont {Murao}},\ }\bibfield  {title} {\bibinfo {title} {Analytical lower bound on query complexity for transformations of unknown unitary operations},\ }\href {https://doi.org/10.1103/drp2-rzzw} {\bibfield  {journal} {\bibinfo  {journal} {Phys. Rev. Lett.}\ }\textbf {\bibinfo {volume} {135}},\ \bibinfo {pages} {230603} (\bibinfo {year} {2025})}\BibitemShut {NoStop}%
\bibitem [{\citenamefont {Chen}\ \emph {et~al.}(2025{\natexlab{b}})\citenamefont {Chen}, \citenamefont {Yu},\ and\ \citenamefont {Zhang}}]{Chen2025tight}%
  \BibitemOpen
  \bibfield  {author} {\bibinfo {author} {\bibfnamefont {K.}~\bibnamefont {Chen}}, \bibinfo {author} {\bibfnamefont {N.}~\bibnamefont {Yu}},\ and\ \bibinfo {author} {\bibfnamefont {Z.}~\bibnamefont {Zhang}},\ }\bibfield  {title} {\bibinfo {title} {Tight bound for quantum unitary time-reversal},\ }\href@noop {} {\bibfield  {journal} {\bibinfo  {journal} {arXiv preprint arXiv:2507.05736}\ } (\bibinfo {year} {2025}{\natexlab{b}})}\BibitemShut {NoStop}%
\bibitem [{\citenamefont {Yoshida}\ \emph {et~al.}(2023)\citenamefont {Yoshida}, \citenamefont {Soeda},\ and\ \citenamefont {Murao}}]{Rever-unknown}%
  \BibitemOpen
  \bibfield  {author} {\bibinfo {author} {\bibfnamefont {S.}~\bibnamefont {Yoshida}}, \bibinfo {author} {\bibfnamefont {A.}~\bibnamefont {Soeda}},\ and\ \bibinfo {author} {\bibfnamefont {M.}~\bibnamefont {Murao}},\ }\bibfield  {title} {\bibinfo {title} {Reversing unknown qubit-unitary operation, deterministically and exactly},\ }\href {https://doi.org/10.1103/PhysRevLett.131.120602} {\bibfield  {journal} {\bibinfo  {journal} {Phys. Rev. Lett.}\ }\textbf {\bibinfo {volume} {131}},\ \bibinfo {pages} {120602} (\bibinfo {year} {2023})}\BibitemShut {NoStop}%
\bibitem [{\citenamefont {Mo}\ \emph {et~al.}(2025{\natexlab{a}})\citenamefont {Mo}, \citenamefont {Zhang}, \citenamefont {Chen}, \citenamefont {Liu}, \citenamefont {Lin},\ and\ \citenamefont {Wang}}]{mo2025parameterized}%
  \BibitemOpen
  \bibfield  {author} {\bibinfo {author} {\bibfnamefont {Y.}~\bibnamefont {Mo}}, \bibinfo {author} {\bibfnamefont {L.}~\bibnamefont {Zhang}}, \bibinfo {author} {\bibfnamefont {Y.-A.}\ \bibnamefont {Chen}}, \bibinfo {author} {\bibfnamefont {Y.}~\bibnamefont {Liu}}, \bibinfo {author} {\bibfnamefont {T.}~\bibnamefont {Lin}},\ and\ \bibinfo {author} {\bibfnamefont {X.}~\bibnamefont {Wang}},\ }\bibfield  {title} {\bibinfo {title} {Parameterized quantum comb and simpler circuits for reversing unknown qubit-unitary operations},\ }\href {https://doi.org/10.1038/s41534-025-00979-1} {\bibfield  {journal} {\bibinfo  {journal} {npj Quantum Information}\ }\textbf {\bibinfo {volume} {11}},\ \bibinfo {pages} {1} (\bibinfo {year} {2025}{\natexlab{a}})}\BibitemShut {NoStop}%
\bibitem [{\citenamefont {Zhu}\ \emph {et~al.}(2024)\citenamefont {Zhu}, \citenamefont {Mo}, \citenamefont {Chen},\ and\ \citenamefont {Wang}}]{zhu2024reversing}%
  \BibitemOpen
  \bibfield  {author} {\bibinfo {author} {\bibfnamefont {C.}~\bibnamefont {Zhu}}, \bibinfo {author} {\bibfnamefont {Y.}~\bibnamefont {Mo}}, \bibinfo {author} {\bibfnamefont {Y.-A.}\ \bibnamefont {Chen}},\ and\ \bibinfo {author} {\bibfnamefont {X.}~\bibnamefont {Wang}},\ }\bibfield  {title} {\bibinfo {title} {Reversing unknown quantum processes via virtual combs for channels with limited information},\ }\href {https://doi.org/10.1103/PhysRevLett.133.030801} {\bibfield  {journal} {\bibinfo  {journal} {Physical Review Letters}\ }\textbf {\bibinfo {volume} {133}},\ \bibinfo {pages} {030801} (\bibinfo {year} {2024})}\BibitemShut {NoStop}%
\bibitem [{\citenamefont {Sedl\'ak}\ \emph {et~al.}(2019)\citenamefont {Sedl\'ak}, \citenamefont {Bisio},\ and\ \citenamefont {Ziman}}]{Storage-Retrieval}%
  \BibitemOpen
  \bibfield  {author} {\bibinfo {author} {\bibfnamefont {M.}~\bibnamefont {Sedl\'ak}}, \bibinfo {author} {\bibfnamefont {A.}~\bibnamefont {Bisio}},\ and\ \bibinfo {author} {\bibfnamefont {M.}~\bibnamefont {Ziman}},\ }\bibfield  {title} {\bibinfo {title} {Optimal probabilistic storage and retrieval of unitary channels},\ }\href {https://doi.org/10.1103/PhysRevLett.122.170502} {\bibfield  {journal} {\bibinfo  {journal} {Phys. Rev. Lett.}\ }\textbf {\bibinfo {volume} {122}},\ \bibinfo {pages} {170502} (\bibinfo {year} {2019})}\BibitemShut {NoStop}%
\bibitem [{\citenamefont {Quintino}\ \emph {et~al.}(2019{\natexlab{b}})\citenamefont {Quintino}, \citenamefont {Dong}, \citenamefont {Shimbo}, \citenamefont {Soeda},\ and\ \citenamefont {Murao}}]{Prob-exact}%
  \BibitemOpen
  \bibfield  {author} {\bibinfo {author} {\bibfnamefont {M.~T.}\ \bibnamefont {Quintino}}, \bibinfo {author} {\bibfnamefont {Q.}~\bibnamefont {Dong}}, \bibinfo {author} {\bibfnamefont {A.}~\bibnamefont {Shimbo}}, \bibinfo {author} {\bibfnamefont {A.}~\bibnamefont {Soeda}},\ and\ \bibinfo {author} {\bibfnamefont {M.}~\bibnamefont {Murao}},\ }\bibfield  {title} {\bibinfo {title} {Probabilistic exact universal quantum circuits for transforming unitary operations},\ }\href {https://doi.org/10.1103/PhysRevA.100.062339} {\bibfield  {journal} {\bibinfo  {journal} {Phys. Rev. A}\ }\textbf {\bibinfo {volume} {100}},\ \bibinfo {pages} {062339} (\bibinfo {year} {2019}{\natexlab{b}})}\BibitemShut {NoStop}%
\bibitem [{\citenamefont {Quintino}\ \emph {et~al.}(2019{\natexlab{c}})\citenamefont {Quintino}, \citenamefont {Dong}, \citenamefont {Shimbo}, \citenamefont {Soeda},\ and\ \citenamefont {Murao}}]{Invert-general}%
  \BibitemOpen
  \bibfield  {author} {\bibinfo {author} {\bibfnamefont {M.~T.}\ \bibnamefont {Quintino}}, \bibinfo {author} {\bibfnamefont {Q.}~\bibnamefont {Dong}}, \bibinfo {author} {\bibfnamefont {A.}~\bibnamefont {Shimbo}}, \bibinfo {author} {\bibfnamefont {A.}~\bibnamefont {Soeda}},\ and\ \bibinfo {author} {\bibfnamefont {M.}~\bibnamefont {Murao}},\ }\bibfield  {title} {\bibinfo {title} {Reversing unknown quantum transformations: Universal quantum circuit for inverting general unitary operations},\ }\href {https://doi.org/10.1103/PhysRevLett.123.210502} {\bibfield  {journal} {\bibinfo  {journal} {Phys. Rev. Lett.}\ }\textbf {\bibinfo {volume} {123}},\ \bibinfo {pages} {210502} (\bibinfo {year} {2019}{\natexlab{c}})}\BibitemShut {NoStop}%
\bibitem [{\citenamefont {Haah}\ \emph {et~al.}(2023)\citenamefont {Haah}, \citenamefont {Kothari}, \citenamefont {O’Donnell},\ and\ \citenamefont {Tang}}]{Haah2023query}%
  \BibitemOpen
  \bibfield  {author} {\bibinfo {author} {\bibfnamefont {J.}~\bibnamefont {Haah}}, \bibinfo {author} {\bibfnamefont {R.}~\bibnamefont {Kothari}}, \bibinfo {author} {\bibfnamefont {R.}~\bibnamefont {O’Donnell}},\ and\ \bibinfo {author} {\bibfnamefont {E.}~\bibnamefont {Tang}},\ }\bibfield  {title} {\bibinfo {title} {Query-optimal estimation of unitary channels in diamond distance},\ }in\ \href@noop {} {\emph {\bibinfo {booktitle} {2023 IEEE 64th Annual Symposium on Foundations of Computer Science (FOCS)}}}\ (\bibinfo {organization} {IEEE},\ \bibinfo {year} {2023})\ pp.\ \bibinfo {pages} {363--390}\BibitemShut {NoStop}%
\bibitem [{\citenamefont {Mo}\ \emph {et~al.}(2025{\natexlab{b}})\citenamefont {Mo}, \citenamefont {Lin},\ and\ \citenamefont {Wang}}]{mo2025}%
  \BibitemOpen
  \bibfield  {author} {\bibinfo {author} {\bibfnamefont {Y.}~\bibnamefont {Mo}}, \bibinfo {author} {\bibfnamefont {T.}~\bibnamefont {Lin}},\ and\ \bibinfo {author} {\bibfnamefont {X.}~\bibnamefont {Wang}},\ }\href {https://arxiv.org/abs/2506.20570} {\bibinfo {title} {Efficient inversion of unknown unitary operations with structured hamiltonians}} (\bibinfo {year} {2025}{\natexlab{b}}),\ \Eprint {https://arxiv.org/abs/2506.20570} {arXiv:2506.20570 [quant-ph]} \BibitemShut {NoStop}%
\bibitem [{\citenamefont {Aaronson}(2020)}]{Shadow-Tomo}%
  \BibitemOpen
  \bibfield  {author} {\bibinfo {author} {\bibfnamefont {S.}~\bibnamefont {Aaronson}},\ }\bibfield  {title} {\bibinfo {title} {Shadow tomography of quantum states},\ }\href {https://doi.org/10.1137/18M120275X} {\bibfield  {journal} {\bibinfo  {journal} {SIAM Journal on Computing}\ }\textbf {\bibinfo {volume} {49}},\ \bibinfo {pages} {STOC18} (\bibinfo {year} {2020})},\ \Eprint {https://arxiv.org/abs/https://doi.org/10.1137/18M120275X} {https://doi.org/10.1137/18M120275X} \BibitemShut {NoStop}%
\bibitem [{\citenamefont {Huang}\ \emph {et~al.}(2020)\citenamefont {Huang}, \citenamefont {Kueng},\ and\ \citenamefont {Preskill}}]{Huang2020}%
  \BibitemOpen
  \bibfield  {author} {\bibinfo {author} {\bibfnamefont {H.-Y.}\ \bibnamefont {Huang}}, \bibinfo {author} {\bibfnamefont {R.}~\bibnamefont {Kueng}},\ and\ \bibinfo {author} {\bibfnamefont {J.}~\bibnamefont {Preskill}},\ }\bibfield  {title} {\bibinfo {title} {Predicting many properties of a quantum system from very few measurements},\ }\href {https://doi.org/10.1038/s41567-020-0932-7} {\bibfield  {journal} {\bibinfo  {journal} {Nature Physics}\ }\textbf {\bibinfo {volume} {16}},\ \bibinfo {pages} {1050} (\bibinfo {year} {2020})}\BibitemShut {NoStop}%
\bibitem [{\citenamefont {Zhao}\ \emph {et~al.}(2023)\citenamefont {Zhao}, \citenamefont {Zhao}, \citenamefont {Xia},\ and\ \citenamefont {Wang}}]{recover-shadow}%
  \BibitemOpen
  \bibfield  {author} {\bibinfo {author} {\bibfnamefont {X.}~\bibnamefont {Zhao}}, \bibinfo {author} {\bibfnamefont {B.}~\bibnamefont {Zhao}}, \bibinfo {author} {\bibfnamefont {Z.}~\bibnamefont {Xia}},\ and\ \bibinfo {author} {\bibfnamefont {X.}~\bibnamefont {Wang}},\ }\bibfield  {title} {\bibinfo {title} {Information recoverability of noisy quantum states},\ }\href@noop {} {\bibfield  {journal} {\bibinfo  {journal} {Quantum}\ }\textbf {\bibinfo {volume} {7}},\ \bibinfo {pages} {978} (\bibinfo {year} {2023})}\BibitemShut {NoStop}%
\bibitem [{\citenamefont {Chiribella}\ \emph {et~al.}(2008)\citenamefont {Chiribella}, \citenamefont {D'Ariano},\ and\ \citenamefont {Perinotti}}]{Chiribella2008quantum}%
  \BibitemOpen
  \bibfield  {author} {\bibinfo {author} {\bibfnamefont {G.}~\bibnamefont {Chiribella}}, \bibinfo {author} {\bibfnamefont {G.~M.}\ \bibnamefont {D'Ariano}},\ and\ \bibinfo {author} {\bibfnamefont {P.}~\bibnamefont {Perinotti}},\ }\bibfield  {title} {\bibinfo {title} {Quantum circuit architecture},\ }\href {https://doi.org/10.1103/PhysRevLett.101.060401} {\bibfield  {journal} {\bibinfo  {journal} {Phys. Rev. Lett.}\ }\textbf {\bibinfo {volume} {101}},\ \bibinfo {pages} {060401} (\bibinfo {year} {2008})}\BibitemShut {NoStop}%
\bibitem [{\citenamefont {Zhen}(2025)}]{ZhenGuocheng_shadow_unitary_inversion_numerics}%
  \BibitemOpen
  \bibfield  {author} {\bibinfo {author} {\bibfnamefont {G.}~\bibnamefont {Zhen}},\ }\href {https://github.com/ZhenGuocheng/shadow-unitary-inversion-numerics} {\bibinfo {title} {shadow-unitary-inversion-numerics: Reproducibility code for numerical experiments (table i) in ``structure, optimality, and symmetry in shadow unitary inversion''}} (\bibinfo {year} {2025}),\ \bibinfo {note} {gitHub repository (MIT License)}\BibitemShut {NoStop}%
\bibitem [{\citenamefont {Brzić}\ \emph {et~al.}(2025)\citenamefont {Brzić}, \citenamefont {Yoshida}, \citenamefont {Murao},\ and\ \citenamefont {Quintino}}]{murao2025higherorderquantumcomputingknown}%
  \BibitemOpen
  \bibfield  {author} {\bibinfo {author} {\bibfnamefont {V.}~\bibnamefont {Brzić}}, \bibinfo {author} {\bibfnamefont {S.}~\bibnamefont {Yoshida}}, \bibinfo {author} {\bibfnamefont {M.}~\bibnamefont {Murao}},\ and\ \bibinfo {author} {\bibfnamefont {M.~T.}\ \bibnamefont {Quintino}},\ }\href {https://arxiv.org/abs/2510.20530} {\bibinfo {title} {Higher-order quantum computing with known input states}} (\bibinfo {year} {2025}),\ \Eprint {https://arxiv.org/abs/2510.20530} {arXiv:2510.20530 [quant-ph]} \BibitemShut {NoStop}%
\bibitem [{\citenamefont {Durrett}(2019)}]{Durrett_2019}%
  \BibitemOpen
  \bibfield  {author} {\bibinfo {author} {\bibfnamefont {R.}~\bibnamefont {Durrett}},\ }\href@noop {} {\emph {\bibinfo {title} {Probability: Theory and Examples}}},\ \bibinfo {edition} {5th}\ ed.,\ Cambridge Series in Statistical and Probabilistic Mathematics\ (\bibinfo  {publisher} {Cambridge University Press},\ \bibinfo {year} {2019})\BibitemShut {NoStop}%
\bibitem [{\citenamefont {M{\'e}liot}(2017)}]{Meliot2017}%
  \BibitemOpen
  \bibfield  {author} {\bibinfo {author} {\bibfnamefont {P.-L.}\ \bibnamefont {M{\'e}liot}},\ }\href {https://doi.org/10.1201/9781315371016} {\emph {\bibinfo {title} {Representation Theory of Symmetric Groups}}},\ \bibinfo {edition} {1st}\ ed.\ (\bibinfo  {publisher} {Chapman and Hall/CRC},\ \bibinfo {address} {New York},\ \bibinfo {year} {2017})\ p.\ \bibinfo {pages} {682}\BibitemShut {NoStop}%
\bibitem [{\citenamefont {Fulton}\ and\ \citenamefont {Harris}(1991)}]{FultonHarris1991}%
  \BibitemOpen
  \bibfield  {author} {\bibinfo {author} {\bibfnamefont {W.}~\bibnamefont {Fulton}}\ and\ \bibinfo {author} {\bibfnamefont {J.}~\bibnamefont {Harris}},\ }\href {https://doi.org/10.1007/978-1-4612-0979-9} {\emph {\bibinfo {title} {Representation Theory: A First Course}}},\ \bibinfo {edition} {1st}\ ed.,\ \bibinfo {series} {Graduate Texts in Mathematics}, Vol.\ \bibinfo {volume} {129}\ (\bibinfo  {publisher} {Springer},\ \bibinfo {address} {New York, NY},\ \bibinfo {year} {1991})\ p.\ \bibinfo {pages} {551}\BibitemShut {NoStop}%
\bibitem [{\citenamefont {Chiribella}\ \emph {et~al.}(2009)\citenamefont {Chiribella}, \citenamefont {D'Ariano},\ and\ \citenamefont {Perinotti}}]{Theo-frame}%
  \BibitemOpen
  \bibfield  {author} {\bibinfo {author} {\bibfnamefont {G.}~\bibnamefont {Chiribella}}, \bibinfo {author} {\bibfnamefont {G.~M.}\ \bibnamefont {D'Ariano}},\ and\ \bibinfo {author} {\bibfnamefont {P.}~\bibnamefont {Perinotti}},\ }\bibfield  {title} {\bibinfo {title} {Theoretical framework for quantum networks},\ }\href {https://doi.org/10.1103/PhysRevA.80.022339} {\bibfield  {journal} {\bibinfo  {journal} {Phys. Rev. A}\ }\textbf {\bibinfo {volume} {80}},\ \bibinfo {pages} {022339} (\bibinfo {year} {2009})}\BibitemShut {NoStop}%
\bibitem [{\citenamefont {Chen}\ \emph {et~al.}(2023)\citenamefont {Chen}, \citenamefont {Wang}, \citenamefont {Long},\ and\ \citenamefont {Ying}}]{Unitarity-estimation}%
  \BibitemOpen
  \bibfield  {author} {\bibinfo {author} {\bibfnamefont {K.}~\bibnamefont {Chen}}, \bibinfo {author} {\bibfnamefont {Q.}~\bibnamefont {Wang}}, \bibinfo {author} {\bibfnamefont {P.}~\bibnamefont {Long}},\ and\ \bibinfo {author} {\bibfnamefont {M.}~\bibnamefont {Ying}},\ }\bibfield  {title} {\bibinfo {title} {Unitarity estimation for quantum channels},\ }\href {https://doi.org/10.1109/TIT.2023.3263645} {\bibfield  {journal} {\bibinfo  {journal} {IEEE Transactions on Information Theory}\ }\textbf {\bibinfo {volume} {69}},\ \bibinfo {pages} {5116} (\bibinfo {year} {2023})}\BibitemShut {NoStop}%
\bibitem [{\citenamefont {Quintino}\ and\ \citenamefont {Ebler}(2022)}]{Per-ope}%
  \BibitemOpen
  \bibfield  {author} {\bibinfo {author} {\bibfnamefont {M.~T.}\ \bibnamefont {Quintino}}\ and\ \bibinfo {author} {\bibfnamefont {D.}~\bibnamefont {Ebler}},\ }\bibfield  {title} {\bibinfo {title} {Deterministic transformations between unitary operations: {E}xponential advantage with adaptive quantum circuits and the power of indefinite causality},\ }\href {https://doi.org/10.22331/q-2022-03-31-679} {\bibfield  {journal} {\bibinfo  {journal} {{Quantum}}\ }\textbf {\bibinfo {volume} {6}},\ \bibinfo {pages} {679} (\bibinfo {year} {2022})}\BibitemShut {NoStop}%
\bibitem [{\citenamefont {Kowalski}(2025)}]{kowalski2025introduction}%
  \BibitemOpen
  \bibfield  {author} {\bibinfo {author} {\bibfnamefont {E.}~\bibnamefont {Kowalski}},\ }\href {https://books.google.com/books?id=vdJp0QEACAAJ} {\emph {\bibinfo {title} {An Introduction to the Representation Theory of Groups}}},\ Graduate Studies in Mathematics Series\ (\bibinfo  {publisher} {American Mathematical Society},\ \bibinfo {year} {2025})\BibitemShut {NoStop}%
\end{thebibliography}%

\appendix
\newpage

\vspace{2cm}
\onecolumngrid
\vspace{2cm}

\renewcommand{\theequation}{S\arabic{equation}}
\renewcommand{\theproposition}{S\arabic{proposition}}
\renewcommand{\thedefinition}{S\arabic{definition}}
\renewcommand{\thelemma}{S\arabic{lemma}}
\renewcommand{\thetheorem}{S\arabic{theorem}}
\renewcommand{\thefigure}{S\arabic{figure}}
\setcounter{equation}{0}
\setcounter{table}{0}
\setcounter{section}{0}
\setcounter{proposition}{0}
\setcounter{lemma}{0}
\setcounter{theorem}{0}
\setcounter{definition}{0}
\setcounter{figure}{0}
\setcounter{subsection}{0}

\begin{center}
\large{\textbf{Supplemental Material for
`Structure, Optimality, and Symmetry in Shadow Unitary Inversion'}}
\end{center}

\section{Measure theory and representation theory preliminaries}\label{sec:preliminary}

\subsection*{Measure Theory}
In the study of quantum information theory, measure-theoretic tools play a central role. In particular, when dealing with objects that carry a group structure, one often seeks to define a ``uniform'' distribution that captures the underlying symmetry. Haar measure provides precisely such a natural 
framework. Let us recall its properties.
\begin{definition}[Haar measure]\label{haar-basic}
Let $G$ be a locally compact Hausdorff topological group. The Haar measure on $G$ is the unique regular Borel probability measure $\mu_{H}$ on $G$ such that
\begin{equation*}
\mu_{H}(gE) \;=\; \mu_{H}(E) \;=\; \mu_{H}(Eg)
\end{equation*}
for every Borel set $E\subseteq G$ and every $g\in G$.
\end{definition}
Moreover, in the more general setting of measure spaces, one requires classical integration theorems to handle integrability on product spaces and the exchange of the order of integration. In this context, Tonelli’s theorem (together with Fubini’s theorem) becomes indispensable.

Let $(X,\mathcal A,\mu)$ be a measure space, we call it a $\sigma$-finite measure space if the set $X$ can be covered with at most countably many measurable sets with finite measure, that is there are sets $A_n \in \mathcal{A}$ such that
\begin{equation*}
    \bigcup_{n \in \mathbb{N}} A_{n} = X, \quad \mu(A_n) < \infty \,\, \text{for all $n \in \mathbb{N}$}. 
\end{equation*}
Then we will introduce a very useful lemma which is a successor of Fubini's theorem \cite{Durrett_2019}.
\begin{lemma}[Tonelli’s theorem]
Let $(X,\mathcal A,\mu)$ and $(Y,\mathcal B,\nu)$ be $\sigma$-finite measure spaces, and let 
\begin{equation*}
f \colon X \times Y \;\longrightarrow\; [0,\infty]
\end{equation*}
be a non-negative measurable function.  Then
\begin{equation*}
    \begin{split}
    \int_X \Bigl(\int_Y f(x,y)\,d\nu(y)\Bigr)\,d\mu(x)
    \;&=\;
    \int_Y \Bigl(\int_X f(x,y)\,d\mu(x)\Bigr)\,d\nu(y)\\
    &=\;
    \int_{X\times Y} f(x,y)\,d(\mu\times\nu)(x,y)
    \end{split}
\end{equation*}
where $\mu\times\nu$ represents the product measure on $X \times Y$.
\end{lemma}

\subsection*{Group Representation Theory}\label{g-r-t}
 In this section, we will introduce some results about group representation theory which will be useful for our following discussion. We begin with recalling the definition of finite group representation and irreducible representation.
 
 \begin{definition}[Finite group representation]\label{finite group representation}
    Let \( (G, \cdot) \) be a finite group and $V$ be a finite dimensional vector space over field $\mathbb{F}$. A linear representation of $G$ is a group homomorphism $\rho: G \to \operatorname{GL}(V) = \operatorname{Aut}(V)$ where $\operatorname{GL}(V)$ is the general linear group of $V$ and $\operatorname{Aut}(V)$ is the automorphism group.
    
 \end{definition}
 
 \begin{definition}[Irreducible representation]\label{irreducible representation}
    A representation \( R: G \to \operatorname{GL}(V) \) is called \textbf{irreducible} if there is no non-zero subspace \( W \subsetneq V \) such that \( R(g)W \subseteq W \) for all \( g \in G \). That is, the representation has no non-trivial invariant subspaces.
 \end{definition}

 We now turn to the two representations relevant to the Schur-Weyl duality. Recall that the symmetric
 group $S_n$ of degree $n$ is the group of all permutations of $n$ objects. Then we have the following natural representation of the symmetric group on the space $(\mathbb{C}^d)^{\otimes n}$:
 \begin{equation*}
    P_{\pi}\,|i_1\rangle \otimes |i_2\rangle \otimes \cdots \otimes |i_n\rangle
    := |i_{\pi^{-1}(1)}\rangle \otimes |i_{\pi^{-1}(2)}\rangle \otimes \cdots \otimes |i_{\pi^{-1}(n)}\rangle,
    \label{eq:symmetric_representation}
 \end{equation*}
 where $\pi \in S_n$ is a permutation and $\pi(i)$ is the label describing the action of $\pi$ on label $i$. For example, if we are considering $S_3$ and the permutation $\pi = (12)(3)$, then
 \begin{equation*}
 P_{\pi}\,|i_1, i_2, i_3\rangle = |i_2, i_1, i_3\rangle.
 \end{equation*}
 Next we turn to the representation of the unitary group. Let $\opr{U}(d)$ denote the group of $d \times d$ unitary operators. Then there is also a natural representation of $U_d$ on the space $(\mathbb{C}^d)^{\otimes n}$ given by
 \begin{equation*}
    U^{\otimes n}\,|i_1\rangle \otimes |i_2\rangle \otimes \cdots \otimes |i_n\rangle
    := U|i_1\rangle \otimes U|i_2\rangle \otimes \cdots \otimes U|i_n\rangle
 \end{equation*}
 for any $U \in U_d$. In representation theory, combinatorial structures such as Young diagrams and their associated tableaux serve as fundamental tools for describing and analyzing representations of symmetric and general linear groups. Hence let us introduce their definitions and properties.
 \begin{definition}[Partition of a natural number]
    Let $n \in \mathbb{N}$, and let $\lambda = (\lambda_{1},\cdots,\lambda_{k})$ be such that
    \begin{equation*}
        \sum_{i=1}^{k} \lambda_{i} = n, \quad and \quad \lambda_{i} \geq \lambda_{i+1} \quad \text{for every}\,\, i=1,\cdots,k-1.
    \end{equation*}
    Then, $\lambda$ is called a partition of $n$, and we write $\lambda \vdash n$.
 \end{definition}
 \noindent Based on this, we will give the definition of Young diagram as following:
 \begin{definition}[Young diagram]\label{Young diagram}
    Let $n \in \mathbb{N}$ and let $\lambda = (\lambda_{1},\cdots,\lambda_{k})$ be a partition of $n$. The Young diagram $Y_{\lambda}$ with size $n$ corresponding to $\lambda$ is a planar arrangement of $n$ boxes that are left-aligned and top-aligned, such that the $i$-th row of $Y_{\lambda}$ contains exactly $\lambda_{i}$ empty boxes. 
 \end{definition}
 For example, the Young diagrams with size $4$ corresponding to the partitions $(4),(3,1),(2,2)$ respectively are 
\begin{equation*}
 \begin{array}{|c|c|c|c|}
    \hline
    &  &  &  \\
    \hline
 \end{array}
 \,\, , \,\,
 \begin{array}{|c|c|c|}
    \hline
    &  &  \\
    \hline
    & \multicolumn{2}{c}{} \\
    \cline{1-1}
 \end{array}
 \,\, , \,\,
 \begin{array}{|c|c|}
    \hline
    &  \\
    \hline
    &  \\
    \cline{1-1}\cline{2-2}
 \end{array}.
\end{equation*}
We denote $Y_d^n$ the set of Young diagrams with size $n$ and no more than $d$ rows and we will introduce the Schur-Weyl duality theory \cite{Meliot2017,FultonHarris1991} which will play an important role in the following discussion.
\begin{lemma}[Schur-Weyl duality]\label{Schur-Weyl-duality}
Let $U^{\otimes n}$ and $P_{\pi}$ be the representations of group $U_d$ and $S_n$, respectively. Then, we have the following decomposition:
    \begin{align*}
        (\mathbb{C}^d)^{\otimes n} &\cong \bigoplus_{\lambda\in Y_d^n} \mathcal{U}_{\lambda} \otimes \mathcal{S}_{\lambda},\\
        P_{\pi} &\cong\bigoplus_{\lambda\in Y_d^n} I_{\mathcal{U}_\lambda} \otimes \pi_{\lambda},\\
        U^{\otimes n} &\cong\bigoplus_{\lambda\in Y_d^n} U_\lambda \otimes I_{\mathcal{S}_{\lambda}},
    \end{align*}
    where $U_{\lambda}$ and $\pi_{\lambda}$ are irreducible representations of $U^{\otimes n}$ and $P_{\pi}$ labeled by $\lambda$, respectively. Moreover, $\mathcal{U}_{\lambda}$ and $\mathcal{S}_{\lambda}$ are respective representation spaces.
\end{lemma}
\noindent It is important for us to know the transform that project $U^{\otimes n}$ onto its irreducible representations. In order for this, we need to introduce the following combinatorial tools.
 \begin{definition}[Standard Young Tableau]\label{Standard Young Tableau}
    Let $Y_{\lambda}$ be a particular Young diagram of size $n$. A Standard Young Tableau (SYT) of shape $Y_{\lambda}$ is the diagram $Y_{\lambda}$ where each box is filled with a unique number in $[n] = \{1,\cdots,n\}$ and each number occurring once such that the numbers increase from left to right and from top to bottom in each row and column.
 \end{definition}
\noindent For example, there are three kinds of SYTs of shape $Y_{(3,1)}$:
\begin{equation*}
    \begin{array}{|c|c|c|}
        \hline
        1 & 3 & 4 \\
        \hline
        2 \\
        \cline{1-1}
    \end{array}
    \,\, , \,\,
    \begin{array}{|c|c|c|}
        \hline
        1 & 2 & 4 \\
        \hline
        3 \\
        \cline{1-1}
    \end{array}
    \,\, , \,\,
    \begin{array}{|c|c|c|}
        \hline
        1 & 2 & 3 \\
        \hline
        4 \\
        \cline{1-1}
    \end{array}.
\end{equation*}
 
 \begin{definition}[Semi-standard Young Tableau]\label{Semi-standard Young Tableau}
   Let $Y_{\lambda}$ be a particular Young diagram of size $n$. A Semi-standard Young Tableau (SSYT) of shape $Y_{\lambda}$ with filling $d$ is the diagram $Y_{\lambda}$ where each box is filled with a unique number in $[d] = \{1,\cdots,d\}$ such that the numbers non-decrease from left to right and strictly increase from top to bottom in each row and column. 
 \end{definition}
 \noindent For example, there are two kinds of SSYTs of shape $Y_{(2,1)}$ with filling 2:
\begin{equation*}
    \begin{array}{|c|c|}
        \hline
        1 & 1 \\
        \hline
        2 \\
        \cline{1-1}
    \end{array}
    \,\, , \,\,
    \begin{array}{|c|c|}
        \hline
        1 & 2 \\
        \hline
        2 \\
        \cline{1-1}
    \end{array}.
\end{equation*}
 
Now it is natural to ask the number of SYTs and SSYTs with the shape of $Y_{\lambda}$ and hence we will introduce the hook length formula.
\begin{definition}[Hook length]\label{hook length}
Let $Y_{\lambda}$ be a particular Young diagram and fill each empty box with one more than the total number of boxes lying to the right and underneath it, we will denote the number of the box at position $(i,j)$ as $h_{\lambda}(i,j)$. Then the hook length of $Y_{\lambda}$ which is denoted by $H_{Y_{\lambda}}$ is given by the product of all numbers appearing in the resulting tableau, that is 
\begin{equation*}
H_{Y_{\lambda}} = \prod_{i,j} h_{\lambda}(i,j).
\end{equation*}
 \end{definition}
\noindent For example, let us consider the Young diagram with size $8$ corresponding to the partition \( \lambda = (4, 3, 1) \):
    \begin{equation*}
    \begin{array}{|c|c|c|c|}
        \hline
        &  &  & \\
        \hline
        &  &  \\
        \cline{1-3}
        & \multicolumn{2}{c}{} \\
        \cline{1-1}
    \end{array}
    \quad \longrightarrow \quad
    \begin{array}{|c|c|c|c|}
        \hline
        6 & 4 & 3 & 1\\
        \hline
        4 & 2 & 1 \\
        \cline{1-3}
        1 &\multicolumn{2}{c}{} \\
        \cline{1-1}
    \end{array}.
    \end{equation*}
Then the hook length of $Y_{(4,3,1)}$ is
\begin{equation*}
    H_{Y_{(4,3,1)}} = \prod_{i,j} h_{(4,3,1)}(i,j) = 6 \times 4^2 \times 3 \times 2 \times 1^3 = 576.
\end{equation*}
We now establish the relationship between the dimensions of $\mathcal{U}_{\lambda}$ and $\mathcal{S}_{\lambda}$ and the hook lengths of $Y_{\lambda}$.

\begin{lemma}[Dimension and multiplicity of irreducible representation]\label{D-M}
 Let $U^{\otimes n}$ and $P_{\pi}$ be the representations of group $U_d$ and $S_n$, respectively. Then the dimension of the irreducible representation $\pi_{\lambda}$ of $P_{\pi}$ labeled by $\lambda$ is exactly equal to the number of SYTs of shape $Y_{\lambda}$ which can be calculated as
    \begin{equation*}
    \dim(\mathcal{S}_{\lambda}) = \frac{n!}{H_{Y_{\lambda}}}.
    \end{equation*}
    Moreover, the dimension of the irreducible representation $U_{\lambda}$ of $U^{\otimes n}$ labeled by $\lambda$ is exactly equal to the number of SSYTs of shape $Y_{\lambda}$ with filling $d$ which can be calculated as
    \begin{equation*}
    \dim(\mathcal{U}_{\lambda}) = \frac{\prod_{i,j} (d+j-i)}{H_{Y_{\lambda}}}.
    \end{equation*}
 \end{lemma}
In the following, we will provide a specific method for calculating the Schur basis under which the representation matrix of $U^{\otimes n}$ for any $U \in \opr{U}(d)$ will take on a block-diagonal form from the computation basis. Generally, we take $V = \mathbb{C}^{d}$ and the standard orthogonal basis of $\mathbb{C}^{d}$ is denoted as $\{e_1, e_2, \cdots, e_d\}$. Then the standard orthogonal tensor basis of $(\mathbb{C}^{d})^{\otimes n}$ which we called computation basis is given by 
\begin{equation*}
\{e_{i_1} \otimes e_{i_2} \otimes \cdots \otimes e_{i_n}|\, 1 \leq i_l \leq d, \, l=1,\cdots, n\}.
\end{equation*}
Next, we will illustrate the process of constructing Schur basis:\\
\textbf{1. Calculate the Young Symmetrizer using the Standard Young Tableaux (SYT):}
 
 \begin{enumerate}
    \item List all Young diagrams of size $n$ with no more than $d$ rows, that is the set $Y_{d}^{n}$.
    \item For each $Y_{\lambda} \in Y_{d}^{n}$, determine all possible SYTs $\{\theta _{i}^{\lambda}\}_{i}$ with shape $\lambda$ and all possible SSYTs $\{\Phi_{j}^{\lambda}(d)\}$ with shape $\lambda$ and filling $d$.
    \item For each SYT $\theta_{i}^{\lambda}$, construct the corresponding unnormalized Young Symmetrizer $P_{\theta_{i}^{\lambda}}= R_{\theta_{i}^{\lambda}} C_{\theta_{i}^{\lambda}}$:
    \begin{enumerate}
        \item Row Symmetrizer \( R_{\theta_{i}^{\lambda}} \): Each row symmetrizer \( R_{\theta_{i}^{\lambda}}^{j}  \) of $\theta_{i}^{\lambda}$ is defined as the sum of all permutations of the numbers in row $j$ of the Young tableau \( \theta_{i}^{\lambda} \) with normalization coefficient. Formally:
        \begin{equation*}
        R_{\theta_{i}^{\lambda}}^{j} = \frac{1}{m_{j}!} \sum_{\sigma \in \text{Row}_j(\theta_{i}^{\lambda})} \sigma
        \end{equation*}
        where \( \text{Row}_j(\theta_{i}^{\lambda}) \) denotes the symmetric group acting on the elements of row $j$ of $\theta_{i}^{\lambda}$ and $m_{j}$ is the number of boxes in the row $j$. Then we denote 
        \begin{equation*}
        R_{\theta_{i}^{\lambda}} = \prod_{j} R_{\theta_{i}^{\lambda}}^{j}.
        \end{equation*}
        
        \item Column Anti-symmetrizer \( C_{\theta_{i}^{\lambda}} \): Each column anti-symmetrizer \( C_{\theta_{i}^{\lambda}}^{k} \) is defined as the alternating sum of all permutations of the numbers in column $k$ of the Young tableau \( \theta_{i}^{\lambda} \). Formally:
        \begin{equation*}
        C_{\theta_{i}^{\lambda}}^{k} = \frac{1}{l_{k}!} \sum_{\tau \in \text{Col}_{k}(\theta_{i}^{\lambda})} \text{sgn}(\tau) \cdot \tau
        \end{equation*}
        where \( \text{sgn}(\tau) \) represents the sign of the permutation \( \tau \) which is determined by its parity, \( \text{Col}_{k}(\theta_{i}^{\lambda}) \) denotes the symmetric group acting on the elements of column $k$ of $\theta_{i}^{\lambda}$ and $l_{k}$ is the number of boxes in the column $k$. Then we denote 
        \begin{equation*}
            C_{\theta_{i}^{\lambda}} = \prod_{k} C_{\theta_{i}^{\lambda}}^{k}.
        \end{equation*}
        
        \item For example, take the SYT $\theta^{(3,2)}$ with shape $(3,2)$ as 
        \begin{equation*}
        \begin{array}{|c|c|c|}
            \hline
            1 & 3 & 4 \\
            \hline
            2 & 5   \\
            \cline{1-2}
            \multicolumn{2}{c}{} \\
        \end{array}
        \end{equation*}
        and we can calculate that
        \begin{equation*}
        \begin{aligned}
        R_1 &= \frac{1}{3!} \sum_{\sigma \in \text{Row}_{1}(\theta^{(3,2)})} \sigma = \frac{1}{6} ((1) + (13) + (14) + (34) + (143) + (134)),\\
        R_2 &= \frac{1}{2!} \sum_{\sigma \in \text{Row}_{2}(\theta^{(3,2)})} \sigma = \frac{1}{2} ((1) + (25))
        \end{aligned}
        \end{equation*}
        where we always denote $(1)$ as the identity permutation and we can get the Row Symmetrizer
        \begin{equation*}
            \begin{aligned}
                R_{\theta^{(3,2)}} = &\left( \frac{1}{6} \big((1) + (13) + (14) + (34) + (143) + (134)\big) \right) \cdot \\
                 &\left( \frac{1}{2} \big((1) + (25)\big) \right).
            \end{aligned}
        \end{equation*}
        In the similar way, we can calculate that
        \begin{equation*}
            \begin{aligned}
                C_1 &= \frac{1}{2!} \sum_{\tau \in \text{Col}_{1}(\theta^{(3,2)})} \text{sgn}(\tau) \tau = \frac{1}{2} \big((1) - (12)\big)\\
                C_2 &= \frac{1}{2!} \sum_{\tau \in \text{Col}_{2}(\theta^{(3,2)})} \text{sgn}(\tau) \tau = \frac{1}{2} \big((1) - (35)\big)\\
                C_3 &= \frac{1}{1!} \sum_{\tau \in \text{Col}_{3}(\theta^{(3,2)})} \text{sgn}(\tau) \tau = (1)\\
            \end{aligned}
        \end{equation*}
        and hence the Column Anti-symmetrization operator is:
        \begin{equation*}
        C_{\theta^{(3,2)}} = \left( \frac{1}{2} \big((1) - (12)\big) \right) \cdot \left( \frac{1}{2} \big((3) - (35)\big) \right) \cdot (1).
        \end{equation*}
        Then we can get the unnormalized Young Symmetrizer $P_{\theta^{(3,2)}}$
        \begin{equation*}
            \begin{aligned}
            P_{\theta^{(3,2)}} = &\left( \frac{1}{6} \big((1) + (13) + (14) + (34) + (143) + (134)\big) \right) \cdot \\
            &\left( \frac{1}{2} \big((1) + (25)\big) \right) \cdot \left( \frac{1}{2} \big((1) - (12)\big) \right) \cdot\\
            &\left( \frac{1}{2} \big((3) - (35)\big) \right) \cdot (1).
            \end{aligned}
        \end{equation*}
    \end{enumerate}
 \end{enumerate}
\textbf{2. Assign each seed vector to each $\Phi_{j}^{\lambda}$ corresponding to each $\theta_{i}^{\lambda}$:}
 
 \begin{enumerate}
    \item For each SSYT $\Phi_{j}^{\lambda}(d)$ and SYT $\theta_{i}^{\lambda}$, we assign a unique computation basis vector corresponding to the order of filling numbers in $\theta_{i}^{\lambda}$ to $\Phi_{j}^{\lambda}(d)$ as the seed vector $e_{\theta_{i},\Phi_{j}(d)}^{\lambda}$. 
    
    \item For example, we take $\lambda = (3,2)$, $d = 3$ and the SSYT $\Phi^{(3,2)}(3)$ as 
        \begin{equation*}
        \begin{array}{|c|c|c|}
            \hline
            1 & 2 & 2 \\
            \hline
            3 & 3   \\
            \cline{1-2}
            \multicolumn{2}{c}{} \\
        \end{array}.
        \end{equation*}
        Moreover, we take the SYT $\theta^{(3,2)}$ as 
        \begin{equation*}
        \begin{array}{|c|c|c|}
            \hline
            1 & 2 & 4 \\
            \hline
            3 & 5   \\
            \cline{1-2}
            \multicolumn{2}{c}{} \\
        \end{array}
        \end{equation*}
        and then we will get the seed vector
        \begin{equation*}
        e_{\theta, \Phi(3)}^{(3,2)} = e_1 \otimes e_2 \otimes e_3 \otimes e_2 \otimes e_3.
        \end{equation*}
 \end{enumerate}
 \textbf{3. Construct the Schur basis matrix of $U^{\otimes n}$ for any $U \in \opr{U}(d)$:}
 
 \begin{enumerate}
    \item From Schur-Weyl duality theory we know that the irreducible representations of $U^{\otimes n}$ can be labeled by the Young diagram $Y_{\lambda} \in Y_{d}^{n}$. Next we will construct the Schur basis for each $U_{\lambda}$.
    
    \item For each $\theta_{i}^{\lambda}$, we construct the following space
    \begin{equation*}
        V_{\theta_{i}^{\lambda}} = \{P_{\theta_{i}^{\lambda}}(e_{\theta_{i},\Phi_{j}(d)}^{\lambda})\}_{j}
    \end{equation*}
    by applying the unnormalized Young Symmetrizer $P_{\theta_{i}^{\lambda}}$ successively to each $e_{\theta_{i},\Phi_{j}(d)}^{\lambda}$ with iterating $j$.
    
    \item Then for each $i$, $V_{\theta_{i}^{\lambda}}$ is an irreducible representation space $\mathcal{U}_{\lambda}$ and we can find that its dimension matches the number of SSYTs of shape $Y_{\lambda}$ with filling $d$. Moreover, we can also find that its multiplicity matches the number of SYTs of shape $Y_{\lambda}$ and these two results match Lemma~\ref{D-M}.
    
    \item Finally, we define the matrix
    \begin{equation*}
    Q_{\theta_i^{\lambda}} = 
    \bigl[\, P_{\theta_i^{\lambda}}(e_{\theta_i, \Phi_j(d)}^{\lambda}) \,\bigr]_j
    \quad \text{and} \quad
    Q_{\lambda} = 
    \bigl[\, Q_{\theta_i^{\lambda}} \,\bigr]_i 
    \end{equation*}
    which means the columns of \( Q_{\theta_i^{\lambda}} \) are the vectors 
    \( P_{\theta_i^{\lambda}}(e_{\theta_i, \Phi_j(d)}^{\lambda}) \) for different \( j \),
    and the matrices \( Q_{\theta_i^{\lambda}} \) are concatenated side by side 
    to form \( Q_{\lambda} \). Then we construct the Schur matrix $Q$ by
    \begin{equation*}
     \tilde{Q} 
     = \bigl[\, Q_{\lambda} \,\bigr]_{\lambda}
     \xrightarrow{\text{Gram–Schmidt orthonormalization}}
     Q
    \end{equation*}
    which means that for all $U \in \opr{U}(d)$
    \begin{equation*}
        Q^{\dagger}U^{\otimes n}Q = \bigoplus_{Y_\lambda \in Y_{d}^{n}} \operatorname{diag}(\underbrace{U^{\lambda}, U^{\lambda}, \dots, U^{\lambda}}_{\dim(\mathcal{S}_{\lambda}) \text{times}}).
    \end{equation*}
 
 \end{enumerate}
 Next, we will take $n = d = 2$ as an example to illustrate this process: take $\mathbb{C}^2$ with the standard orthogonal basis $\{e_1,e_2\}$. Then the standard orthogonal tensor basis of $(\mathbb{C}^{2})^{\otimes 2}$ is
    \begin{equation*}
    \{e_1\otimes e_1,\; e_1\otimes e_2,\; e_2\otimes e_1,\; e_2\otimes e_2\}.
    \end{equation*}
 Then the set $Y_{d}^{n}$ of all Young diagrams with size $n = 2$ and at most $d = 2$ rows is:
    \begin{equation*}
    Y_{2}^{2} = \{\lambda_{1} = (2), \lambda_{2} = (1,1)\}.
    \end{equation*}
 For $\lambda_{1} = (2)$, there is only one SYT with the shape $Y_{(2)}$, that is
 \begin{equation*}
    \theta^{(2)} =
    \begin{array}{|c|c|}
        \hline
        1 & 2  \\
        \hline
    \end{array}
 \end{equation*}
and we can calculate its corresponding unnormalized Young Symmetrizer
\begin{equation*}
    \begin{aligned}
    P_{\theta^{(2)}} & = R_{\theta^{(2)}} C_{\theta^{(2)}}\\
    & = \frac{1}{2!}\big((1)+(12)\big) \cdot (1)\\
    & = \frac{1}{2}\big((1)+(12)\big).
    \end{aligned}
\end{equation*}
Next, there are three SSYTs with the shape $Y_{(2)}$ and filling $2$, they are
\begin{equation*}
\Phi_{1}^{(2)}(2) = \begin{array}{|c|c|}
    \hline
    1 & 1  \\
    \hline
    \end{array}
    \,\, , 
\Phi_{2}^{(2)}(2) = \begin{array}{|c|c|}
        \hline
        1 & 2  \\
        \hline
    \end{array}
    \,\, , 
\Phi_{3}^{(2)}(2) = \begin{array}{|c|c|}
        \hline
        2 & 2  \\
        \hline
    \end{array}
\end{equation*}
and their seed vectors corresponding to $\theta^{(2)}$ are
\begin{equation*}
    e_{\theta, \Phi_{1}(2)}^{(2)} = e_{1} \otimes e_{1}\,\, , e_{\theta, \Phi_{2}(2)}^{(2)} = e_{1} \otimes e_{2}\,\, , e_{\theta, \Phi_{3}(2)}^{(2)} = e_{2} \otimes e_{2}.
\end{equation*}
Hence we can get the matrix $Q_{(2)} = Q_{\theta^{(2)}}$ as
\begin{equation*}
    \begin{aligned}
    Q_{(2)} = Q_{\theta^{(2)}} &= [P_{\theta^{(2)}}(e_{\theta, \Phi_{j}(2)}^{(2)})]_{j = 1,2,3}\\
    &= [e_{1} \otimes e_{1}, \frac{1}{2}(e_{1} \otimes e_{2} + e_{2} \otimes e_{1}), e_{2} \otimes e_{2}].
    \end{aligned}
\end{equation*}
Then for $\lambda_{2} = (1,1)$, there is also only one SYT with the shape $Y_{(1,1)}$, that is
\begin{equation*}
    \theta^{(1,1)} = \begin{array}{|c|}
        \hline
        1  \\
        \hline
        2\\
        \cline{1-1}
    \end{array}
\end{equation*}
and we can calculate its corresponding unnormalized Young Symmetrizer
\begin{equation*}
    \begin{aligned}
        P_{\theta^{(1,1)}} & = R_{\theta^{(1,1)}} C_{\theta^{(1,1)}}\\
        & = (1) \cdot \frac{1}{2}\big((1) - (12)\big)\\
        & = \frac{1}{2}\big((1) - (12)\big).
    \end{aligned}
\end{equation*}
Next, there is only one SSYT with the shape $Y_{(1,1)}$ and filling 2, that is
\begin{equation*}
    \Phi^{(1,1)}(2) = \begin{array}{|c|}
        \hline
        1  \\
        \hline
        2\\
        \cline{1-1}
    \end{array}
\end{equation*}
and its seed vector corresponding to $\theta^{(1,1)}$ is
\begin{equation*}
    e_{\theta, \Phi(2)}^{(1,1)} = e_{1} \otimes e_{2}.
\end{equation*}
Hence we can get the matrix $Q_{(1,1)} = Q_{\theta^{(1,1)}}$ as
\begin{equation*}
    \begin{aligned}
        Q_{(1,1)} = Q_{\theta^{(1,1)}} &= [P_{\theta^{(1,1)}}(e_{\theta, \Phi(2)}^{(1,1)})]\\
        &= [\frac{1}{2}(e_{1} \otimes e_{2} - e_{2} \otimes e_{1})].
    \end{aligned}
\end{equation*}
Then we will get the following unnormalized matrix
\begin{equation*}
    \begin{aligned}
        \tilde{Q} &= [Q_{(2)},Q_{(1,1)}]\\
        &= [e_{1} \otimes e_{1}, \frac{1}{2}(e_{1} \otimes e_{2} + e_{2} \otimes e_{1}), e_{2} \otimes e_{2}, \frac{1}{2}(e_{1} \otimes e_{2} - e_{2} \otimes e_{1})].
    \end{aligned}
\end{equation*}
After the process of Gram–Schmidt orthonormalization, we can get the well-known Schur matrix $Q$ for $U^{\otimes 2}$ where $U \in \opr{U}(d)$:
\begin{equation*}
    Q = [e_{1} \otimes e_{1}, \frac{1}{\sqrt{2}}(e_{1} \otimes e_{2} + e_{2} \otimes e_{1}), e_{2} \otimes e_{2}, \frac{1}{\sqrt{2}}(e_{1} \otimes e_{2} - e_{2} \otimes e_{1})].
\end{equation*}

\section{Quantum comb}
\label{app:comb}
In this section, we will introduce the notion of multi-slot parallel and sequential quantum comb. The linear spaces associated to input and output are described by the tensor product of $i$ subspaces. In this work, we will use bold letters to indicate this tensor product subsystem stucture:
\begin{equation*}
\mathcal{H}_{\boldsymbol{I}} := \otimes_{i = 1}^{k} \mathcal{H}_{I_i}, \quad \mathcal{H}_{\boldsymbol{O}} := \otimes_{i = 1}^{k} \mathcal{H}_{O_i}.
\end{equation*}
Sequential quantum comb represents general quantum circuits where different encoder operations are applied in between the uses of the input channels $\mathcal{C}_i$ \cite{Theo-frame}. We can see Figure~\ref{fig:comb_framework} for the illustration of the difference among three kinds of quantum comb. 
\begin{figure}[h!]
    \centering
    \includegraphics[width=0.9\linewidth]{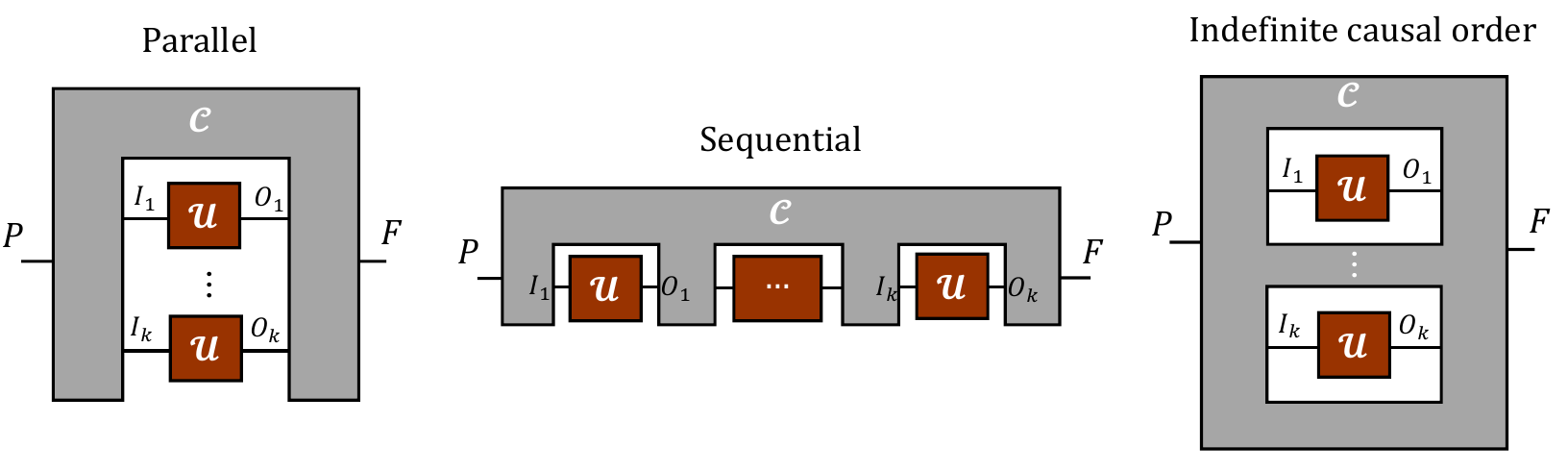}
    \caption{Three kinds of quantum combs involving the parallel, sequential and indefinite causal order. The alphabets $P,I_j,O_j$, and $F$ label the corresponding Hilbert spaces $\mathcal{H}_P, \mathcal{H}_{I_j},\mathcal{H}_{O_j}$, and $\mathcal{H}_F$, respectively.}
    \label{fig:comb_framework}
\end{figure}
For instance, in the case of $k=2$ slots, sequential quantum comb consist of two encoding channels with Choi operators $E_1, E_2$ and one decoder channel with Choi operator $D$. If we plug in two input channels with Choi operators $C_1$ and $C_2$, the output channel $C_{\mathrm{out}}$ is given by the composition 
\begin{equation*}
    C_{\mathrm{out}} = D * C_2 * E_2 * C_1 * E_1.
\end{equation*}
Formally, we can define sequential quantum comb as follows.
\medskip
\begin{definition}
A linear operator $S \in \mathsf{L}\left(\mathcal{H}_P \bigotimes_{i=1}^{k} (\mathcal{H}_{I_{i}} \otimes \mathcal{H}_{O_{i}}) \otimes \mathcal{H}_F\right)$ is a $k$-slot sequential quantum comb if there exist a linear space $\mathcal{H}_{\mathrm{aux}}$, a quantum channel 
$
\mathcal{E}_1 : \mathsf{L}(\mathcal{H}_P) \to \mathsf{L}(\mathcal{H}_{\mathrm{aux}} \otimes \mathcal{H}_{I_1}),
$
a set of quantum channels 
$
\mathcal{E}_i : \mathsf{L}(\mathcal{H}_{\mathrm{aux}} \otimes \mathcal{H}_{O_{i-1}}) \to \mathsf{L}(\mathcal{H}_{\mathrm{aux}} \otimes \mathcal{H}_{I_i})
$
for $ i \in \{2,\dots,k\}$, and a quantum channel 
$
\mathcal{D}: \mathsf{L}(\mathcal{H}_{\mathrm{aux}} \otimes \mathcal{H}_{O_k}) \to \mathsf{L}(\mathcal{H}_F)
$
such that
\begin{equation*}
S = E_1 * E_2 * \cdots * E_k * D,
\end{equation*}
where $E_{i}$ is the Choi operator of $\mathcal{E}_i$ for $i \in \{1,\dots,k\}$ and $D$ is the Choi operator of $\mathcal{D}$.
\end{definition}

Sequential quantum comb can also be characterised in terms of linear and positive semidefinite constraints. We state as follows: 
a linear operator \\
$S \in \mathsf{L}\left(\mathcal{H}_P \bigotimes_{i=1}^{k} (\mathcal{H}_{I_{i}} \otimes \mathcal{H}_{O_{i}}) \otimes \mathcal{H}_F\right)$ represents a sequential quantum comb with $k$-slots if and only if \cite{Chiribella2008quantum,Theo-frame}
\begin{align}
S &\ge 0, \nonumber\\
\tr_F(S) &= \tr_{O_k F}(S) \otimes \frac{I_{O_k}}{d_{O_k}}, \nonumber\\
\tr_{I_k O_k F}(S) &= \tr_{O_{k-1} I_k O_k F}(S)\otimes \frac{I_{O_{k-1}}}{d_{O_{k-1}}}, \nonumber\\
&\;\;\vdots \nonumber\\
\tr_{I_1 O_1 \cdots I_k O_k F}(S) &= \tr_{P I_1 O_1 \cdots I_k O_k F}(S)\otimes \frac{I_P}{d_P}, \nonumber\\
\tr(S) &= d_P d_{\boldsymbol{O}} .
\label{eq:sequential_constraints}
\end{align}

Parallel quantum comb can be characterised by a single encoder and a single decoder channel. More precisely, we can give the definition.

\begin{definition}
A linear operator $S \in \mathsf{L}(\mathcal{H}_P \otimes \mathcal{H}_{\boldsymbol{I}} \otimes \mathcal{H}_{\boldsymbol{O}} \otimes \mathcal{H}_F)$ is a $k$-slot parallel quantum comb if there exist a linear space $\mathcal{H}_{\mathrm{aux}}$, a quantum channel $\mathcal{E}: \mathsf{L}(\mathcal{H}_P)\to \mathsf{L}(\mathcal{H}_{\mathrm{aux}}\otimes \mathcal{H}_{\boldsymbol{I}})$, and $\mathcal{D}: \mathsf{L}(\mathcal{H}_{\mathrm{aux}}\otimes \mathcal{H}_{\boldsymbol{O}})\to \mathsf{L}(\mathcal{H}_F)$ with Choi operators $E$ and $D$ such that $S = E \ast D$.
\end{definition}

Similarly, it can be shown that a linear operator $S \in \mathsf{L}(\mathcal{H}_P \otimes \mathcal{H}_{\boldsymbol{I}} \otimes \mathcal{H}_{\boldsymbol{O}} \otimes \mathcal{H}_F)$ is a $k$-slot parallel quantum comb if and only if
\begin{equation*}
\begin{aligned}
S &\ge 0 \\
\operatorname{Tr}_F(S) &= \operatorname{Tr}_{\boldsymbol{O}F}(S)\otimes \frac{I_{\boldsymbol{O}}}{d_{\boldsymbol{O}}} \\
\operatorname{Tr}_{\boldsymbol{IO}F}(S) &= \operatorname{Tr}_{P\boldsymbol{IO}F}(S) \otimes \frac{I_P}{d_P} \\
\operatorname{Tr}(S) &= d_P d_{\boldsymbol{O}}. \label{eq:parallel-constraints}
\end{aligned}
\end{equation*}

When transforming quantum operations, parallel implementations are often desirable due to their simpler structure, as they can be realised by a single encoder and a single decoder channel. Also, parallel quantum comb can be realised by a quantum circuit with short depth (encoder, input channels, decoder) while a sequential use of the input operations may result in a long depth, and consequently, in a longer time to finish the whole transformation.

\section{Qubit shadow unitary inversion}

As we stated earlier, in {qubit case} we may restrict our analysis to the case $O = Z$ without loss of generality. More precisely, if the equation \eqref{shadow-inversion} holds for $O = Z$, one can show that it is equivalent to
\begin{equation*}
\operatorname{Tr}\!\bigl[\mathcal{N}_{U}(\rho)\proj{i} \bigr]
    = \operatorname{Tr}\!\bigl[U^{\dagger} \rho U \proj{i} \bigr], \,\, \forall i\in\{0,1\}.
\end{equation*}
For any qubit observable $O$, there exists unitary $V$ such that $O = V \Sigma V^{\dagger}$ where $\Sigma$ is real diagonal. Hence one can simply append $V$ at the output stage of the circuit, thereby reducing the problem to the $Z$-observable scenario.

\subsection{A circuit implementation}\label{proof-of-thm}

Here we present the circuit construction that reverses any qubit unitary $U$ under observable $Z$ by querying $3$ times of $U$, thereby proving Theorem \ref{main-thm} {which we restate here for clarity:}
\maintheorem*
\noindent\textit{Proof sketch of Theorem~\ref{main-thm}.}
Without loss of generality, we assume that $U$ is special, i.e. $\det U=1$. For any unknown $2$-dimensional unitary $U$ and any unknown input state $\ket\psi$, there exist fixed quantum circuits $V_0,V_3\in \opr{U}(8)$, $V_1,V_2\in\opr{U}(16)$ satisfying
\begin{align}
    \ket{\Psi_{\opr{I}}}&:= (I\otimes I\otimes U)\cdot V_0\cdot(\ket{0}\otimes\ket{0}\otimes\ket{\psi})=\frac12\sum_{j=0}^3\ket{j}\otimes UP_j\ket{\psi},\label{eq:Psi1}\\
    \ket{\Psi_{\opr{II}}}&:= (I\otimes I\otimes I\otimes U)\cdot V_1\cdot(\ket{0}\otimes\ket{\Psi_{\opr{I}}})\nonumber\\
    &\begin{aligned}
        =\frac{1}{2\sqrt{3}}\big(&\ket{v_{01}}\otimes (UXU^\dag-X)+i\ket{v_{23}}\otimes (UXU^\dag+X)+\\
        &\ket{v_{02}}\otimes (UYU^\dag- Y)-i\ket{v_{13}}\otimes (UYU^\dag+Y)+\\
        &\ket{v_{03}}\otimes (UZU^\dag- Z)+i\ket{v_{12}}\otimes (UZU^\dag+Z)\big)\ket{\psi},
    \end{aligned}\label{eq:Psi2}\\
\ket{0}\otimes\ket{\Psi_{\opr{III}}}&:= (I\otimes I\otimes I\otimes U)\cdot V_2\cdot\ket{\Psi_{\opr{II}}}\label{eq:Psi3}\\
    &\begin{aligned}
        =\frac{1}{2\sqrt{3}}\ket{0}\otimes\big(
            &\ket{0}\otimes (2UZU+UZU^\dag )+
              \ket{1}\otimes (2iUYU-UZU^\dag X )+\\
            &\ket{2}\otimes (-2iUXU-UZU^\dag Y)+
              \ket{3}\otimes (2UU-UZU^\dag Z )\big)U^\dag\ket{\psi},
    \end{aligned}\nonumber\\
    \ket{\Psi_{\opr{IV}}}&:=  V_3\cdot\ket{\Psi_{\opr{III}}}\nonumber\\
    &\begin{aligned}
        =\frac{1}{2\sqrt{3}}\big(
        &\ket{0}\otimes2I+
        \ket{1}\otimes(I\opr{Tr}[U^\dag YUY]-iZ\opr{Tr}[U^\dag YUX])+\\
        &\ket{2}\otimes(I\opr{Tr}[U^\dag XUX]+iZ\opr{Tr}[U^\dag XUY])+\\
        &\ket{3}\otimes(iI\opr{Tr}[U^\dag ZUY]+
        Z\opr{Tr}[U^\dag ZUX])\big)U^\dag\ket{\psi},
    \end{aligned}\label{eq:Psi4}
\end{align}
where $P_0:= I,P_1:= X,P_2:= Y,P_3:= Z$,
\begin{equation*}
    \begin{array}{ll}
        \ket{v_{01}} = \ket{000}, &
        \ket{v_{02}} = \ket{001}, \\[2mm]
        \ket{v_{03}} = \ket{010}, &
        \ket{v_{12}} = \tfrac{\sqrt{3}}{2}\ket{110}
        + \tfrac{i}{2}\ket{010}, \\[2mm]
        \ket{v_{13}} = \tfrac{\sqrt{3}}{2}\ket{101}
        - \tfrac{i}{2}\ket{001}, &
        \ket{v_{23}} = \tfrac{\sqrt{3}}{2}\ket{100}
        + \tfrac{i}{2}\ket{000}.
    \end{array}
\end{equation*}
Hence, after tracing the first three qubits in $\ket{\Psi_{\opr{IV}}}$, it derives a quantum circuit $\mathcal N_U$ satisfying  
\begin{equation}\label{eq:Psi4_tr}
    \forall\rho\in\mathsf{D}(\mathbb{C}^2),\ \tr\!\bigl[\mathcal{N}_{U}(\rho)\, Z \bigr]
    = \tr\!\bigl[U^{\dagger} \rho U Z \bigr].
\end{equation}

\begin{proof}[Proof of Theorem~\ref{main-thm}]
The circuit diagram is given by Fig. \ref{fig:decompose_comb}, while its correctness is shown by following calculation. We remark that, since the state is input to the last qubit of our circuit, in our calculation we track the evolution operator which, with a input state $\ket{\psi}$, acts on $I\otimes I\otimes I\otimes \ket{\psi}$ to get the output state.
\begin{figure*}
        \centering
        \begin{quantikz}[transparent]
            \lstick{$\ket{0}$} & {} & {} & \gate[4]{V_1} & {} & \gate[4]{V_2} & {} & {} & \rstick{$\ket{0}$} \\
            \lstick{$\ket{0}$} & \gate[3]{V_0} & {} & {} & {} & {} & {} & \gate[3]{V_{3}} & \trash{\text{trace}} \\
            \lstick{$\ket{0}$} & {} & {} & {} & {} & {} & {} & {} & \trash{\text{trace}} \\
            \lstick{} & {} & \gate[1]{U} & {} & \gate[1]{U} & {} & \gate[1]{U} & {} & \rstick{}
        \end{quantikz}
        \caption{The circuit configuration of decomposed quantum comb for inversing unknown single-qubit unitary regarding the observable $Z$. Three ancilla qubits are employed (initiated to $\ket{0}$); one acts as a catalytic ancilla that is restored to $\ket0$ (\emph{i.e.}, not consumed) by the end of the circuit.}\label{fig:decompose_comb}
    \end{figure*}
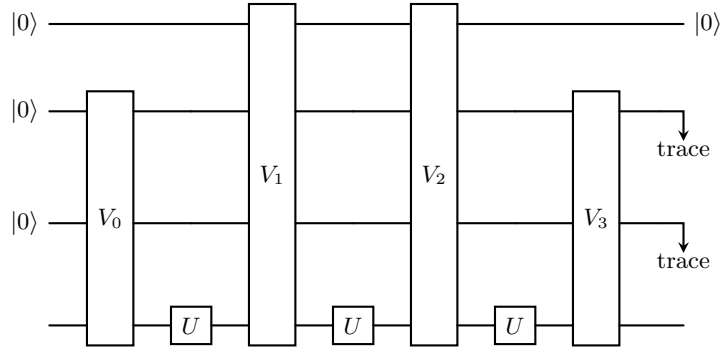

The circuit requires 3 ancilla qubits, which are initiated to be $\ket{000}$.

\noindent\textit{Proof of \eqref{eq:Psi1}.}
The first gate $V_0$ acts two Hadamard gates on the second and the third qubit, followed by a controlled Pauli gate on the last qubit, i.e.
\begin{equation}
    V_0=\sum_{j=0}^3 \proj{j}H^{\otimes2}\otimes P_j,
\end{equation}
which implies
\begin{equation}
    \begin{aligned}
        \ket{\Psi_{\opr{I}}}&:= (I\otimes I\otimes U)\cdot V_0\cdot(\ket{0}\otimes\ket{0}\otimes\ket{\psi})\\
        &=\sum_{j=0}^3 \proj{j}H^{\otimes2}\ket{00}\otimes UP_j\ket{\psi}\\
        &=\frac12\sum_{j=0}^3\ket{j}\otimes UP_j\ket{\psi}.
    \end{aligned}
\end{equation}
\noindent\textit{Proof of \eqref{eq:Psi2}.}
The second gate $V_1$ is defined after vectors $\ket{v_{jk}}$ by
\begin{equation*}
    V_1 = \frac{1}{\sqrt3}\sum_{j\ne k=0}^3\ket{v_{jk}}\!\bra{k}\otimes P_j.
\end{equation*}
where
\begin{equation*}
    \ket{v_{kj}}=-\ket{v_{jk}}\text{ for }0\le j<k\le3,
\end{equation*}
which implies
\begin{equation}
    \begin{aligned}
        \ket{\Psi_{\opr{II}}}
        &:= (I\otimes I\otimes I\otimes U)\cdot V_1\cdot(\ket{0}\otimes\ket{\Psi_{\opr{I}}})\\
        &=\frac{1}{2\sqrt{3}}\sum_{j\ne k=0}^3\ket{v_{jk}}\!\bra{k}\otimes UP_j\cdot\sum_{l=0}^3\ket{l}\otimes UP_l\ket{\psi}\\
        &=\frac{1}{2\sqrt{3}}\sum_{0\le j< k\le3}\ket{v_{jk}}\otimes (UP_jUP_k-UP_kUP_j)\ket{\psi}\\
        &\begin{aligned}
            =\frac{1}{2\sqrt{3}}\big(
            &\ket{v_{01}}\otimes (UXU^\dag-X)+i\ket{v_{23}}\otimes (UXU^\dag+X)+\\
            &\ket{v_{02}}\otimes (UYU^\dag- Y)-i\ket{v_{13}}\otimes (UYU^\dag+Y)+\\
            &\ket{v_{03}}\otimes (UZU^\dag- Z)+i\ket{v_{12}}\otimes (UZU^\dag+Z)\big)\ket{\psi}
        \end{aligned}
    \end{aligned}
\end{equation}
\noindent\textit{Proof of \eqref{eq:Psi3}.}
The third gate $V_2$ is composed into two controlled Pauli gates and a 2-qubit base-change gate $G$:
\begin{equation*}
    V_2 = \left(I\otimes\sum_{j=0}^3\proj{j}\otimes ZP_j\right)
    \cdot\left(G\otimes I\right)
    \cdot\left(I\otimes\sum_{j=0}^3\proj{j}\otimes P_{j+1\opr{mod}4})\right)
\end{equation*}
where $G\in\opr{U}(8)$ satisfying
\begin{equation*}
    \begin{aligned}
        &G\left(-\ket{v_{01}}+i\ket{v_{23}}-\ket{v_{02}}-i\ket{v_{13}}-\ket{v_{03}}+i\ket{v_{12}}\right)=\ket{0}\otimes\frac32(1,1,1,1)^{\intercal},\\
        &G(\ket{v_{01}}+i\ket{v_{23}})=\ket{0}\otimes\frac12(1,1,-1,-1)^{\intercal},\\
        &G(\ket{v_{02}}-i\ket{v_{13}})=\ket{0}\otimes\frac12(1,-1,1,-1)^{\intercal},\\
        &G(\ket{v_{03}}+i\ket{v_{12}})=\ket{0}\otimes\frac12(1,-1,-1,1)^{\intercal}.
    \end{aligned}
\end{equation*}
Then \eqref{eq:Psi3} could be checked by
\begin{equation*}
    \begin{aligned}
        \ket{\Psi_{\opr{II}}}\xrightarrow{\sum_j\proj{j}\otimes P_{j+1}}
        &\frac{1}{2\sqrt{3}}\big(\ket{v_{01}}\otimes (XUXU^\dag-I)+i\ket{v_{23}}\otimes (XUXU^\dag+I)+\\
        &\qquad\ \ket{v_{02}}\otimes (YUYU^\dag- I)-i\ket{v_{13}}\otimes (YUYU^\dag+I)+\\
        &\qquad\ \ket{v_{03}}\otimes (ZUZU^\dag- I)+i\ket{v_{12}}\otimes (ZUZU^\dag+I)\big)\ket{\psi}\\
        =&
        \frac{1}{2\sqrt{3}}{\bigl(}\left(-\ket{v_{01}}+i\ket{v_{23}}-\ket{v_{02}}-i\ket{v_{13}}-\ket{v_{03}}+i\ket{v_{12}}\right)\otimes I\\
        &\qquad+(\ket{v_{01}}+i\ket{v_{23}})\otimes XUXU^\dag\\
        &\qquad
        +
        (\ket{v_{02}}-i\ket{v_{13}})\otimes YUYU^\dag\\
        &\qquad +(\ket{v_{03}}+i\ket{v_{12}})\otimes ZUZU^\dag{\bigr)}\ket{\psi}\\
        \xrightarrow{G}
        &\ket{0}\otimes\frac{1}{4\sqrt{3}}\left(3(1,1,1,1)^{\intercal}\otimes U+(1,1,-1,-1)^{\intercal}\otimes XUX+\right.\\
        &\qquad\qquad\quad\left.(1,-1,1,-1)^{\intercal}\otimes YUY+(1,-1,-1,1)^{\intercal}\otimes ZUZ\right)U^\dag\ket{\psi}\\
        =
        &\ket{0}\otimes\frac{1}{4\sqrt{3}}\left(
        \ket{0}\otimes (4U+2U^\dag )+
        \ket{1}\otimes (4U-2XU^\dag X )+\right.\\
        &\qquad\qquad\quad\left.\ket{2}\otimes (4U-2YU^\dag Y)+
        \ket{3}\otimes (4U-2ZU^\dag Z )\right)U^\dag\ket{\psi}\\
        \xrightarrow{\sum_j\proj{j}\otimes UZP_{j}}&
        \ket{0}\otimes\frac{1}{2\sqrt{3}}\left(
        \ket{0}\otimes (2UZU+UZU^\dag )+
        \ket{1}\otimes (2iUYU-UZU^\dag X )+\right.\\
        &\qquad\qquad\quad\left.
        \ket{2}\otimes (-2iUXU-UZU^\dag Y)+
        \ket{3}\otimes (2UU-UZU^\dag Z )\right)U^\dag\ket{\psi}.
    \end{aligned}
\end{equation*}
It is noted that the first ancilla qubit will be at state $\ket{0}$ after operation $G$, and those vectors $\ket{v_{jk}}$ are designed to deduce the dimension of ancilla system into $4$ here and to be orthogonal to each other.
\noindent\textit{Proof of \eqref{eq:Psi4}.}
The fourth gate $V_3$ is composed into two controlled Pauli gates and two Hadamard gates:
\begin{equation}
    V_3=(CCX)\cdot(H^{\otimes2}\otimes I)\cdot(\proj{0}\otimes Z
            -i\proj{1}\otimes Y+i\proj{2}\otimes X
        -\proj{3}\otimes I)
\end{equation}
Then the circuit output reads
\begin{align}
    \ket{\Psi_{\opr{III}}}
    \xrightarrow[+i\proj{2}\otimes X
    -\proj{3}\otimes I]{\proj{0}\otimes Z
        -i\proj{1}\otimes Y}
    &\frac{1}{2\sqrt{3}}\left(
    \ket{0}\otimes (2ZUZU+ZUZU^\dag )+\right.\nonumber\\
    &
    \qquad\ket{1}\otimes (2YUYU+iYUZU^\dag X )+\nonumber\\
    &\qquad
    \ket{2}\otimes (2XUXU-iXUZU^\dag Y)+\nonumber\\
    &\qquad\left.
    \ket{3}\otimes (-2U^2+UZU^\dag Z )\right)U^\dag\ket{\psi}\nonumber\\
    \xrightarrow{H^{\otimes2}}
    &\frac{1}{4\sqrt{3}}\left(
    (1,1,1,1)^{\intercal}\otimes (2ZUZU+ZUZU^\dag )+\right.\nonumber\\
    &\qquad\left.
    (1,-1,1,-1)^{\intercal}\otimes (2YUYU+iYUZU^\dag X )+\right.\nonumber\\
    &\qquad\left.
    (1,1,-1,-1)^{\intercal}\otimes (2XUXU-iXUZU^\dag Y)+\right.\nonumber\\
    &\qquad\left.
    (1,-1,-1,1)^{\intercal}\otimes (-2UU+UZU^\dag Z )\right)U^\dag\ket{\psi}\nonumber\\
    =&\frac{1}{2\sqrt{3}}\left(
    \ket{0}\otimes2I+
    \ket{1}\otimes(I\opr{Tr}[U^\dag YUY]-iZ\opr{Tr}[U^\dag YUX])+\right.\nonumber\\
    &\qquad\left.
    \ket{2}\otimes(I\opr{Tr}[U^\dag XUX]+iZ\opr{Tr}[U^\dag XUY])+\right.\nonumber\\
    &\qquad\left.
    \ket{3}\otimes(-iY\opr{Tr}[U^\dag ZUX]+iX\opr{Tr}[U^\dag ZUY])\right)U^\dag\ket{\psi}\nonumber\\
    \xrightarrow{CCX}&\frac{1}{2\sqrt{3}}\left(
    \ket{0}\otimes2I+
    \ket{1}\otimes(I\opr{Tr}[U^\dag YUY]-iZ\opr{Tr}[U^\dag YUX])+\right.\nonumber\\
    &\qquad\left.
    \ket{2}\otimes(I\opr{Tr}[U^\dag XUX]+iZ\opr{Tr}[U^\dag XUY])+\right.\label{eq:1}\\
    &\qquad\left.
    \ket{3}\otimes(iI\opr{Tr}[U^\dag ZUY]+
    Z\opr{Tr}[U^\dag ZUX])\right)U^\dag\ket{\psi}.\nonumber
\end{align}

\noindent\textit{Proof of \eqref{eq:Psi4_tr}.}
After tracing all ancilla qubits, we obtain
\begin{align*}
    &\opr{Tr}[\mathcal N_U(\rho)Z]\\
    =&\frac{1}{12}\opr{Tr}\left[U^\dag\rho U\big(
    2I\cdot Z\cdot 2I+\right.\\
    &\qquad\left.(I\opr{Tr}[U^\dag YUY]-iZ\opr{Tr}[U^\dag YUX])^\dag Z(I\opr{Tr}[U^\dag YUY]-iZ\opr{Tr}[U^\dag YUX])+\right.\\
    &\qquad\left.
    (I\opr{Tr}[U^\dag XUX]+iZ\opr{Tr}[U^\dag XUY])^\dag Z(I\opr{Tr}[U^\dag XUX]+iZ\opr{Tr}[U^\dag XUY])+\right.\\
    &\qquad\left.(iI\opr{Tr}[U^\dag ZUY]-Z\opr{Tr}[U^\dag ZUX])^\dag Z(iI\opr{Tr}[U^\dag ZUY]-Z\opr{Tr}[U^\dag ZUX])
    \big)\right]\\
    =&\frac{1}{12}\opr{Tr}[U^\dag\rho U\big(
    4Z+(\opr{Tr}[U^\dag YUY]^2+\opr{Tr}[U^\dag YUX]^2)Z+\\
    &\qquad(\opr{Tr}[U^\dag XUX]^2+\opr{Tr}[U^\dag XUY]^2)Z+(\opr{Tr}[U^\dag ZUY]^2+\opr{Tr}[U^\dag ZUX]^2)Z)
    \big)]\\
    =&\opr{Tr}[U^\dag\rho UZ]\cdot\frac{1}{12}\left(
    4+(\opr{Tr}[ XUYU^\dag]^2+\opr{Tr}[ YUYU^\dag]^2+\opr{Tr}[ ZUYU^\dag]^2)+\right.\\
    &\qquad\left.(\opr{Tr}[ XUXU^\dag]^2+\opr{Tr}[ YUXU^\dag]^2+\opr{Tr}[ ZUXU^\dag]^2)
    \right)\\
    =&\opr{Tr}[U^\dag\rho UZ]\cdot\frac{1}{12}(
    4+4+4)\\
    =&\opr{Tr}[U^\dag\rho UZ],
\end{align*}
where any density matrix $\rho$ could be regarded as a linear combinations of pure states $\proj{\psi}$.
\end{proof}
    
In the end of this section, we remark that our construction (Fig. \ref{fig:decompose_comb}) also realize a probabilistic unitary inversion with a fixed success probability: notice that in Eq.~\eqref{eq:1}, the first term (the $\ket{0}$ term) only contains $U^\dag$ with a constant coefficient. Therefore whenever the computational-basis measurement on the second and third qubit outputs $\ket{00}$, we implement a $U^\dag$ on the last qubit.

\subsection{Necessary and sufficient condition}

From the calculation in the proof of Theorem \ref{main-thm} one can also conclude that our circuit in fact acts like
\begin{equation*}
    \mathcal{N}_{U}(\rho) = p(U) U^{\dagger} \rho U + q(U)ZU^{\dagger}  \rho UZ
\end{equation*}
for some function $p(U) + q(U) = 1$. Then it is natural to ask, should all shadow qubit-unitary inversion admit this form?
{As announced in Section~\ref{prob-formu}, Proposition~\ref{form-des} determines the structure of $t$-query shadow inversion of $2$-dimensional unitary under Pauli-$Z$.}

\mainprop*

\begin{proof}
Recall that we have derived that the dual map of $\mathcal{N}_{U}$ satisfies the following equation
\begin{equation*}
    \mathcal{N}_{U}^{\dagger}(Z) = UZU^{\dagger}.
\end{equation*}
Then we define a new quantum channel $\mathcal{C}_{U}$:
\begin{equation*}
    \mathcal{C}_{U}(\sigma) := \mathcal{N}_{U}(U \sigma U^{\dagger})
\end{equation*}
and we have the following equation
\begin{equation*}
    \tr(\mathcal{N}_{U}(\rho)Z) = \tr(\mathcal{C}_{U}(U^{\dagger}\rho U)Z) = \tr(U^{\dagger}\rho UZ), \quad \forall \rho \in \mathsf{D}(
    \mathbb{C}^2).
\end{equation*}
Then denote $\sigma = U^{\dagger}\rho U$ and we can get
\begin{equation*}
    \tr(\mathcal{C}_{U}(\sigma)Z) = \tr(\sigma Z), \quad \forall \sigma \in \mathsf{D}(
    \mathbb{C}^2).
\end{equation*}
This is equivalent to 
\begin{equation*}
    \mathcal{C}_{U}^{\dagger}(Z) = Z
\end{equation*}
where $\mathcal{C}_{U}^{\dagger}$ is the dual map of $\mathcal{C}_{U}$. Now we consider the Kraus decomposition
\begin{equation*}
    \mathcal{C}_{U}(\sigma) = \sum_{k} A_{k,u} \sigma A_{k,u}^{\dagger}
\end{equation*}
with the condition
\begin{equation*}
    \sum_{k} A_{k,u}^{\dagger} A_{k,u} = I.
\end{equation*}
Except for this, we also have the following equation:
\begin{equation*}
    \mathcal{C}_{U}^{\dagger}(Z) = \sum_{k} A_{k,u}^{\dagger}ZA_{k,u} = Z.
\end{equation*}
In summary, we can get the Kraus decomposition of $\mathcal{N}_{U}$:
\begin{equation*}
    N_U(\rho) = \sum_k A_{k,u} U^\dagger \rho U A_{k,u}^\dagger,
\end{equation*}
where the Kraus operators $\{A_{k,u}\}_k$ satisfy:
\begin{equation*}
\begin{aligned}
    \sum_k A_{k,u}^\dagger A_{k,u} &= I, \\
    \sum_k A_{k,u}^\dagger Z A_{k,u} &= Z. 
\end{aligned}
\end{equation*}
\noindent
Now we will analyze the structure of $A_{k,u}$, suppose 
\begin{equation*}
A_{k,u} = 
\begin{pmatrix}
    a_{k,u} & b_{k,u} \\
    c_{k,u} & d_{k,u}
\end{pmatrix}.
\end{equation*}
Then
$
\sum_{k} A_{k,u}^{\dagger} A_{k,u} = I
$
is equivalent to
\begin{equation*}
\begin{cases}
    \text{(1)}~ \displaystyle \sum_{k} \left( |a_{k,u}|^{2} + |c_{k,u}|^{2} \right) = 1, \\[6pt]
    \text{(2)}~ \displaystyle \sum_{k} \left( |b_{k,u}|^{2} + |d_{k,u}|^{2} \right) = 1, \\[6pt]
    \text{(3)}~ \displaystyle \sum_{k} \left( a_{k,u}^{*} b_{k,u} + c_{k,u}^{*} d_{k,u} \right) = 0, \\[6pt]
    \text{(4)}~ \displaystyle \sum_{k} \left( b_{k,u}^{*} a_{k,u} + d_{k,u}^{*} c_{k,u} \right) = 0.
\end{cases}
\end{equation*}
\medskip
\noindent
Similarly, 
$
\sum_{k} A_{k,u}^{\dagger} Z A_{k,u} = Z
$
is equivalent to:
\begin{equation*}
\begin{cases}
    \text{(5)}~ \displaystyle \sum_{k} \left( |a_{k,u}|^{2} - |c_{k,u}|^{2} \right) = 1, \\[6pt]
    \text{(6)}~ \displaystyle \sum_{k} \left( |b_{k,u}|^{2} - |d_{k,u}|^{2} \right) = -1, \\[6pt]
    \text{(7)}~ \displaystyle \sum_{k} \left( a_{k,u}^{*} b_{k,u} - c_{k,u}^{*} d_{k,u} \right) = 0, \\[6pt]
    \text{(8)}~ \displaystyle \sum_{k} \left( b_{k,u}^{*} a_{k,u} - d_{k,u}^{*} c_{k,u} \right) = 0.
\end{cases}
\end{equation*}

\medskip
\noindent
From these conditions:
\begin{equation*}
\begin{aligned}
    (1) + (5) &\Rightarrow \sum_{k} |a_{k,u}|^{2} = 1, \\[3pt]
    (1) - (5) &\Rightarrow c_{k,u} = 0 \quad \text{for all } k, \\[3pt]
    (2) + (6) &\Rightarrow b_{k,u} = 0 \quad \text{for all } k, \\[3pt]
    (2) - (6) &\Rightarrow \sum_{k} |d_{k,u}|^{2} = 1.
\end{aligned}
\end{equation*}
Hence we have the following structure
\begin{equation*}
A_{k,u} = 
\begin{pmatrix}
    a_{k,u} & 0 \\
    0 & d_{k,u}
\end{pmatrix}
= 
\frac{a_{k,u} + d_{k,u}}{2} \, I +
\frac{a_{k,u} - d_{k,u}}{2} \, Z,
\end{equation*}
with
\begin{equation}\label{ad_ku}
\sum_{k} |a_{k,u}|^{2} = 
\sum_{k} |d_{k,u}|^{2} = 1.
\end{equation}
If we denote
\begin{equation*}
\alpha_{k,u} = \frac{a_{k,u} + d_{k,u}}{2}, 
\qquad
\beta_{k,u} = \frac{a_{k,u} - d_{k,u}}{2}
\end{equation*}
then from \eqref{ad_ku} we have
\begin{equation*}
\sum_k \left( |\alpha_{k,u} + \beta_{k,u}|^2 \right)
= \sum_k \left( |\alpha_{k,u} - \beta_{k,u}|^2 \right) = 1
\end{equation*}
which is equivalent to
\begin{equation*}
\begin{cases}
    \displaystyle \sum_k (|\alpha_{k,u}|^2 + |\beta_{k,u}|^2) = 1, \\[6pt]
    \displaystyle \sum_k \Re(\alpha_{k,u} \beta_{k,u}^*) = 0.
\end{cases}
\end{equation*}
Therefore, we can get the form of $\mathcal{N}_{U}$
\begin{equation*}
\mathcal{N}_U(\rho)
= \sum_k (\alpha_{k,u} I + \beta_{k,u} Z)
\, U^{\dagger} \rho U \,
(\alpha_{k,u}^* I + \beta_{k,u}^* Z).
\end{equation*}
Expanding gives
\begin{equation*}
\begin{aligned}
    \mathcal{N}_U(\rho)
    &= \sum_k |\alpha_{k,u}|^2 U\rho U^{\dagger}
    + \sum_k \alpha_{k,u} \beta_{k,u}^* \, U^{\dagger} \rho UZ \\[3pt]
    &\quad + \sum_k \beta_{k,u} \alpha_{k,u}^* ZU^{\dagger}\rho U
    + \sum_k |\beta_{k,u}|^2 Z U^{\dagger}\rho UZ.
\end{aligned}
\end{equation*}
Moreover, denote
\begin{equation*}
p(U) = \sum_k |\alpha_{k,u}|^2, \qquad
q(U) = \sum_k |\beta_{k,u}|^2, \qquad
r(U) = \sum_k \alpha_{k,u} \beta_{k,u}^*.
\end{equation*}
Then we can obtain the conclusion that
\begin{equation*}
\mathcal{N}_U(\rho)
= p(U) \, U^{\dagger}\rho U
+ r(U) \, U^{\dagger}\rho UZ
+ r(U)^{*} ZU^{\dagger}\rho U
+ q(U) \, ZU^{\dagger}\rho UZ
\end{equation*}
with the conditions
\begin{equation*}
\begin{cases}
    p(U),\, q(U) \ge 0,\\[4pt]
    p(U) + q(U) = 1,\\[4pt]
    \operatorname{Re}(r(U)) = 0 \ \Longleftrightarrow\ r(U) + r(U)^{*} = 0,\\[4pt]
    |r(U)|^2 \le p(U)\,q(U).
\end{cases}
\end{equation*}
The last inequality comes from the Cauchy–Schwarz inequality and hence complete the proof.
\end{proof}

\begin{corollary}\label{cor-eqlb}
    For any $U \in \opr{U}(2)$, the lower bound of the number of queries to achieve $\mathcal{N}_{U}$ in \eqref{eq:form-of-CPTP} is equivalent to the lower bound of the number of queries to achieve the specific CPTP map 
    \begin{equation*}
    \mathcal{M}_{U}(\rho) = \frac{1}{2}U^{\dagger}\rho U + \frac{1}{2} ZU^{\dagger} \rho UZ.
    \end{equation*}
    That is, $p(U) \equiv 1/2$ and $r(U) \equiv 0$ in \eqref{eq:form-of-CPTP}.
\end{corollary}
\begin{proof}
    If we can achieve $\mathcal{N}_{U}$ by $t$ queries to $U$, then we can also achieve 
    \begin{equation*}
        Z\mathcal{N}_{U}(\rho)Z = p(U) ZU^\dagger \rho UZ + (1 - p(U))  U^\dagger \rho U  -r(U) (U^\dagger \rho U Z - Z U^\dagger \rho U)
    \end{equation*}
    by $t$ queries to $U$ with appending $Z$ at the output stage of the circuit. Using the language of quantum comb, there exist $C_{1}$ and $C_{2}$ which are Choi operators of quantum comb and satisfy
    \begin{align*}
        C_{1} \ast |U\rangle\!\rangle\!\langle\!\langle U|^{\otimes t} &= \mathcal{N}_{U}\\
        C_{2} \ast |U\rangle\!\rangle\!\langle\!\langle U|^{\otimes t} &= Z\mathcal{N}_{U}Z.
    \end{align*}
    Next we construct 
    \begin{equation*}
        C = \frac{1}{2}(C_1 + C_2)
    \end{equation*}
    which satisfies that
    \begin{equation*}
        C \ast |U\rangle\!\rangle\!\langle\!\langle U|^{\otimes t} = \frac{1}{2}\mathcal{N}_{U} + \frac{1}{2} Z\mathcal{N}_{U}Z = \mathcal{M}_{U}.
    \end{equation*}
    Moreover, we know that the convex combination of Choi operators of quantum comb is also a Choi operator of some quantum comb. Hence we can achieve $\mathcal{M}_{U}$ by $t$ queries to $U$ which complete the proof.
\end{proof}
\begin{remark}
    Our numerical results in Table~\ref{tab:seq-vs-paral} suggest that the lower bound of the number of queries to achieve $\mathcal{M}_{U}$ for any $U \in \opr{U}(2)$ is $3$. This indicates an interesting phenomenon that implementing each CPTP map $U^{\dagger}(\cdot)U$ and $ZU^{\dagger}(\cdot)UZ$ deterministically and exactly requires $4$ queries \cite{Rever-unknown}, but their equal-probability mixture can be realized with only $3$ queries. 
\end{remark}

\section{A General Lower Bound on Query Complexity}
\label{sec:lb}
In this section, we will give a lower bound of the number of queries to achieve the shadow inversion of $d$-dimensional unitaries under $O$ where $d \in \mathbb{N^+}$ and $O$ is any fixed $d$-dimensional Hermitian observable. In order to analyze the robustness of the lower bound, we need to introduce the following definition.
\begin{definition}
    For any $d,t \in \mathbb{N}^+$ and $\epsilon > 0$, let $O$ be a $d$-dimensional observable. A quantum circuit $\mathcal{N}$ is said to be a \emph{$t$-query shadow inversion of $d$-dimensional unitaries under $O$ with error $\varepsilon$}, if we have
\begin{equation}\label{shadow-inversion-error}
\sup_{U \in U(d)}||\mathcal{N}_{U}^{\dagger}(O) - UOU^{\dagger}||_{1} \leq \varepsilon 
\end{equation}
where $\mathcal{N}_{U}$ is the output channel after $\mathcal{N}$ query $U$ exactly $t$ times and $\mathcal{N}_{U}^{\dagger}$ is the dual map of $\mathcal{N}_{U}$ in the sense of Hlibert-Schmidt inner product.
\end{definition}
\begin{remark}
    One can also choose a natural definition as
    \begin{equation*}
        \sup_{U \in U(d), \rho \in \mathsf{D}(d)}\Big|\tr\big[\bigl(U^{\dagger} \rho U - \mathcal{N}_{U}(\rho)\bigr)O\big]\Big| \leq \varepsilon
    \end{equation*}
    which is equivalent to 
    \begin{equation}\label{natural-def}
        \sup_{U \in U(d)}||\mathcal{N}_{U}^{\dagger}(O) - UOU^{\dagger}||_{\infty} \leq \varepsilon. 
    \end{equation}
    Due to the fact that $||\cdot||_{\infty} \leq ||\cdot||_{1}$, we can induce \eqref{natural-def} from \eqref{shadow-inversion-error}. 
\end{remark}

We will give a general lower bound for the shadow inversion of $d$-dimensional unitaries under $O$ with error $\varepsilon$ first, then we analyze the situation $\varepsilon = 0$ corresponding to Theorem~\ref{lb-thm}. 
\begin{proposition}\label{lb-err}
     For any $d,t \in \mathbb{N+}$ and $\varepsilon > 0$, if $\mathcal{N}$ is a $t$-query shadow inversion of $d$-dimensional unitaries under some fixed non-trivial $d$-dimensional Hermitian observable $O \not\propto I_{d}$ with error $\varepsilon$, then we have the following lower-bound
    \begin{equation}\label{lb-query-err}
        t \geq d \cdot \sup_{\lambda \in \mathbb{R}} \Bigg\{ \max\bigg\{\sup_{k \geq 1} \frac{||O_{\lambda}||_{k}^{k} - \epsilon ||O_{\lambda}||_{\infty}^{k-1}}{||O_{\lambda}||_{k-1}^{k-1}||O_{\lambda}||_{1}}, \max_{1 \leq l \leq r(\lambda)} \frac{\sum_{i=1}^{l}\sigma_{i}^{\lambda} - \epsilon}{l ||O_{\lambda}||_{1}}\bigg\} \Bigg\} - 1
    \end{equation}
     where $||\cdot||_{p}$ is the Schatten-$p$ norm for $p \in [1,\infty]$ and $O_{\lambda} = O - \lambda I_{d}$ for any $\lambda \in \mathbb{R}$. Moreover, $r(\lambda)$ is the rank of $O_{\lambda}$ and $\sigma_{1}^{\lambda} \geq \sigma_{2}^{\lambda} \geq \cdots \geq \sigma_{r(\lambda)}^{\lambda}> 0$ are positive singular values of $O_{\lambda}$. We also set $||O_{\lambda}||_{0}:= \sum_{i=1}^{r(\lambda)} (\sigma_{i}^{\lambda})^0 = r(\lambda)$.
\end{proposition}
Before we give the proof, we need to introduce the following result \cite[Theorem~3.7]{Chen2025tight} which shows that when $U$ is randomly sampled from the Haar measure and the number of queries $n$ is not too large, the overall channel must be very depolarizing (i.e., has low unitarity \cite{Unitarity-estimation}):
\begin{lemma}\label{str-lem}
    Let $d,n \in \mathbb{N}^+$ and $\mathcal{H}_i \cong \mathbb{C}^{d}, i=0,1,\cdots,2n+2$. Then for any $n$-slot sequential quantum comb $X$ on $(\mathcal{H}_{0},\mathcal{H}_{1},\cdots,\mathcal{H}_{2n+1})$, we have
    \begin{equation}\label{str-eq}
        X \ast \int_{U(d)} C_{U} d\mu_{H}(U) \preceq \frac{n+1}{d} \cdot I_{\mathcal{H}_0} \otimes I_{\mathcal{H}_{2n+2}}
    \end{equation}
    where $C_{U} \in \mathsf{L}(\mathcal{H}_{1} \otimes \mathcal{H}_{2} \otimes \cdots \otimes \mathcal{H}_{2n+2})$ is the Choi operator defined by 
    \[
    C_{U} := |U\rangle\rangle \langle\langle U|^{\otimes n+1}
    \]
    for any $U \in U(d)$ and $\mu_{H}$ denotes the Haar measure on $U(d)$.
\end{lemma}

\noindent Moreover, we also need the following lemma which will be useful in the proof of Proposition~\ref{lb-err}.

\begin{lemma}\label{link-eq}
Let $\mathcal{H}\cong\mathbb{C}^d$ with a fixed number $d \in \mathbb{N}^+$ and for a linear map $\mathcal{E}:\mathcal{L}(\mathcal{H})\to\mathcal{L}(\mathcal{H})$, denote $J_{\mathcal{E}}$ its Choi operator. Then, for all $A,B\in\mathcal{L}(\mathcal{H})$,
\begin{equation}\label{eq:link-EA}
\operatorname{Tr}\!\big[J_{\mathcal{E}}(A^{\mathsf T}\!\otimes B)\big] \;=\; \operatorname{Tr}\!\big[\mathcal{E}(A)\,B\big].
\end{equation}
Moreover, if $\mathcal{E}$ is Hermiticity-preserving, let $\mathcal{E}^\dagger$ denote the Hilbert-Schmidt adjoint of $\mathcal{E}$, defined by
\[
\operatorname{Tr}\!\big[X^\dagger \mathcal{E}(Y)\big] \;=\; \operatorname{Tr}\!\big[(\mathcal{E}^\dagger(X))^\dagger Y\big]\quad\text{for all }X,Y.
\]
Then \eqref{eq:link-EA} is equivalently written as
\begin{equation}\label{eq:link-adjoint}
\operatorname{Tr}\!\big[J_{\mathcal{E}}(A^{\mathsf T}\!\otimes B)\big] \;=\; \operatorname{Tr}\!\big[A\,\mathcal{E}^\dagger(B)\big].
\end{equation}
\end{lemma}

\begin{proof}
We first show \eqref{eq:link-EA} directly from the definitions. Denote $E_{ij}:=\lvert i\rangle\langle j\rvert$ for the matrix units. Using the Choi expansion and the multiplicativity of the trace over tensor products,
\[
\begin{aligned}
\operatorname{Tr}\!\big[J_{\mathcal{E}}(A^{\mathsf T}\!\otimes B)\big]
&= \sum_{i,j=1}^d \operatorname{Tr}\!\big[(E_{ij}\otimes \mathcal{E}(E_{ij}))(A^{\mathsf T}\!\otimes B)\big] \\
&= \sum_{i,j=1}^d \operatorname{Tr}(E_{ij}A^{\mathsf T})\,\operatorname{Tr}(\mathcal{E}(E_{ij})B).
\end{aligned}
\]
By calculation, $\operatorname{Tr}(E_{ij}A^{\mathsf T})=A_{ij}$ and $\mathcal{E}(A)=\sum_{i,j} A_{ij}\,\mathcal{E}(E_{ij})$.
Hence
\[
\operatorname{Tr}\!\big[J_{\mathcal{E}}(A^{\mathsf T}\!\otimes B)\big]
= \sum_{i,j = 1}^{d} A_{ij}\,\operatorname{Tr}(\mathcal{E}(E_{ij})B)
= \operatorname{Tr}\!\Big[\Big(\sum_{i,j = 1}^{d} A_{ij}\,\mathcal{E}(E_{ij})\Big) B\Big]
= \operatorname{Tr}\!\big[\mathcal{E}(A)\,B\big]
\]
which proves \eqref{eq:link-EA}. To obtain \eqref{eq:link-adjoint} under the extra assumption that $\mathcal{E}$ is Hermiticity-preserving, we use the Hilbert-Schmidt adjointness together with the cyclicity of the trace. Specifically, for any $A,B \in \mathcal{L}(\mathcal{H})$
\[
\operatorname{Tr}\!\big[\mathcal{E}(A)\,B\big]
= \operatorname{Tr}\!\big[B\,\mathcal{E}(A)\big]
= \operatorname{Tr}\!\big[(\mathcal{E}^\dagger(B^\dagger))^\dagger A\big]
= \operatorname{Tr}\!\big[A\,\mathcal{E}^\dagger(B)\big],
\]
where the last equality uses that $\mathcal{E}^\dagger$ is Hermiticity-preserving which can be obtained from $\mathcal{E}$ is Hermiticity-preserving and hence $(\mathcal{E}^\dagger(B^\dagger))^\dagger=\mathcal{E}^\dagger(B)$.
Combining this with \eqref{eq:link-EA} yields \eqref{eq:link-adjoint}.
\end{proof}

\begin{lemma}\label{key-lem}
    Let $d \in \mathbb{N}^{+}$ and $\mathcal{H}\cong\mathbb{C}^d$, assume that $\mathcal{E}:\mathcal{L}(\mathcal{H})\to\mathcal{L}(\mathcal{H})$ is a Hermiticity-preserving linear map. If for a fixed $d$-dimensional Hermitian operator $O$, we have
    \begin{equation*}
        ||\mathcal{E}^{\dagger}(O) - O||_{1} \leq \epsilon
    \end{equation*}
    for some $\epsilon > 0$ and the Choi operator of $\mathcal{E}$ satisfies 
    \begin{equation*}
        J_{E} \preceq \alpha \cdot  I_{\mathcal{H}} \otimes I_{\mathcal{H}}
    \end{equation*}
    for some $\alpha > 0$, then we will get the following inequality
    \begin{equation*}
        \alpha \geq \max\Bigg\{\sup_{k \geq 1} \frac{||O||_{k}^{k} - \epsilon ||O||_{\infty}^{k-1}}{||O||_{k-1}^{k-1}||O||_{1}}, \max_{1 \leq l \leq r} \frac{\sum_{i=1}^{l}\sigma_{i} - \epsilon}{l ||O||_{1}}\Bigg\}
    \end{equation*}
    where $||\cdot||_{p}$ is the Schatten-$p$ norm for $p \in [1,\infty]$, $r$ is the rank of $O$ and $\sigma_{1} \geq \sigma_{2} \geq \cdots \geq \sigma_{r}> 0$ are positive singular values of $O$. We also set $||O||_{0}:= \sum_{i=1}^{r} \sigma_{i}^0 = r$.
\end{lemma}

\begin{proof}
    First we take $B = O$ in \eqref{eq:link-adjoint} and get
    \begin{equation*}
        \tr[J_{E}(A^{\intercal} \otimes O)] = \tr[A\mathcal{E}^{\dagger}(O)]
    \end{equation*}
    for all $A \in \mathsf{L}(\mathcal{H})$. Then we have
    \begin{equation*}
        |\tr[J_{E}(A^{\intercal} \otimes O)]| \leq \alpha ||A^{\intercal} \otimes O||_{1} = \alpha ||A||_{1}||O||_{1}.
    \end{equation*}
    On the other side, we denote $\Delta = \mathcal{E}^{\dagger}(O) - O$ and we can get
    \begin{equation*}
    \begin{aligned}
       \big|\tr[A\mathcal{E}^{\dagger}(O)]\big| &= \big|\tr(AO)+\tr(A\Delta)\big|\\
       &\geq \big|\tr(AO)\big| - \big|\tr(A\Delta)\big|\\
       &\geq \big|\tr(AO)\big| - ||A||_{\infty}||\Delta||_{1}
    \end{aligned}
    \end{equation*}
    where we use the Hölder's inequality on the last step. Therefore, we can get
    \begin{equation}\label{eq:master}
    \alpha \geq \frac{\big|\tr(AO)\big| - \varepsilon ||A||_{\infty}}{||A||_{1}||O||_{1}}
    \end{equation}
    for any non-zero $A \in \mathsf{L}(\mathcal{H})$. Next, we will choose $A$ to estimate the lower bound. Suppose we have the following spectral decomposition for $O$:
    \begin{equation*}
        O = \sum_{i=1}^{r} \lambda_{i} |v_{i}\rangle \langle v_{i}|
    \end{equation*}
    with $|\lambda_{i}| = \sigma_{i}$. Then we denote
    \begin{equation*}
        |O| := \sum_{i=1}^{r}|\lambda_{i}| |v_{i}\rangle\langle v_{i}|
    \end{equation*}
    and
    \begin{equation*}
        \operatorname{sign}(O) := \sum_{i = 1}^{r} \operatorname{sign}(\lambda_{i})|v_{i}\rangle \langle v_i|
    \end{equation*}
    where
    \begin{equation*}
        \operatorname{sign}(\lambda) := 
        \begin{cases}
            +1, & \lambda > 0, \\
            0,  & \lambda = 0, \\
            -1, & \lambda < 0.
        \end{cases}
    \end{equation*}
    Then we take $A_{k} = |O|^{k-1} \operatorname{sign}(O), \, k \geq 1$ where we set $|O|^{0} = I_{d}$. We can compute
    \begin{equation*}
        \tr(A_{k}O) = \tr(|O|^{k}) = \sum_{i = 1}^{r} \sigma_{i}^{k} = ||O||_{k}^{k}.
    \end{equation*}
    Moreover, we have
    \begin{equation*}
        ||A_k||_{1} = \tr(|O|^{k-1}) = \sum_{i=1}^{r}\sigma_{i}^{k-1} = ||O||_{k-1}^{k-1}
    \end{equation*}
    and
    \begin{equation*}
        ||A_{k}||_{\infty} = ||O||_{\infty}^{k-1}.
    \end{equation*}
    Hence from \eqref{eq:master} we can get for any $k \geq 1$, we have
    \begin{equation}\label{eq:1-lb}
        \alpha \geq \frac{||O||_{k}^{k} - \epsilon ||O||_{\infty}^{k-1}}{||O||_{k-1}^{k-1}||O||_{1}}.
    \end{equation}
    On the other hand, For $l \in\{1,\dots,r\}$, we define
    \begin{equation*}
    A^{(\ell)} := \sum_{i=1}^{\ell} \operatorname{sign}(\lambda_i)\,\lvert v_i\rangle\!\langle v_i\rvert.
    \end{equation*}
    Then
    \begin{equation*}
    \tr\big(A^{(l)} O\big) = \sum_{i=1}^{l} |\lambda_i|, \qquad
    ||A^{(l)}||_{1} = l, \qquad
    ||A^{(l)}||_{\infty} = 1.
    \end{equation*}
    Applying \eqref{eq:master} with $A^{(l)}$ gives that for any $1 \leq l \leq r$, we have
    \begin{equation}\label{eq:2-lb}
        \alpha \geq \frac{\sum_{i=1}^{l}\sigma_{i} - \epsilon}{l ||O||_{1}}.
    \end{equation}
    Combine the equations \eqref{eq:1-lb} and \eqref{eq:2-lb}, we complete the proof.
\end{proof}
Now we will give the proof of Proposition~\ref{lb-err}.2
\begin{proof}[Proof of Proposition~\ref{lb-err}]
    Denote $\mathcal{N}_{U}$ the channel that $\mathcal{N}$ queries $U$ for $t$ times and $\mathcal{M}_{U} := \mathcal{N}_U\circ U$, then the Choi operator of $\mathcal{M}_{U}$ is
    \begin{equation*}
    \begin{aligned}
        J_{\mathcal{M}_{U}} = J_{\mathcal{N}_{U}} \ast J_{U} &= (N \ast |U\rangle\rangle \langle\langle U|^{\otimes t}) \ast  |U\rangle\rangle \langle\langle U|\\
        &= N \ast |U\rangle\rangle \langle\langle U|^{\otimes t+1}
    \end{aligned}
    \end{equation*}
    where $N$ is the Choi operator of $\mathcal{N}$. From the equation \eqref{shadow-inversion-error} we know that the dual map of $\mathcal{N}_{U}$ satisfies that 
    \begin{equation}\label{eq:eps-1}
    \sup_{U \in U(d)}||\mathcal{N}_{U}^{\dagger}(O) - UOU^{\dagger}||_{1} \leq \varepsilon.
    \end{equation}
    Note that
    \begin{align*}
        \mathcal{M}_{U}^{\dagger}(O) - O &= U^{\dagger} \mathcal{N}_{U}^{\dagger}(O)U - O\\
        &= U^{\dagger} \mathcal{N}_{U}^{\dagger}(O)U - U^{\dagger}(UOU^{\dagger})U\\
        &= U^{\dagger}(\mathcal{N}_{U}^{\dagger}(O) - UOU^{\dagger})U
    \end{align*}
    and hence the equation \eqref{eq:eps-1} is equivalent to
    \begin{equation*}
        \sup_{U \in U(d)}||\mathcal{M}_{U}^{\dagger}(O) - O||_{1} \leq \varepsilon.
    \end{equation*} 
    Then we consider the following map
    \begin{equation*}
        \mathcal{M}(A) := \int_{U(d)} \mathcal{M}_{U}(A) d\mu_{H}(U) 
    \end{equation*}
    where $\mu_{H}$ is the Haar measure on $U(d)$. Then we can verify that
    \begin{align*}
        ||\mathcal{M}^{\dagger}(O) - O||_{1} &= \Bigg\|\int_{U(d)} \left(\mathcal{M}_{U}^{\dagger}(O)-O\right) d\mu_{H}(U)\Bigg\|_{1}\\
        &\leq \int_{U(d)} \Big\|\mathcal{M}_{U}^{\dagger}(O)-O \Big\|_{1} d\mu_{H}(U)\\
        &\leq \int_{U(d)} \varepsilon d\mu_{H}(U) = \varepsilon.
    \end{align*}
    Moreover, the Choi operator of $\mathcal{M}$ is
    \begin{equation*}
        J_{\mathcal{M}} = \int_{U(d)} J_{\mathcal{M}_{U}} d\mu_{H}(U) = N \ast \int_{U(d)} |U\rangle\rangle\langle\langle U|^{\otimes t+1} d\mu_{H}(U).
    \end{equation*}
    By Lemma~\ref{str-lem} we can get
    \begin{equation*}
        J_{\mathcal{M}} \preceq \frac{t+1}{d} \cdot I_{\mathcal{H}_0} \otimes I_{\mathcal{H}_{2n+2}}.
    \end{equation*}
     Therefore, by Lemma~\ref{key-lem} where we use $\alpha = (t+1)/d$ we can get
     \begin{equation}\label{eq:before-lambda}
         t \geq d \cdot \max\Bigg\{\sup_{k \geq 1} \frac{||O||_{k}^{k} - \epsilon ||O||_{\infty}^{k-1}}{||O||_{k-1}^{k-1}||O||_{1}}, \max_{1 \leq l \leq r} \frac{\sum_{i=1}^{l}\sigma_{i} - \epsilon}{l ||O||_{1}}\Bigg\} - 1
     \end{equation}
     where $\sigma_{1} \geq \sigma_{2} \geq \cdots \geq \sigma_{r}>0$ are positive singular values of $O$ and $r$ is the rank of $O$. Moreover, $\mathcal{M}^{\dagger}$ is a unital map since that $\mathcal{M}_{U}^{\dagger}$ is unital for any $U \in U(d)$. Then for any $\lambda \in \mathbb{R}$, we have
     \begin{equation*}
         \mathcal{M}^{\dagger}(O - \lambda I_{d}) = O - \lambda I_{d}.
     \end{equation*}
     Hence we can apply $O_{\lambda} = O - \lambda I_{d}$ to \eqref{eq:before-lambda} for any $\lambda \in \mathbb{R}$ and we can get 
     \begin{equation*}
        t \geq d \cdot \sup_{\lambda \in \mathbb{R}} \Bigg\{ \max\bigg\{\sup_{k \geq 1} \frac{||O_{\lambda}||_{k}^{k} - \epsilon ||O_{\lambda}||_{\infty}^{k-1}}{||O_{\lambda}||_{k-1}^{k-1}||O_{\lambda}||_{1}}, \max_{1 \leq l \leq r(\lambda)} \frac{\sum_{i=1}^{l}\sigma_{i}^{\lambda} - \epsilon}{l ||O_{\lambda}||_{1}}\bigg\} \Bigg\} - 1
    \end{equation*}
    which completes the proof.
    \end{proof}
    Now we will restate and give the proof of Theorem~\ref{lb-thm} based on the result of Proposition~\ref{lb-err}.
    \lbthm*
    \begin{proof}[Proof of Theorem~\ref{lb-thm}]
    Take $\varepsilon = 0$ in equation \eqref{lb-query-err} we can get
    \begin{equation*}
         t \geq d \cdot \sup_{\lambda \in \mathbb{R}} \Bigg\{ \max\bigg\{\sup_{k \geq 1} \frac{||O_{\lambda}||_{k}^{k}}{||O_{\lambda}||_{k-1}^{k-1}||O_{\lambda}||_{1}}, \max_{1 \leq l \leq r(\lambda)} \frac{\sum_{i=1}^{l}\sigma_{i}^{\lambda}}{l ||O_{\lambda}||_{1}}\bigg\} \Bigg\} - 1
    \end{equation*}
     where $\sigma_{1}^{\lambda} \geq \sigma_{2}^{\lambda} \geq \cdots \geq \sigma_{r}^{\lambda}>0$ are positive singular values of $O_{\lambda}$ and $r(\lambda)$ is the rank of $O_{\lambda}$. Then we denote
     \begin{equation*}
         \mu_{l}^{\lambda} = \frac{1}{l} \sum_{i = 1}^{l}\sigma_{i}^{\lambda}, \, l = 1,\cdots,r(\lambda).
     \end{equation*}
     Note that we have
     \begin{equation*}
         (l+1)\mu_{l+1}^{\lambda} = \sum_{i=1}^{l}\sigma_{i}^{\lambda} + \sigma_{l+1}^{\lambda} \leq \sum_{i=1}^{l}\sigma_{i}^{\lambda} + \sigma_{l}^{\lambda} \leq l\mu_{l}^{\lambda} + \mu_{l}^{\lambda} = (l+1)\mu_{l}^{\lambda}
     \end{equation*}
     which means that $\mu_{l+1}^{\lambda} \leq \mu_{l}^{\lambda}$. Hence
     \begin{equation*}
         \max_{1 \leq l \leq r(\lambda)} \mu_{l}^{\lambda} = \mu_{1}^{\lambda} = \sigma_{1}^{\lambda} = ||O_{\lambda}||_{\infty}
     \end{equation*}
     and then we will get
     \begin{equation}\label{eq:mono-1}
         \max_{1 \leq l \leq r(\lambda)} \frac{\sum_{i=1}^{l}\sigma_{i}^{\lambda}}{l ||O_{\lambda}||_{1}} = \frac{||O_{\lambda}||_{\infty}}{||O_{\lambda}||_{1}}.
     \end{equation}
     Moreover, we define
     \begin{equation*}
         S_{s}^{\lambda} = \sum_{i=1}^{r(\lambda)} (\sigma_{i}^{\lambda})^{s}, \, s>0
     \end{equation*}
     and then we have
     \begin{equation*}
         f_k^{\lambda}:= \frac{||O_{\lambda}||_{k}^{k}}{||O_{\lambda}||_{k-1}^{k-1}||O_{\lambda}||_{1}} = \frac{S_k^{\lambda}}{S_{k-1}^{\lambda}S_1^{\lambda}}.
     \end{equation*}
     We define the function
     \begin{equation*}
         \phi^{\lambda}(s) := \ln\left(\sum_{i=1}^{r(\lambda)} (\sigma_{i}^{\lambda})^{s}\right), \, s>0.
     \end{equation*}
     We can compute that 
     \begin{equation*}
         (S_{s}^{\lambda})' = \sum_{i=1}^{r(\lambda)} \frac{d}{ds}(\sigma_{i}^{\lambda})^s = \sum_{i=1}^{r(\lambda)} (\sigma_{i}^{\lambda})^{s} \ln\sigma_i^{\lambda}
     \end{equation*}
     and hence
     \begin{equation*}
         (\phi^{\lambda})'(s) = \frac{(S_{s}^{\lambda})'}{S_s^{\lambda}} = \frac{\sum_{i=1}^{r(\lambda)} (\sigma_{i}^{\lambda})^{s}\ln \sigma_i^{\lambda}}{\sum_{j=1}^{r(\lambda)} (\sigma_{j}^{\lambda})^{s}}.
     \end{equation*}
     We define a probability measure related to $s$ on the set $\{1,2,\cdots,r(\lambda)\}$ which we denote as $p^{\lambda}(s)$:
     \begin{equation*}
         p_{i}^{\lambda}(s) := \frac{(\sigma_{i}^{\lambda})^{s}}{S_s^{\lambda}}, \quad p_{i}^{\lambda}(s) \geq 0, \quad \sum_{i=1}^{r(\lambda)}p_{i}^{\lambda}(s) = 1.
     \end{equation*}
     Then we have
     \begin{equation*}
         (\phi^{\lambda})'(s) = \sum_{i=1}^{r(\lambda)} p_{i}^{\lambda}(s) \ln \sigma_{i}^{\lambda} = \mathbb{E}_{p^{\lambda}(s)}[\ln \sigma^{\lambda}].
     \end{equation*}
     Moreover, we can get that
     \begin{equation*}
         (p_{i}^{\lambda})'(s) = \frac{(\sigma_{i}^{\lambda})^{s} \ln \sigma_{i}^{\lambda} \cdot S_s^{\lambda} - (\sigma_{i}^{\lambda})^{s} \cdot (S_{s}^{\lambda})'}{(S_{s}^{\lambda})^2} = p_{i}^{\lambda}(s)\bigl(\ln \sigma_{i}^{\lambda} - (\phi^{\lambda})'(s) \bigr).
     \end{equation*}
     Hence we can compute that
     \begin{equation*}
         \begin{aligned}
             (\phi^{\lambda})''(s) &= \frac{d}{ds}\bigl(\sum_{i=1}^{r(\lambda)}p_{i}^{\lambda}(s)\ln \sigma_i^{\lambda}\bigr) = \sum_{i=1}^{r(\lambda)} (p_{i}^{\lambda})'(s) \ln \sigma_i^{\lambda}\\
             &= \sum_{i=1}^{r(\lambda)} p_{i}^{\lambda}(s)\bigl(\ln \sigma_{i}^{\lambda} - (\phi^{\lambda})'(s) \bigr)\ln \sigma_i^{\lambda}\\
             &= \sum_{i=1}^{r(\lambda)} p_{i}^{\lambda}(s)(\ln \sigma_i^{\lambda})^2 - (\phi^{\lambda})'(s)\sum_{i=1}^{r(\lambda)}p_{i}^{\lambda}(s)\ln \sigma_i^{\lambda}\\
             &= \sum_{i=1}^{r(\lambda)} p_{i}^{\lambda}(s)(\ln \sigma_i^{\lambda})^2 - ((\phi^{\lambda})'(s))^2\\
             &= \mathbb{E}_{p^{\lambda}(s)}[(\ln \sigma^{\lambda})^2] - \bigl(\mathbb{E}_{p^{\lambda}(s)}[\ln \sigma^{\lambda}]\bigr)^2 = \operatorname{Var}_{p^{\lambda}(s)}(\ln \sigma^{\lambda}) \geq 0.
         \end{aligned}
     \end{equation*}
     Therefore, $\phi^{\lambda}(s)$ is convex and then for any integer $k \geq 2$, we have
     \begin{equation*}
         2\phi^{\lambda}(k) \leq \phi^{\lambda}(k-1) + \phi^{\lambda}(k+1)
     \end{equation*}
     which is equivalent to 
     \begin{equation*}
         \frac{S_{k+1}^{\lambda}}{S_{k}^{\lambda}} \geq \frac{S_k^{\lambda}}{S_{k-1}^{\lambda}}.
     \end{equation*}
     Hence we can derive that $f_k^{\lambda}$ is a non-decreasing sequence and
     \begin{equation*}
         \lim_{k \to \infty} f_k^{\lambda} = \lim_{k \to \infty} \frac{M^{\lambda}(\sigma_{1}^{\lambda})^k}{(M^{\lambda}(\sigma_{1}^{\lambda})^{k-1})S_1^{\lambda}} = \frac{\sigma_{1}^{\lambda}}{S_1^{\lambda}} = \frac{||O_{\lambda}||_{\infty}}{||O_{\lambda}||_{1}}
     \end{equation*}
     where $M^{\lambda}$ is the multiplicity of $\sigma_1^{\lambda}$. Hence we can get
     \begin{equation}\label{eq:mono-2}
         \sup_{k \geq 1} \frac{||O_{\lambda}||_{k}^{k}}{||O_{\lambda}||_{k-1}^{k-1}||O_{\lambda}||_{1}} = \frac{||O_{\lambda}|||_{\infty}}{||O_{\lambda}||_{1}}.
     \end{equation}
     From \eqref{eq:mono-1} and \eqref{eq:mono-2} we can obtain
     \begin{equation}\label{eq:guodu}
         t \geq d \cdot \sup_{\lambda \in \mathbb{R}}\frac{||O - \lambda I_{d}||_{\infty}}{||O - \lambda I_{d}||_{1}} - 1.
     \end{equation}
     Together with the following Lemma~\ref{lem-c-d}, we complete the proof.
     \end{proof}
     \begin{lemma}\label{lem-c-d}
         For any $d \in \mathbb{N}^{+}$, suppose $O \not\propto I_{d}$ is a fixed non-trivial $d$-dimensional Hermitian observable, then
         \begin{equation}\label{sup-max}
        \sup_{\lambda \in \mathbb{R}}\frac{||O - \lambda I_{d}||_{\infty}}{||O - \lambda I_{d}||_{1}} =  \max_{k = 1, \cdots, d}\Big\{ \frac{\max_{1 \leq i \leq d}|\lambda_{i} - \lambda_{k}|}{\sum_{i=1}^{d}|\lambda_{i} - \lambda_{k}|} \Big\}
         \end{equation}
         where $\lambda_{1} \geq \cdots \geq \lambda_{d}$ are the eigenvalues of $O$.
     \end{lemma}
     \begin{proof}
     Since $O$ is Hermitian, it is unitarily diagonalizable. As both $\|\cdot\|_{\infty}$ and $\|\cdot\|_{1}$ are unitarily invariant, we may, without loss of generality, assume that
\[
O = \operatorname{diag}(\lambda_1,\ldots,\lambda_d),
\]
with $\lambda_1 \ge \cdots \ge \lambda_d$. For each $\lambda \in \mathbb{R}$, the eigenvalues of $O - \lambda I_d$ are $\lambda_i - \lambda$, hence
\[
\|O - \lambda I_d\|_\infty = \max_{1 \le i \le d} |\lambda_i - \lambda|,
\qquad
\|O - \lambda I_d\|_1 = \sum_{i=1}^d |\lambda_i - \lambda|.
\]
If we define
\[
f(\lambda) := \frac{\|O - \lambda I_d\|_\infty}{\|O - \lambda I_d\|_1}
= \frac{\max_{1 \le i \le d} |\lambda_i - \lambda|}{\sum_{i=1}^d |\lambda_i - \lambda|},
\]
then the desired identity \eqref{sup-max} is equivalent to
\[
\sup_{\lambda \in \mathbb{R}} f(\lambda)
= \max_{k = 1,\ldots,d} f(\lambda_k).
\]
Observe first that $f$ is invariant under a simultaneous affine change of variables:
for any $a \in \mathbb{R}$ and any $c>0$, if we replace all eigenvalues $\lambda_i$ by
$\tilde{\lambda}_i := c \lambda_i + a$ and simultaneously replace $\lambda$ by
$\tilde{\lambda} := c \lambda + a$, then
\[
|\tilde{\lambda}_i - \tilde{\lambda}| = c|\lambda_i - \lambda|,
\qquad
\sum_i |\tilde{\lambda}_i - \tilde{\lambda}| = c \sum_i |\lambda_i - \lambda|,
\]
so the ratio $f(\lambda)$ is unchanged. Hence we may assume without loss of generality that
\[
\lambda_1 = 1,\qquad \lambda_d = -1.
\]
We now analyze $f$ on the real line.

\smallskip
\emph{1. Outside the spectral interval.}
For $\lambda \ge 1$ we have $\lambda_i \le 1 \le \lambda$ for all $i$, so
\[
|\lambda_i - \lambda| = \lambda - \lambda_i,\quad
\|O-\lambda I_d\|_\infty = \lambda - \lambda_d = \lambda+1,\quad
\|O-\lambda I_d\|_1 = d\lambda - \sum_{i=1}^d \lambda_i.
\]
Thus, we can get
\[
f(\lambda) = \frac{\lambda+1}{d\lambda - \sum_i \lambda_i}.
\]
A direct derivative computation shows
\[
f'(\lambda)
= \frac{d\lambda - \sum_i \lambda_i - d(\lambda+1)}{(d\lambda - \sum_i \lambda_i)^2}
= \frac{-d - \sum_i\lambda_i}{(d\lambda - \sum_i \lambda_i)^2} < 0,
\]
Hence $f$ is strictly decreasing on $[1,\infty)$, so
\[
\sup_{\lambda \ge 1} f(\lambda) = f(1) = f(\lambda_1).
\]
Similarly, for $\lambda \le -1$, one checks that $f$ is strictly increasing on $(-\infty,-1]$, so
\[
\sup_{\lambda \le -1} f(\lambda) = f(-1) = f(\lambda_d).
\]

\smallskip
\emph{2. Between the extreme eigenvalues.}
It remains to consider $\lambda \in [-1,1]$. For each $k = 1,\ldots,d-1$ and
$\lambda \in (\lambda_{k+1},\lambda_k)$, exactly $k$ eigenvalues satisfy
$\lambda_i \ge \lambda$ and $d-k$ satisfy $\lambda_i < \lambda$. Thus
\[
\sum_{i=1}^d |\lambda_i - \lambda|
= \sum_{i=1}^k (\lambda_i - \lambda) + \sum_{i=k+1}^d (\lambda - \lambda_i)
= A_k + B_k \lambda,
\]
for suitable real constants $A_k,B_k$ (depending only on $\lambda_1,\dots,\lambda_d$, not on $\lambda$).

On the other hand, for any $\lambda \in [-1,1]$, the largest distance to the set
$\{\lambda_1,\dots,\lambda_d\}\subset[-1,1]$ is always attained at one of the extreme points
$\lambda_1 = 1$ or $\lambda_d = -1$, so
\[
\max_{1 \le i \le d} |\lambda_i - \lambda|
= \max\{1-\lambda,\,1+\lambda\}
=
\begin{cases}
1 - \lambda, & \lambda \le 0,\\[2pt]
1 + \lambda, & \lambda \ge 0.
\end{cases}
\]
Therefore, on each interval $(\lambda_{k+1},\lambda_k)$ that does not contain $0$, both the numerator and the denominator of $f(\lambda)$ are affine functions of $\lambda$, so
\[
f(\lambda) = \frac{\alpha + \beta \lambda}{A_k + B_k \lambda}
\]
for suitable constants $\alpha,\beta$. The derivative of this fractional linear function has a fixed sign on the interval, hence $f$ is monotone on that interval. Consequently, $f$ cannot attain a local maximum in the interior of any such interval; its maximum on $(\lambda_{k+1},\lambda_k)$ is attained at one of the endpoints $\lambda_k$ or $\lambda_{k+1}$.

\smallskip
\emph{3. The interval containing $0$.}
There is at most one index $p\in\{1,\ldots,d-1\}$ such that
\[
\lambda_{p+1} \le 0 \le \lambda_p.
\]
If $0$ coincides with some eigenvalue $\lambda_k$, then $0$ is already among the points
$\{\lambda_1,\ldots,\lambda_d\}$ where we are taking the maximum, and there is nothing to prove. Otherwise, we have a strict inequality $\lambda_{p+1}<0<\lambda_p$, and we must check that $f(0)$ is not strictly larger than both $f(\lambda_{p+1})$ and $f(\lambda_p)$. On the interval $[\lambda_{p+1},\lambda_p]$, the sign pattern of $\lambda_i - \lambda$ is constant: for $i\le p$ we have $\lambda_i \ge \lambda$, and for $i>p$ we have $\lambda_i \le \lambda$. Hence
\[
\sum_{i=1}^d |\lambda_i - \lambda|
= \sum_{i\le p} (\lambda_i - \lambda) + \sum_{i>p} (\lambda - \lambda_i)
= A + B\lambda,
\]
where
\[
A := \sum_{i\le p}\lambda_i - \sum_{i>p}\lambda_i
= \sum_{i=1}^d |\lambda_i| > 0,\qquad
B := d - 2p.
\]
Moreover, as above,
\[
\max_{i}|\lambda_i - \lambda|
=
\begin{cases}
1 - \lambda, & \lambda \in [\lambda_{p+1},0],\\[2pt]
1 + \lambda, & \lambda \in [0,\lambda_p].
\end{cases}
\]
Thus
\[
f(\lambda)
=
\begin{cases}
\dfrac{1 - \lambda}{A + B\lambda}, & \lambda \in [\lambda_{p+1},0],\\[6pt]
\dfrac{1 + \lambda}{A + B\lambda}, & \lambda \in [0,\lambda_p].
\end{cases}
\]
Recall that
\[
A = \sum_{i=1}^d |\lambda_i| > 0, \qquad B = d - 2p \in \mathbb{Z},
\]
and, by construction,
\[
A + B\lambda = \sum_{i=1}^d |\lambda_i - \lambda| > 0
\quad\text{for all } \lambda \in [\lambda_{p+1},\lambda_p].
\]
We now show in detail that
\[
f(0) \le \max\{f(\lambda_{p+1}),\, f(\lambda_p)\}.
\]
To this end, suppose for contradiction that $f(0)$ is strictly larger than both
endpoint values, i.e.
\begin{equation}\label{eq:contr-assumption}
f(0) > f(\lambda_{p+1})
\quad\text{and}\quad
f(0) > f(\lambda_p).
\end{equation}
Using the explicit formulas,
\[
f(0) = \frac{1}{A}, \qquad
f(\lambda_{p+1}) = \frac{1-\lambda_{p+1}}{A + B\lambda_{p+1}}, \qquad
f(\lambda_p) = \frac{1+\lambda_p}{A + B\lambda_p},
\]
we can rewrite the inequalities in \eqref{eq:contr-assumption} as algebraic
conditions on $A$ and $B$.

\medskip
\noindent\textbf{Left endpoint.}
From $f(0) > f(\lambda_{p+1})$ we obtain
\[
\frac{1}{A}
>
\frac{1-\lambda_{p+1}}{A + B\lambda_{p+1}}.
\]
Both denominators are positive, since $A>0$ and
$A+B\lambda_{p+1} = \sum_i |\lambda_i-\lambda_{p+1}|>0$, so we may
cross-multiply without changing the inequality sign:
\[
A(1 - \lambda_{p+1}) < A + B\lambda_{p+1}.
\]
Expanding and cancelling $A$ yields
\[
-A\lambda_{p+1} < B\lambda_{p+1}
\quad\Longleftrightarrow\quad
(-A - B)\lambda_{p+1} < 0
\quad\Longleftrightarrow\quad
(A+B)\lambda_{p+1} > 0.
\]
Since $\lambda_{p+1} < 0$, this implies
\begin{equation}\label{eq:AplusB-neg}
A + B < 0.
\end{equation}

\medskip
\noindent\textbf{Right endpoint.}
Similarly, from $f(0) > f(\lambda_p)$ we get
\[
\frac{1}{A}
>
\frac{1+\lambda_p}{A + B\lambda_p},
\]
and again $A+B\lambda_p = \sum_i |\lambda_i-\lambda_p|>0$, so
\[
A(1 + \lambda_p) < A + B\lambda_p.
\]
Expanding and cancelling $A$ gives
\[
A\lambda_p < B\lambda_p
\quad\Longleftrightarrow\quad
(A-B)\lambda_p < 0.
\]
Because $\lambda_p > 0$, this inequality is equivalent to
\begin{equation}\label{eq:AminusB-neg}
A - B < 0.
\end{equation}

\medskip
\noindent\textbf{Deriving the contradiction.}
From \eqref{eq:AplusB-neg} we have $A+B<0$, hence
\[
B < -A < 0,
\]
so in particular $B$ is strictly negative. On the other hand,
\eqref{eq:AminusB-neg} says $A-B<0$, i.e.
\[
A < B.
\]
Combining these two inequalities yields
\[
A < B < 0,
\]
which contradicts $A = \sum_{i=1}^d |\lambda_i| > 0$. This contradiction shows that our assumption \eqref{eq:contr-assumption}
was false. Therefore $f(0)$ cannot be strictly larger than both
$f(\lambda_{p+1})$ and $f(\lambda_p)$, and we must have
\[
f(0) \le \max\{f(\lambda_{p+1}),\, f(\lambda_p)\}.
\]
In particular, also on the interval containing $0$, the maximum of $f$ is
attained at one of the eigenvalues $\lambda_{p+1}$ or $\lambda_p$. Collecting the above cases, we conclude that the global supremum of $f$ over $\mathbb{R}$ is achieved at some eigenvalue $\lambda_k$, $1 \le k \le d$. This proves that
\[
\sup_{\lambda \in \mathbb{R}} \frac{\|O - \lambda I_{d}\|_{\infty}}{\|O - \lambda I_{d}\|_{1}}
= \max_{k = 1, \ldots, d}
\left\{
\frac{\max_{1 \le i \le d}|\lambda_{i} - \lambda_{k}|}
{\sum_{i=1}^{d}|\lambda_{i} - \lambda_{k}|}
\right\},
\]
as claimed.
\end{proof}
Next, beyond the analytical lower bound, we can also formulate an SDP-type lower bound, which is theoretically at least as strong as the analytical one. 
\begin{corollary}
    Under the assumptions of Proposition~\ref{lb-err}, the lower bound for $t$ is the solution of the following SDP:
\begin{equation*}
  \begin{aligned}
    &\min_{J,T,\lambda,\beta}\ \beta \\[0.6ex]
    \text{subject to}\quad
    & J \in \mathsf{L}(\mathbb{C}^{d} \otimes \mathbb{C}^{d}),\quad
      T \in \mathsf{L}(\mathbb{C}^{d}),\quad
      \lambda,\beta \in \mathbb{R},\\
    & J \succeq 0,\quad T \succeq 0,\\
    & \operatorname{tr}_{2}(J) = I_{d},\\
    & J \preceq \frac{\beta + 1}{d}\, I_{d} \otimes I_{d},\\
    & \operatorname{tr}(T) \leq \epsilon,\\
    &
      \begin{pmatrix}
        T & \operatorname{tr}_{2}\!\big(J^{\intercal}(I_{d} \otimes O_{\lambda}^{\intercal})\big) - O_{\lambda}\\
        \operatorname{tr}_{2}\!\big(J^{\intercal}(I_{d} \otimes O_{\lambda}^{\intercal})\big) - O_{\lambda} & T
      \end{pmatrix} \succeq 0.
  \end{aligned}
\end{equation*}
\end{corollary}

\section{Comparison with the virtual comb approach}
\label{app:virtualcomb}

Here we compare the efficiency of our approach and the virtual comb approach \cite{mo2025parameterized} in estimating $\tr\!\bigl[U^{\dagger}\rho UO]$ for any qubit observable $O$. Without loss of generality we again take $O=Z$. 
In particular, we calculate the variance when estimating $\tr\!\left[U^{\dagger}\rho UZ\right]$ through these two methods.   

It is noted that $\tr\!\left[U^{\dagger}\rho UZ\right]$ could be considered as an expectation of the following binary distribution 
\begin{equation}
    B=\left\{\left(1,\tr\!\left[U^{\dagger}\rho U\ket{0}\!\bra{0}\right]\right),\left(-1,\tr\!\left[U^{\dagger}\rho U\ket{1}\!\bra{1}\right]\right)\right\},
\end{equation}
the expectation and the variance of which are
\begin{equation}
\begin{aligned}
    \mathbb E_{V\sim B}[V]&=\tr\!\left[U^{\dagger}\rho U\ket{0}\!\bra{0}\right]-\tr\!\left[U^{\dagger}\rho U\ket{1}\!\bra{1}\right]=\tr\!\left[U^{\dagger}\rho UZ\right],\\
    \mathbb D_{V\sim B}[V]&=\mathbb E_{V\sim B}\left[V^2\right]-\mathbb E_{V\sim B}[V]^2=1-\tr\!\left[U^{\dagger}\rho UZ\right]^2.
\end{aligned}
\end{equation}
As a result, the distribution $B$ is an unbiased estimation of $\tr\!\left[U^{\dagger}\rho UZ\right]$.

For our shadow inversion approach, we have the following:
\begin{proposition}
    The $3$-query shadow inversion derives an unbiased estimation of $\tr\!\left[U^{\dagger}\rho UZ\right]$, whose variance is $\frac{3}{N}\left(1-\tr\!\left[U^{\dagger}\rho UZ\right]^2\right)$ with totally $N$ uses of $U$.
\end{proposition}
\begin{proof}
Suppose $N$ is a multiple of $3$, and we could sample distribution $B$ for $N/3$ times. Hence the sample mean $\overline V$ is also an unbiased estimation of expectation $\tr\!\left[U^{\dagger}\rho UZ\right]$ with variance
\begin{equation}
    \mathbb D_{V_j\sim B}\left[\overline V\right]=\frac{3}{N}\mathbb D_{V\sim B}[V]=\frac{3}{N}\left(1-\tr\!\left[U^{\dagger}\rho UZ\right]^2\right).
\end{equation}
\end{proof}

On the other hand, for the virtual comb approach it holds that
\begin{proposition}
The variance of the estimation of $\tr\!\left[U^{\dagger}\rho UZ\right]$ based on $1$-slot virtual comb in \cite{mo2025parameterized} is at least $\frac{1}{N}\left(9-\tr\!\left[U^{\dagger}\rho UZ\right]^2\right)$ with totally $N$ uses of $U$.
\end{proposition}
\begin{proof}
By \cite{mo2025parameterized}, there exists a virtual comb 
\begin{equation}
    \cC=(1+\eta)\cC_+-\eta\cC_-
\end{equation}
satisfying
\begin{equation}
    \cC(U)(\rho)=U^{\dagger}\rho U,
\end{equation}
with both $\cC_+$ and $\cC_-$ 1-slot combs and  the minimal value of $\eta$ is exactly $1$. 
Denote the two binary distributions derived by the $Z$-measures of $\cC_+(U)(\rho)$ and $\cC_-(U)(\rho)$ as
\begin{equation}
    B_\pm=\left\{\left(1,\tr\!\left[\cC_\pm(U)(\rho)\ket{0}\!\bra{0}\right]\right),\left(-1,\tr\!\left[\cC_\pm(U)(\rho)\ket{1}\!\bra{1}\right]\right)\right\},
\end{equation}
and we have 
\begin{equation}
\begin{aligned}
    \mathbb E_{V\sim B_\pm}[V]&=\tr\!\left[\cC_\pm(U)(\rho)\ket{0}\!\bra{0}\right]-\tr\!\left[\cC_\pm(U)(\rho)\ket{1}\!\bra{1}\right]=\tr\!\left[\cC_\pm(U)(\rho)Z\right],\\
    \mathbb D_{V\sim B_\pm}[V]&=\mathbb E_{V\sim B_\pm}\left[V^2\right]-\mathbb E_{V\sim B_\pm}[V]^2=1-\tr\!\left[\cC_\pm(U)(\rho)Z\right]^2.
\end{aligned}
\end{equation}
In the case with minimal $\eta$, we find 
\begin{equation}
    \left|\tr\!\left[\cC_+(U)(\rho)Z\right]\right|=\left|\tr\!\left[\cC_-(U)(\rho)Z\right]\right|
\end{equation}
and moreover
\begin{equation}
    \tr\!\left[\cC_+(U)(\rho)Z\right]=-\tr\!\left[\cC_-(U)(\rho)Z\right]=\frac13\tr\!\left[\cC(U)(\rho)Z\right].
\end{equation}
If we use $x$ times of $U$ in estimating $\tr\!\left[\cC_+(U)(\rho)Z\right]$ and $N-x$ times of $U$ in estimating $\tr\!\left[\cC_-(U)(\rho)Z\right]$, we have $(1+\eta)\overline V-\eta \overline W$ is an unbiased estimation of $\tr\!\left[U^{\dagger}\rho UZ\right]$ where each $V_j\sim B_+$ and each $W_k\sim B_-$, whose variance is
\begin{align}
    &\mathbb D_{V_j\sim B_+,W_k\sim B_-}\left[2\overline V-\overline W\right]\\
    =&4\mathbb D_{V_j\sim B_+}\left[\overline V\right]+\mathbb D_{W_k\sim B_-}\left[\overline W\right]\\
    =&\frac{4}{x}\mathbb D_{V\sim B_+}\left[V\right]+\frac{1}{N-x}\mathbb D_{W\sim B_-}\left[W\right]\\
    =&\left(\frac{4}{x}+\frac{1}{N-x}\right)\left(1-\frac{1}{9}\tr\!\left[U^{\dagger}\rho UZ\right]^2\right)\\
    \ge&\frac{1}{N}\left(9-\tr\!\left[U^{\dagger}\rho UZ\right]^2\right),
\end{align}
where the last equality holds only if $x=2N/3$.
\end{proof}





\section{General SDP framework for shadow inversion}\label{sec:optimizatoin_problem}

In this section, we will give the formulation of the SDP model tailored for the general shadow inversion problem with any $d,t \in \mathbb{N}^{+}$ introduced in Section~\ref{prob-formu}. Specifically, we will show some crucial properties which allow us to reduced the size of variables in the SDP significantly.
\subsection{General SDP formulation and symmetry property}

As we have introduced in Section~\ref{prob-formu}, for $d,t \in \mathbb{N}^{+}$ the shadow inversion condition equals
\begin{equation*}
\mathcal{N}_{U}^{\dagger}(O) = UOU^{\dagger}
\end{equation*}
for $U \in \opr{U}(d)$ where $O$ is any fixed $d$-dimensional observable and $\mathcal{N}_{U}$ is the output channel of the quantum comb $C$ after $t$ queries to $U$. Formulate in the language of Choi operator and link product \cite{Chiribella2008quantum}
(see Appendix~\ref{sec:optimizatoin_problem} for details), it yields
\begin{equation*}
\begin{aligned}
    \mathcal{N}_{U}^{\dagger}(O) &= F(C \ast |U\rangle\!\rangle\!\langle\!\langle U|_{\boldsymbol{IO}}^{\otimes t})^{\intercal}F \ast O_{F}\\
    &= \tr_{F}\bigl[F(C \ast |U\rangle\!\rangle\!\langle\!\langle U|_{\boldsymbol{IO}}^{\otimes t})^{\intercal}F(O_{F}^{\intercal} \otimes I_{P})\bigr].
\end{aligned}
\end{equation*}
where $F$ is the switch operator $F(\ket{a}\otimes\ket{b})=\ket{b}\otimes\ket{a}$ and 
$C \in \mathsf{L}\left(\mathcal{H}_P \otimes \bigotimes_{i=1}^{\intercal} (\mathcal{H}_{I_i} \otimes \mathcal{H}_{O_i}) \otimes \mathcal{H}_F\right)$ is the Choi operator of a $t$-slots sequential quantum comb. 
\begin{remark}
    We will take sequential quantum comb as an example and one can also get the results for parallel quantum comb in a similar way. When dealing with the parallel situation, note that the sequential and parallel quantum comb have Choi operators $C$ of different system order. This means for sequential situation we have
\begin{equation*}
    C \ast |U\rangle\!\rangle\!\langle\!\langle U|^{\otimes t} = Tr_{\boldsymbol{IO}}\bigl[C^{T_{\boldsymbol{IO}}}(I_{P} \otimes |U\rangle\!\rangle\!\langle\!\langle U|^{\otimes t} \otimes I_{F})\bigr]
\end{equation*}
while for parallel situation we have
\begin{equation*}
    C \ast |U\rangle\!\rangle\!\langle\!\langle U|^{\otimes t} = Tr_{\boldsymbol{IO}}\bigl[C^{T_{\boldsymbol{IO}}}\left(I_{P} \otimes \Pi(|U\rangle\!\rangle\!\langle\!\langle U|^{\otimes t})\Pi \otimes I_{F}\right)\bigr]
\end{equation*}
where $\Pi$ is the permutation operator maps tensor factors from the ordering
\begin{equation*}
(I_1, O_1, \ldots, I_t, O_t)
\end{equation*}
to the ordering
\begin{equation*}
( I_1, \ldots, I_t, O_1, \ldots, O_t).
\end{equation*}
\end{remark}
\noindent Moreover, We also have the following equation
\begin{equation*}
    UOU^{\dagger} = |U\rangle\!\rangle\!\langle\!\langle U|_{FP} \ast O_F
    = \tr_{F}\bigl[|U\rangle\!\rangle\!\langle\!\langle U|_{FP}(O_{F}^{\intercal} \otimes I_{P})\bigr].
\end{equation*}
Here we regard $|U\rangle\!\rangle\!\langle\!\langle U| \in \mathsf{L}(\mathcal{H}_{F} \otimes \mathcal{H}_{P})$, that is, we take an input observable $O$ in $\mathcal{H}_{F}$ and obtain an output hermitian operator on $\mathcal{H}_P$. 
Next, we will give the following SDP model for general $t$-query shadow inversion of $d$-dimensional unitary under $O$ in the setting of sequential quantum combs:
\begin{equation}\label{formulation-SDP}
    \begin{aligned}
        \min_{C} \;\;
        & \int_{\opr{U}(d)}
        \Big\|
        \tr_{F}\!\Big[
        F\big(C \ast |U\rangle\!\rangle\!\langle\!\langle U|_{\boldsymbol{IO}}^{\otimes t}\big)^{\intercal}
        F (O_{F}^{\intercal} \otimes I_P)  \\
        &\qquad\qquad\quad
        - |U\rangle\!\rangle\!\langle\!\langle U|_{FP}(O_{F}^{\intercal} \otimes I_{P})
        \Big]
        \Big\|\,
        d\mu_{H}(U) \\[1em]
        \text{s.t.}\;\;
        & 0 \le C \in \mathsf{L}\left(\mathcal{H}_P \bigotimes_{i=1}^{k} (\mathcal{H}_{I_{i}} \otimes \mathcal{H}_{O_{i}}) \otimes \mathcal{H}_F\right), \\
        & \tr_{F}(C)
        = \tr_{O_t F}(C) \otimes \frac{I_{O_t}}{d_{O_t}},\\
        & \tr_{I_t O_t F}(C)
        = \tr_{O_{t-1} I_t O_t F}(C)\otimes \frac{I_{O_{t-1}}}{d_{O_{t-1}}},\\
        & \ \vdots \\
        & \tr_{I_1 O_1 \cdots I_t O_t F}(C)
        = \tr_{P I_1 O_1 \cdots I_t O_t F}(C)\otimes \frac{I_P}{d_P},\\
        & \tr(C) = d_{P}\, d_{\boldsymbol{O}}.
    \end{aligned}
\end{equation}
The norm $||\cdot||$ here can be chosen as any legal matrix norm and we choose the Frobenius norm in practice.
\begin{remark}
In the traditional setting about deterministic and exact inversion of unknown unitary operation, one can take the average channel fidelity between $C \ast |U\rangle\!\rangle\!\langle\!\langle U|_{\boldsymbol{IO}}^{\otimes t}$ and $|f(U)\rangle\!\rangle\!\langle\!\langle f(U)|_{PF}$ as the target function in SDP model. By this, the performance operator
\begin{equation*}
\Omega := \frac{1}{d^2} \int_{\opr{U}(d)} |f(U)\rangle\!\rangle\!\langle\!\langle f(U)|_{PF} \otimes |U^*\rangle\!\rangle\!\langle\!\langle U^*|_{\boldsymbol{IO}}^{\otimes t} d\mu_{H}(U)
\end{equation*} 
has been brought up in \cite{Per-ope} and the target function will be $\tr(C\Omega)$. Due to the good symmetry property of $\Omega$, the entire SDP can be considered on the basis $E_{ij}^{\mu}$ as mentioned in \cite{Rever-unknown}. However, fidelity-based measures are not suitable for the shadow information \cite{recover-shadow} which is one of the crucial difference between these two settings.
\end{remark}
    \noindent Meanwhile, the symmetry property of the Choi operator $C$ will also be quite different which will be dependent on the observable $O$. Therefore, we will discuss about this in the following paragraph.
\medskip
\begin{proposition}\label{sym-thm}
    For any $d,t \in \mathbb{N}^+$, if $C$ is a feasible solution of the SDP \eqref{formulation-SDP} for general $t$-query shadow inversion of $d$-dimensional unitary under some fixed $d$-dimensional observable $O$ in the setting of sequential quantum combs, then we can construct $\phi(C)$ that is also a feasible solution, where
    \begin{equation*}
        \phi(C) = \left(V_{P} \otimes (V_{I} \otimes V_{O})^{\otimes t} \otimes V_{F}\right) C \left(V_{P} \otimes (V_{I} \otimes V_{O})^{\otimes t} \otimes V_{F}\right)^{\dagger}
    \end{equation*}
    and $V_{P}, V_I, V_O, V_F$ satisfy the following constraints:
    \begin{align}
    &V_{P}, V_{I}, V_{O}, V_{F} \in \opr{U}(d), \nonumber \\
    &V_P = V_O, \label{p-equal-o}\\
    &V_{I}, V_{O} \in O'\cap\opr{U}(d). \label{i-co-f}
    \end{align}
    Here $O'$ is the commutant of the observable $O$:
    \begin{equation*}
    O' = \{U\,|\, UO = OU\}.
    \end{equation*}
\end{proposition}
\begin{proof}
For fixed $d$-dimensional observable $O$, we first denote that
\begin{equation}\label{fO-trans}
    \begin{aligned}
        f_O(C)
        = \int_{\opr{U}(d)}
        \Big\|
        \tr_{F}\!\Big[
        &\, F
        \Big(
        C \ast |U\rangle\!\rangle\!\langle\!\langle U|_{\boldsymbol{IO}}^{\otimes t}
        \Big)^{\intercal}
        F (O_{F}^{\intercal} \otimes I_P)  \\[0.6em]
        &\; -\, |U\rangle\!\rangle\!\langle\!\langle U|_{FP}
        (O_{F}^{\intercal} \otimes I_{P})
        \Big]
        \Big\|
        \, d\mu_{H}(U).
    \end{aligned}
\end{equation}
Then, due to the well-known fact
\begin{equation*}
    (A \otimes B)|U\rangle\!\rangle = |BUA^{\intercal}\rangle\!\rangle,
\end{equation*}
we will get the following equation:
\begin{equation*}
\begin{aligned}
    &\left(V_{P} \otimes (V_{I} \otimes V_{O})^{\otimes t} \otimes V_{F}\right) C \left(V_{P} \otimes (V_{I} \otimes V_{O})^{\otimes t} \otimes V_{F}\right)^{\dagger} \ast |U\rangle\!\rangle\!\langle\!\langle U|_{\boldsymbol{IO}}^{\otimes t}\\
    =& (V_P \otimes V_F)(C \ast |V_O^{\intercal} UV_I\rangle\!\rangle\!\langle\!\langle V_O^{\intercal} UV_I|_{\boldsymbol{IO}}^{\otimes t})(V_P^{\dagger} \otimes V_F^{\dagger}).
\end{aligned}
\end{equation*}
Hence we have
\begin{equation*}
    \begin{aligned}
        f_O(\phi(C))
        &= \int_{\opr{U}(d)}
        \Big\|
        \tr_{F}\!\Big[
        \, F (V_P^* \otimes V_F^*) 
        \Big(
        C \ast 
        \big|
        V_{O}^{\intercal} U V_{I}
        \big\rangle\!\big\rangle 
        \big\langle\!\big\langle
        V_{O}^{\intercal} U V_{I}
        \big|_{\boldsymbol{IO}}^{\otimes t}
        \Big)^{\intercal}
         \\[0.6em]
        &(V_P^{\intercal} \otimes V_F^{\intercal}) F 
        (O_{F}^{\intercal} \otimes I_{P})\; -\, 
        \big|
        U
        \big\rangle\!\big\rangle 
        \big\langle\!\big\langle
        U
        \big|_{FP} 
        (O_{F}^{\intercal} \otimes I_{P})
        \Big]
        \Big\|
        \, \mathrm{d}\mu_{H}(U) .
    \end{aligned}
\end{equation*}
Then by the condition \eqref{i-co-f}, we can obtain that
\begin{equation*}
\begin{aligned}
&|V_{O}^{\intercal} UV_I\rangle\!\rangle\!\langle\!\langle V_{O}^{\intercal} UV_I|_{FP}(O_F^{\intercal} \otimes I_P)\\
=&(V_I^{\intercal} \otimes V_O^{\intercal})|U\rangle\!\rangle\!\langle\!\langle U|_{FP}(V_I^* \otimes V_O^*)(O_F^{\intercal} \otimes I_P)\\
=&(V_I^{\intercal} \otimes V_O^{\intercal})|U\rangle\!\rangle\!\langle\!\langle U|_{FP}(O_F^{\intercal} \otimes I_P)(V_I^* \otimes V_O^*).
\end{aligned}
\end{equation*}
Moreover, again by \eqref{i-co-f}, we will derive that 
\begin{equation*}
\begin{split}
    f_O(\phi(C)) = &\int_{\opr{U}(d)} 
        \Big\|
            \tr_{F}\!\Big[
                (V_F^* \otimes V_P^*)F\big(C \ast |V_{O}^{\intercal} UV_I\rangle\!\rangle\!\langle\!\langle V_{O}^{\intercal} UV_I|_{\boldsymbol{IO}}^{\otimes t}\big)^{\intercal}F
                \\
                &(O_{F}^{\intercal} \otimes I_P)(V_F^{\intercal} \otimes V_P^{\intercal})- (V_I^* \otimes V_O^*)|V_O^{\intercal}UV_I\rangle\!\rangle\!\langle\!\langle V_O^{\intercal}UV_I|_{FP}\\
                &(O_{F}^{\intercal} \otimes I_{P})(V_I^{\intercal} \otimes V_O^{\intercal})\Big] 
        \Big\| \, d\mu_{H}(U).
\end{split}     
\end{equation*}
Next, due to the property of partial trace, we have
\begin{equation*}
\begin{aligned}
&\tr_{F}\!\Big[
                (V_F^* \otimes V_P^*)F\big(C \ast |V_{O}^{\intercal} UV_I\rangle\!\rangle\!\langle\!\langle V_{O}^{\intercal} UV_I|_{\boldsymbol{IO}}^{\otimes t}\big)^{\intercal}F
                (O_{F}^{\intercal} \otimes I_P)(V_F^{\intercal} \otimes V_P^{\intercal})
            \Big]\\
=& V_P^* \tr_{F}\!\Big[
                F\big(C \ast |V_{O}^{\intercal} UV_I\rangle\!\rangle\!\langle\!\langle V_{O}^{\intercal} UV_I|_{\boldsymbol{IO}}^{\otimes t}\big)^{\intercal}F
                (O_{F}^{\intercal} \otimes I_P)
            \Big] V_P^{\intercal}.
\end{aligned}
\end{equation*}
While in the same way we can get
\begin{equation*}
\begin{aligned}
&\tr_{F}\!\Big[
(V_I^* \otimes V_O^*)|V_O^{\intercal}UV_I\rangle\!\rangle\!\langle\!\langle V_O^{\intercal}UV_I|_{FP}(O_{F}^{\intercal} \otimes I_{P})(V_I^{\intercal} \otimes V_O^{\intercal})
\Big]\\
=& V_O^* \tr_F \!\Big[|V_O^{\intercal}UV_I\rangle\!\rangle\!\langle\!\langle V_O^{\intercal}UV_I|_{FP}(O_{F}^{\intercal} \otimes I_{P})\Big]V_O^{\intercal}.
\end{aligned}
\end{equation*}
Hence by making use of the condition \eqref{p-equal-o}, the property of Haar measure $\mu_H$ and the unitary invariance of the norm, we will get \begin{equation*}
    f_O(C) = f_O(\phi(C))
\end{equation*}
and then we complete the proof.
\end{proof}
Before we show the symmetry property of the Choi operator as a corollary of Theorem~\ref{sym-thm}, we need to analyze the structure of the commutant of the observable $O$ in the unitary group $\opr{U}(d)$.
\begin{lemma}\label{Haar-CO}
    For any $d \in \mathbb{N}^+$, let $O$ be a $d$-dimensional observable, the commutant of $O$ in the unitary group $\opr{U}(d)$
    \begin{equation*}
        O'\cap\opr{U}(d) \;=\;\{\,U\in \opr{U}(d)\colon UO = OU\}
    \end{equation*}
    is a closed subgroup of \(\opr{U}(d)\), and hence induces the Haar measure  $\mu_{H,O}$.
\end{lemma}

\begin{proof}
    We first check that $O'\cap\opr{U}(d)$ is a subgroup of $\opr{U}(d)$.
    \begin{itemize}
        \item The identity $I_d \in \opr{U}(d)$ clearly commutes with $O$, hence $I_d \in O'\cap\opr{U}(d)$.
        
        \item For any $U,V \in O'\cap\opr{U}(d)$, we will have
        \begin{equation*}
            (UV)O = U(VO) = U(OV) = (UO)V = (OU)V = O(UV)
        \end{equation*}
        which means that $UV \in O'\cap\opr{U}(d)$.
        
        \item For any $U \in O'\cap\opr{U}(d)$, we can get that
        \begin{equation*}
            U^{-1}O = U^{-1}OUU^{-1} = U^{-1}UOU^{-1} = OU^{-1}
        \end{equation*}
        which means that $U^{-1} \in O'\cap\opr{U}(d)$.
    \end{itemize}
    Next, let us consider the continuous map:
    \begin{equation*}
        \begin{aligned}
            f \colon \opr{U}(d) &\longrightarrow M_{d}(\mathbb{C}), \\
            U \;&\longmapsto\; UO - OU.
        \end{aligned}
    \end{equation*}
    Note that \(\{0\}\subset M_d(\mathbb{C})\) is closed and
    \begin{equation*}
        O'\cap\opr{U}(d) = f^{-1}(\{0\}),
    \end{equation*}
    it follows that \(O'\cap\opr{U}(d)\) is a closed subset of the compact group \(\opr{U}(d)\).  Then \(O'\cap\opr{U}(d)\) is itself compact (hence locally compact). Therefore, $O'\cap\opr{U}(d)$ carries the Haar measure which we will denote as $\mu_{H,O}$.
\end{proof}

Next, we establish the connection between the structure of $O'\cap\opr{U}(d)$ and unitary groups, a relation that enables us to apply Schur–Weyl duality theory in the analysis of the shadow inversion problem.

\begin{lemma}
    For any $d \in \mathbb{N}^+$, let $O$ be a $d$-dimensional observable, then the commutant of $O$ in the unitary group $\opr{U}(d)$ is isomorphic to a direct product of unitary groups:
    \begin{equation*}
        O'\cap\opr{U}(d) \;\cong\; \opr{U}(l_1) \times \opr{U}(l_2) \times \cdots \times \opr{U}(l_k),
    \end{equation*}
    where $l_j = \dim(E_j)$ and $E_j$ is the eigenspace of $O$ corresponding to the distinct eigenvalue $\lambda_j$.
\end{lemma}

\begin{proof}
    By the spectral theorem \cite{Nielson-chuang}, we can decompose $O$ as
    \[
    O = \sum_{j=1}^k \lambda_j P_j,
    \]
    where $\lambda_1,\ldots,\lambda_k$ are the distinct eigenvalues of $O$ and $P_j$ denotes the orthogonal projection onto the corresponding eigenspace
    \begin{equation*}
        E_j = \ker(O - \lambda_j I).
    \end{equation*}
    If $U \in O'\cap\opr{U}(d)$, then $UO = OU$ implies
    \begin{equation*}
        U \Big( \sum_{j=1}^k \lambda_j P_j \Big) 
        = \Big( \sum_{j=1}^k \lambda_j P_j \Big) U.
    \end{equation*}
    Using the fact that the projections $P_j$ are mutually orthogonal, this condition is equivalent to
    \begin{equation*}
        U P_j = P_j U \quad \text{for all } j=1,\dots,k.
    \end{equation*}
    Hence $U$ must map each eigenspace $E_j$ into itself. In other words,
    $U$ is block diagonal with respect to the direct sum decomposition
    \begin{equation*}
        \mathbb{C}^d = \bigoplus_{j=1}^k E_j.
    \end{equation*}
    Moreover, the action of $U$ on $E_j$ is an arbitrary unitary in $U(l_j)$, where $l_j = \dim(E_j)$. 
    Therefore,
    \begin{equation*}
        O'\cap\opr{U}(d) \;\cong\; U(l_1) \times U(l_2) \times \cdots \times U(l_k).
    \end{equation*}
\end{proof}
\begin{corollary}\label{cor-sym}
For any $d,t \in \mathbb{N}^+$, without loss of generality, in SDP \eqref{formulation-SDP} for general $t$-query shadow inversion of $d$-dimensional unitary under some fixed $d$-dimensional observable $O$ in the setting of sequential quantum combs we can assume
\begin{equation}\label{assump-sym}
    [C,U \otimes (V \otimes U)^{\otimes t} \otimes W] = 0. \quad \forall\,\, U \in \opr{U}(d), V,W \in O'\cap\opr{U}(d)
\end{equation}
where $O'$ is the commutant of the observable $O$.
\end{corollary}

\begin{proof}
We will keep use the notation $f_{O}(C)$ which has been introduced in \eqref{fO-trans}. Suppose that $C = C_{opt}$ achieves the minimum value of $f_O(C)$ in the SDP $\eqref{formulation-SDP}$, we can construct the following operator:
\begin{equation*}
    \begin{aligned}
        C_{opt}'
        = \int_{\opr{U}(d) \times O'\cap\opr{U}(d) \times O'\cap\opr{U}(d)}
        &\;
        \left(U' \otimes (V' \otimes U')^{\otimes t} \otimes W'\right) \,
        C_{opt} \,
        \left(U' \otimes (V' \otimes U')^{\otimes t} \otimes W'\right)^{\dagger} \\[0.6em]
        &\; \mathrm{d}\mu_{H}'(U' \times V' \times W')
    \end{aligned}
\end{equation*}
where $\mu_{H}'$ is the Haar measure on $\opr{U}(d) \times O'\cap\opr{U}(d) \times O'\cap\opr{U}(d)$ by Lemma~\ref{Haar-CO}. Denote
\begin{equation*}
    \phi_{U',V',W'}(C) = \left(U' \otimes (V' \otimes U')^{\otimes t} \otimes W'\right) C \left(U' \otimes (V' \otimes U')^{\otimes t} \otimes W'\right)^{\dagger}
\end{equation*}
and by Tonelli's Theorem and Proposition~\ref{sym-thm}, we can get
\begin{equation*}
\begin{aligned}
f_O(C_{opt}')
  &= \int_{\opr{U}(d)} \Big\| \int_{\opr{U}(d) \times O'\cap\opr{U}(d) \times O'\cap\opr{U}(d)} 
        \tr_{F}\!\Big[
            F\big(\phi_{U',V',W'}(C_{opt}) \ast |U\rangle\!\rangle\!\langle\!\langle U|_{\boldsymbol{IO}}^{\otimes t}\big)^{\intercal}F\\
  &(O_{F}^{\intercal} \otimes I_P) - |U\rangle\!\rangle\!\langle\!\langle U|_{FP}(O_{F}^{\intercal} \otimes I_{P})
        \Big]\, d\mu_{H}'(U' \times V' \times W')\Big\| \, d\mu_{H}(U) \\
  & \leq  \int_{\opr{U}(d)}  \int_{\opr{U}(d) \times O'\cap\opr{U}(d) \times O'\cap\opr{U}(d)} 
        \Big\| \tr_{F}\!\Big[
            F\big(\phi_{U',V',W'}(C_{opt}) \ast |U\rangle\!\rangle\!\langle\!\langle U|_{\boldsymbol{IO}}^{\otimes t}\big)^{\intercal}F \\
  &(O_{F}^{\intercal} \otimes I_P) - |U\rangle\!\rangle\!\langle\!\langle U|_{FP}(O_{F}^{\intercal} \otimes I_{P})
        \Big]\Big\|\, d\mu_{H}'(U' \times V' \times W') \, d\mu_{H}(U) \\
    & =  \int_{\opr{U}(d) \times O'\cap\opr{U}(d) \times O'\cap\opr{U}(d)} \int_{\opr{U}(d)}
            \Big\| \tr_{F}\!\Big[
                F\big(\phi_{U',V',W'}(C_{opt}) \ast |U\rangle\!\rangle\!\langle\!\langle U|_{\boldsymbol{IO}}^{\otimes t}\big)^{\intercal}F \\
      &(O_{F}^{\intercal} \otimes I_P) - |U\rangle\!\rangle\!\langle\!\langle U|_{FP}(O_{F}^{\intercal} \otimes I_{P})
            \Big]\Big\|\, d\mu_{H} (U)\, d\mu_{H}'(U' \times V' \times W') \\
    & = \int_{\opr{U}(d) \times O'\cap\opr{U}(d) \times O'\cap\opr{U}(d)} f_{O}(C_{opt}) d\mu_{H}'(U' \times V' \times W')\\
    & = f_{O}(C_{opt}).
\end{aligned}
\end{equation*}
Moreover, when $C = C_{opt}$ satisfies the sequential quantum comb conditions (See Eq.~\eqref{eq:sequential_constraints} in Appendix \ref{app:comb}) , $C_{opt}'$ also satisfies the conditions. Hence $C = C_{opt'}$ also achieves the minimum value of the SDP $\eqref{formulation-SDP}$.
\end{proof}
\begin{remark}
    For any $d,t \in \mathbb{N}^+$, if we consider the SDP model for general $t$-query shadow inversion of $d$-dimensional unitary under some fixed $d$-dimensional observable $O$ in the setting of parallel quantum combs, then the symmetry equation~\eqref{assump-sym} corresponding to the system order of parallel will be
    \begin{equation}\label{assump-sym-parallel}
        [C, U \otimes V^{\otimes t} \otimes U^{\otimes t} \otimes W] = 0. \quad \forall U \in \opr{U}(d), V,W \in O'\cap\opr{U}(d)
    \end{equation}
    where $O'$ is the commutant of the observable $O$.
\end{remark}
\subsection{The number of variables in the simplified SDP}
In this section, we will make use of Character theory to compute the number of variables in the simplified SDP. We take the symmetry equation \eqref{assump-sym} as example and it is equivalent to 
\begin{equation*}
    [P_{\pi}CP_{\pi}, U^{\otimes t+1} \otimes V^{\otimes t} \otimes W] = 0 \quad \forall U \in \opr{U}(d), V,W \in O'\cap\opr{U}(d).
\end{equation*}
Here $P_{\pi}$ is any permutation operator of $\pi \in S_{2t+2}$ that maps the tensor factors from the causal ordering
\begin{equation*}
(P, I_1, O_1, \ldots, I_t, O_t, F)
\end{equation*}
to the grouped ordering
\begin{equation*}
(P, O_1, \ldots, O_t, I_1, \ldots, I_t, F),
\end{equation*}
i.e., all output systems come before all input systems (while the relative order within each group is irrelevant). For example, we can take $\pi$ as follows
\begin{equation*}
    \pi : 
    \begin{cases}
        1 \;\mapsto\; 1, \\[6pt]
        2j \;\mapsto\; t+j+1, & \text{for } j = 1, \dots, t, \\[6pt]
        2j+1 \;\mapsto\; j+1, & \text{for } j = 1, \dots, t, \\[6pt]
        2t+2 \;\mapsto\; 2t+2.
    \end{cases}
\end{equation*}
\begin{remark}
Starting from here, We label the systems \(P, I_1, O_1, \ldots, I_t, O_t, F\) by the integers \(1, 2, \ldots, 2t+2\), respectively.
\end{remark}
\noindent Then by the Schur-Weyl duality and the tools of Young tableau introduced in Section~\ref{g-r-t}, we can construct a Schur unitary matrix for the representation
\begin{equation*}
\begin{aligned}
 \rho_{O,t} \colon G \times O'\cap\opr{U}(d) \times O'\cap\opr{U}(d) &\to \mathrm{GL}\left((\mathbb{C}^{d})^{\otimes (2t+2)}\right) \\
    (U,V,W) &\mapsto U^{\otimes t+1} \otimes V^{\otimes t} \otimes W
\end{aligned}
\end{equation*}
which we denote as $Q_{O,t}$. It means for any $U \in \opr{U}(d)$, $V,W \in O'\cap\opr{U}(d)$, 
\begin{equation*}
Q_{O,t}^{\dagger} \rho_{O,t}(U,V,W) Q_{O,t} \cong \bigoplus_{r \in I} m_{r} \,\rho_{O,t}^{r}
\end{equation*}
where $\rho_{O,t}^{r}$ is the irreducible representation of $\rho_{O,t}$ labeled by $r$ and $m_r$ is its multiplicity. Then the column vectors of $Q_{O,t}$ can be labeled by three parameters $r$, $a$ and $\alpha$. More precisely, if we consider the following direct sum decomposition:
\begin{equation*}
(\mathbb{C}^{d})^{\otimes (2t+2)} = \bigoplus_{r \in I} \big(V_{r} \otimes I_{m_{r}}\big),
\end{equation*}
where $V_r$ is the representation space for $\rho_{O,t}^{r}$ and then each column vector of $Q_{O,t}$ can be written as $|v_{r,a,\alpha} \rangle$ with $r \in I$, $a \in \{1,\cdots,\dim(V_{r})\}$ and $\alpha \in \{1,\cdots m_{r}\}$. Then the Choi operator $C$ can be written as:
\begin{equation*}
P_{\pi}CP_{\pi} = \sum_{r \in I} \sum_{\alpha,\beta =1}^{m_r} c_{\alpha,\beta}^{(r)}(\sum_{a= 1}^{\dim(V_r)}|v_{r,a,\alpha} \rangle \langle v_{r,a,\beta}|).
\end{equation*}
Hence $C$ is positive semi-definite is equivalent to for each $r \in I$,
\begin{equation*}
C^{(r)} = [c_{\alpha,\beta}^{(r)}]_{1 \leq \alpha, \beta \leq m_r } \geq 0.
\end{equation*}
Moreover, we will set up an analysis about the number of variables $c_{\alpha,\beta}^{(r)}$ in the simplified SDP.
\begin{proposition}\label{num-ub}
For any $d,t \in \mathbb{N}^{+}$, let $O$ be some fixed $d$-dimensional observable and $N_{O,t}$ be the number of variables in the simplified SDP for general $t$-query shadow inversion of $d$-dimensional unitary under $O$. We assume that $O'\cap\opr{U}(d) \cong U(l_{1}) \times U(l_2) \times \cdots \times U(l_m)$ for some $m \in \mathbb{N}^+$ with $l_{i} \in \mathbb{N}^+$ and $\sum_{i=1}^{m} l_i = d$, then we will have the following equation:
\begin{equation*}
    N_{O,t} = m\,I_{t+1}(d)\,J_{t}(d)
\end{equation*}
where $I_{k}(d)$ and $J_{s}(d)$ for $k,s,d \in \mathbb{N}^{+}$ are given as
\begin{align*}
    I_{k}(d) &= \sum_{\substack{\lambda \vdash k\\ l(\lambda) \leq d}} \left(\frac{k!}{H_{\lambda}}\right)^2,\\
    J_{s}(d) &= \sum_{k_{1}+\cdots+k_m = s} \Bigl[\frac{s!}{k_{1}!k_{2}!\cdots k_{m}!}\Bigr]^{2} \prod_{r=1}^{m} I_{k_r}(l_r).
\end{align*}
Here $H_{\lambda}$ is the hook length of the Young diagram $Y_{\lambda}$.
\end{proposition}
\begin{proof}
In this proof, we will use the notations and properties introduced in Section~\ref{g-r-t}. By Character theory \cite{kowalski2025introduction}[Theorem~5.5.1]
\begin{equation*}
\begin{aligned}
    N_{O,t} = \sum_{r} m_{r}^2 &= \langle \chi_{\rho_{O,t}},\chi_{\rho_{O,t}}\rangle_{L^{2}(G \times O'\cap\opr{U}(d) \times O'\cap\opr{U}(d))}\\
    &= \int_{G \times O'\cap\opr{U}(d) \times O'\cap\opr{U}(d)} |(\tr U)^{t+1} \, (\tr V)^{\intercal} \, \tr W|^2 d\mu_{H'}(U \times V \times W)\\
    &= I_{t+1}(d)\,J_{t}(d)\,J_{1}(d)
\end{aligned}
\end{equation*}
where we introduce the notation:
\begin{equation*}
    I_{k}(d):= \int_{\opr{U}(d)} |\tr U|^{2k} dU, \,\, J_{s}(d):= \int_{O'\cap\opr{U}(d)} |\tr V|^{2s} dV.
\end{equation*}
Next let $a_r = \tr(M_r), M_{r} \in U(l_r)$ for $r = 1,2,\cdots,m$, then we can get
\begin{equation*}
\begin{aligned}
    |\tr V|^{2s} &= (a_{1}+a_{2}+\cdots+a_{m})^{s} (\bar{a}_{1} + \bar{a}_{2} + \cdots + \bar{a}_{m})^s\\
    &= \sum_{\substack{k_{1}+\cdots+k_m = s\\
    n_{1}+\cdots+n_{m} = s}} \frac{s!}{k_{1}! k_{2}! \cdots k_{m}!} \cdot \frac{s!}{n_{1}! n_{2}! \cdots n_{m}!} \cdot \prod_{r=1}^{m} a_{r}^{k_r} \cdot \bar{a}_{r}^{n_r}.
\end{aligned}
\end{equation*}
Then by the orthogonality of irreducible characters
\begin{equation*}
    \begin{aligned}
      J_{s}(d) &= \sum_{\substack{k_{1}+\cdots+k_m = s\\
    n_{1}+\cdots+n_{m} = s}} \frac{s!}{k_{1}! k_{2}! \cdots k_{m}!} \cdot \frac{s!}{n_{1}! n_{2}! \cdots n_{m}!} \cdot \prod_{r=1}^{m} \int_{U(l_r)} a_{r}^{k_r} \cdot \bar{a}_{r}^{n_r} dM_{r}\\
    &= \sum_{\substack{k_{1}+\cdots+k_m = s\\
    n_{1}+\cdots+n_{m} = s}} \frac{s!}{k_{1}! k_{2}! \cdots k_{m}!} \cdot \frac{s!}{n_{1}! n_{2}! \cdots n_{m}!} \cdot \prod_{r=1}^{m} \delta_{k_{r},n_{r}} I_{k_r}(l_{r})\\
    &= \sum_{k_{1}+\cdots+k_m = s} \Bigl[\frac{s!}{k_{1}!k_{2}!\cdots k_{m}!}\Bigr]^{2} \prod_{r=1}^{m} I_{k_r}(l_r).
    \end{aligned}
\end{equation*}
Next by Character theory again we can get 
\begin{equation*}
\begin{aligned}
    I_{k}(d) &= \sum_{\substack{\lambda \vdash k\\ l(\lambda) \leq d}} \dim(\mathcal{S}_{\lambda})^2\\
    &= \sum_{\substack{\lambda \vdash k\\ l(\lambda) \leq d}} \left(\frac{k!}{H_{\lambda}}\right)^2.
\end{aligned}
\end{equation*}
Together with the following equation 
\begin{equation*}
    J_{1}(d) = \sum_{k_{1}+\cdots+k_m = 1} \Bigl[\frac{1}{k_{1}!k_{2}!\cdots k_{m}!}\Bigr]^{2} \prod_{r=1}^{m} I_{k_r}(l_r) = m,
\end{equation*}
we will complete the proof.
\end{proof}

\upperboundvariables*


\begin{proof}
    Note that we have the following estimation
    \begin{equation*}
        I_{k}(d) = \sum_{\substack{\lambda \vdash k\\ l(\lambda) \leq d}} \left(\frac{k!}{H_{\lambda}}\right)^2 \leq \sum_{\lambda \vdash k}\left(\frac{k!}{H_{\lambda}}\right)^2 = k!
    \end{equation*}
    for any $d \in \mathbb{N}^{+}$. Then we can get
    \begin{equation*}
        \begin{aligned}
            J_{t}(d) &= \sum_{k_{1}+\cdots+k_m = t} \Bigl[\frac{t!}{k_{1}!k_{2}!\cdots k_{m}!}\Bigr]^{2} \prod_{r=1}^{m} I_{k_r}(l_r)\\
            &\leq \sum_{k_{1}+\cdots+k_m = t} \Bigl[\frac{t!}{k_{1}!k_{2}!\cdots k_{m}!}\Bigr]^{2} \prod_{r=1}^{m} k_{r}!\\
            &= t! \sum_{k_{1}+\cdots+k_m = t} \frac{t!}{k_{1}!k_{2}!\cdots k_{m}!}\\
            &= t! \, m^{\intercal}.
        \end{aligned}
    \end{equation*}
    Finally we can derive that
    \begin{equation*}
        N_{O,t} = m\,I_{t+1}(d)\,J_{t}(d) \leq (t+1)! \, t! \, m^{t+1} \leq (t+1)! \, t! \, d^{t+1}.
    \end{equation*}
\end{proof}
\subsection{The simplification process for SDP}\label{subsec:simplification}
In this section, we will show how to make use of the symmetry property of the Choi operator $C$ to simplify the SDP \eqref{formulation-SDP} which is established for sequential quantum comb and we will directly give the results established for parallel quantum comb which can be derived through the same process by making use of the symmetry property \eqref{assump-sym-parallel}. We expand each $|v_{r,a,\alpha}\rangle \in (\mathbb{C}^{d})^{\otimes (2t+2)}$ in the computation basis 
\begin{equation*}
|v_{r,a,\alpha}\rangle = \sum_{i_{1},i_{2},\cdots,i_{2t+2} \in [d]} p_{i_{1},i_{2},\cdots,i_{2t+2}}^{r,a,\alpha} |e_{i_1} \otimes e_{i_2} \cdots \otimes e_{i_{2t+2}} \rangle
\end{equation*}
where $\{e_{i}\}_{i=1}^{d}$ is the standard orthogonal basis of $\mathbb{C}^{d}$. Hence we have
\begin{equation}\label{eq:form-of-C}
    \begin{split}
        C = \; &\sum_{r \in I}\sum_{\alpha,\beta = 1}^{m_r} \sum_{a=1}^{\dim(V_r)}
        \sum_{\substack{i_{1},\cdots,i_{2t+2} \in [d] \\[0.4em] j_{1},\cdots,j_{2t+2} \in [d]}} 
        c_{\alpha,\beta}^{(r)}\, p_{i_{1},i_{2},\cdots,i_{2t+2}}^{r,a,\alpha}\, \Big(p_{j_{1},j_{2},\cdots,j_{2t+2}}^{r,a,\beta}\Big)^{*}\\[1mm]
        &\quad \; |e_{i_{\pi(1)}} \otimes e_{i_{\pi(2)}} \otimes \cdots \otimes e_{i_{\pi(2t+2)}} \rangle 
        \langle e_{j_{\pi(1)}} \otimes e_{j_{\pi(2)}} \otimes \cdots \otimes e_{j_{\pi(2t+2)}}|,
    \end{split}
\end{equation}
where we make use of the definition that for any $\pi \in S_{2t+2}$, we have
\begin{equation*}
    P_{\pi^{-1}}|e_{i_1} \otimes e_{i_2} \otimes \cdots \otimes e_{i_{2t+2}} \rangle = |e_{i_{\pi(1)}} \otimes e_{i_{\pi(2)}} \otimes \cdots \otimes e_{i_{\pi(2t+2)}} \rangle.
\end{equation*}
Next, we relabel the summation indices by the permutation $\pi$. Denote
\begin{equation*}
    i_{k}' = i_{\pi(k)}, \quad j_{k}' = j_{\pi(k)}
\end{equation*}
for each $k \in \{1,\cdots,2t+2\}$ and rename $i_{k'} \to i_{k}, j_{k'} \to j_{k}$. Since the sums run over all index values, this relabeling leaves the total invariant but restores the tensor product factors to canonical order, that is we can rewritten $C$ as:
\begin{equation}\label{Form-of-C}
    \begin{split}
        C = \; &\sum_{r \in I}\sum_{\alpha,\beta = 1}^{m_r} \sum_{a=1}^{\dim(V_r)}
    \sum_{\substack{i_{1},\cdots,i_{2t+2} \in [d] \\[0.4em] j_{1},\cdots,j_{2t+2} \in [d]}} 
        c_{\alpha,\beta}^{(r)}\, p_{P_{\pi}(i_{1},i_{2},\cdots,i_{2t+2})}^{r,a,\alpha}\, \left(p_{P_{\pi}(j_{1},j_{2},\cdots,j_{2t+2})}^{r,a,\beta}\right)^{*}\\[1mm]
        &\quad \; |e_{i_{1}} \otimes e_{i_{2}} \otimes \cdots \otimes e_{i_{2t+2}} \rangle 
        \langle e_{j_{1}} \otimes e_{j_{2}} \otimes \cdots \otimes e_{j_{2t+2}}|
    \end{split}
\end{equation}
where we use the notation that
\begin{equation*}
    P_{\pi}(i_{1}, i_{2}, \cdots, i_{2t+2}) = (i_{\pi^{-1}(1)},i_{\pi^{-1}(2)},\cdots,i_{\pi^{-1}(2t+2)}).
\end{equation*}
\subsubsection{Simplification for the constraints}
To compute the partial trace of sub-systems,  we introduce the following notation: for any $K \in \{1,\cdots,2t+2\}$, we set
\begin{equation*}
    \mathbf{K} = \{2t+3-K, \cdots, 2t+2\}
\end{equation*}
to be the index set of last $K$ systems, namely, systems $2t+3-K$ through $2t+2$. We also denote
\begin{equation*}
    \mathbf{R_{K}} = \{1,\cdots, 2t+2-K\} 
\end{equation*}
be the index set of first $2t+2-K$ systems, namely, systems $1$ through $2t+2-K$. We also denote 
\begin{equation*}
i_{\mathbf{K}} = (i_{k})_{k \in \mathbf{K}},\,\, j_{\mathbf{R_K}} = (j_{l})_{l \in \mathbf{R_K}},\,\,
|e_{i_{\mathbf{K}}}\rangle = \bigotimes_{k \in \mathbf{K}} |e_{i_k}\rangle,\,\,
|e_{j_{\mathbf{R_K}}}\rangle = \bigotimes_{l \in \mathbf{R_K}} |e_{i_l}\rangle
\end{equation*}
for convenience. Then, the constraints for quantum comb (See Eq.~\eqref{eq:sequential_constraints} in Appendix \ref{app:comb}) can be written as:
\begin{equation*}
    \tr_{\mathbf{K-1}}(C) = \tr_{\mathbf{K}}(C) \otimes \frac{I_{d}}{d}, \, \forall K \in \{2k : k = 1,2,\cdots,t+1\}
\end{equation*}
and
\begin{equation*}
    \tr(C) = d^{t+1}.
\end{equation*}
We all know that for a given operator $A \in \mathsf{L}(\mathcal{H}_{1}\otimes \cdots \otimes \mathcal{H}_{2t+2})$, its partial trace over the subsystems $\mathbf{K}$ is defined by
\begin{equation*}
\operatorname{Tr}_{\mathbf{K}}(A) = \sum_{i_{\mathbf{K}} \in [d]^{|\mathbf{K}|}} \Biggl( \bigotimes_{l\in \mathbf{R_{K}}} I_l \otimes \langle e_{i_{\mathbf{K}}}| \Biggr) A \Biggl( \bigotimes_{l\in \mathbf{R_K}} I_l \otimes |e_{i_{\mathbf{K}}}\rangle \Biggr).
\end{equation*}
Using the orthogonality of the standard basis
\begin{equation*}
\langle e_{i_k}|e_{j_k}\rangle = \delta_{i_k,j_k},
\end{equation*}
we can calculate that
\begin{equation*}
\operatorname{Tr}_{\mathbf{K}}\Bigl( |e_{i_1}\otimes \cdots \otimes e_{i_{2t+2}}\rangle \langle e_{j_1}\otimes \cdots \otimes e_{j_{2t+2}}| \Bigr)
= \prod_{k\in \mathbf{K}} \delta_{i_k,j_k} \; |e_{i_{\mathbf{R_K}}}\rangle \langle e_{j_{\mathbf{R_K}}}|.
\end{equation*}
When applying the partial trace to the operator $C$ in \eqref{Form-of-C}, we can get
\begin{equation*}
    \begin{aligned}
    \tr_{\mathbf{K}}(C) = &\sum_{r \in I}\sum_{\alpha,\beta = 1}^{m_r} \sum_{a=1}^{\dim(V_r)}
    \sum_{i_{\mathbf{R_K}}, j_{\mathbf{R_K}} \in [d]^{|\mathbf{R_K}|}} \sum_{i_{\mathbf{K}} \in [d]^{|\mathbf{K}|}}
    c_{\alpha,\beta}^{(r)}\, p_{P_{\pi}(i_{\mathbf{R_K}},i_{\mathbf{K}})}^{r,a,\alpha}\\ & \times \left(p_{P_{\pi}(j_{\mathbf{R_K}},i_{\mathbf{K}})}^{r,a,\beta}\right)^{*} |e_{i_{\mathbf{R_K}}} \rangle \langle e_{j_{\mathbf{R_K}}}|.
    \end{aligned}
\end{equation*}
Moreover, we can calculate that
\begin{equation*}
    \begin{aligned}
        \operatorname{Tr}_{\mathbf{K-1}}(C) &= \sum_{r \in I}\sum_{\alpha,\beta = 1}^{m_r} \sum_{a=1}^{\dim(V_r)} \sum_{i_{\mathbf{R_{K}}},j_{\mathbf{R_{K}}} \in [d]^{|\mathbf{R_{K}}|}} \sum_{x,y \in [d]} \sum_{i_{\mathbf{K-1}}  \in [d]^{|\mathbf{K-1}|}}\\
         &\quad \times c_{\alpha,\beta}^{(r)}\, p_{P_{\pi}(i_{\mathbf{R_{K}}},x,i_{\mathbf{K-1}})}^{r,a,\alpha}\,\left(p_{P_{\pi}(j_{\mathbf{R_{K}}},y,i_{\mathbf{K-1}})}^{r,a,\beta}\right)^{*} \\
        &\quad \times \left(|e_{i_{\mathbf{R_{K}}}}\rangle \langle e_{j_{\mathbf{R_{K}}}}| \otimes |e_x\rangle \langle e_y|\right)
    \end{aligned}
\end{equation*}
while we have
\begin{equation*}
    \begin{aligned}
        \operatorname{Tr}_{\mathbf{K}}(C) \otimes \frac{I_d}{d} &=  \frac{1}{d}\sum_{r \in I}\sum_{\alpha,\beta = 1}^{m_r} \sum_{a=1}^{\dim(V_r)} \sum_{i_{\mathbf{R_{K}}},j_{\mathbf{R_{K}}} \in [d]^{|\mathbf{R_{K}}|}} \sum_{z,y \in [d]} \sum_{i_{\mathbf{K-1}}  \in [d]^{|\mathbf{K-1}|}}\\ &\quad \times c_{\alpha,\beta}^{(r)}\, p_{P_{\pi}(i_{\mathbf{R_{K}}},z,i_{\mathbf{K-1}})}^{r,a,\alpha}\left(p_{P_{\pi}(j_{\mathbf{R_{K}}},z,i_{\mathbf{K-1}})}^{r,a,\beta}\right)^{*} \\
        &\quad \times \left(|e_{i_{\mathbf{R_{K}}}}\rangle \langle e_{j_{\mathbf{R_{K}}}}| \otimes |e_y\rangle \langle e_y|\right).
    \end{aligned}
\end{equation*}
\medskip
Therefore, for any fixed $i_{\mathbf{R_K}},j_{\mathbf{R_K}} \in [d]^{|\mathbf{R_K}|}$ and $x,y \in [d],x \neq y$, we have
\begin{equation*}
    \sum_{r \in I}\sum_{\alpha,\beta = 1}^{m_r} \sum_{a=1}^{dim(V_r)}\sum_{i_{\mathbf{K-1}}  \in [d]^{|\mathbf{K-1}|}} c_{\alpha,\beta}^{(r)} \, p_{P_{\pi}(i_{\mathbf{R_K}},x,i_{\mathbf{K-1}})}^{r,a,\alpha}\,\left(p_{P_{\pi}(j_{\mathbf{R_K}},y,i_{\mathbf{K-1}})}^{r,a,\beta}\right)^{*} = 0.
\end{equation*}
Also for any fixed $i_{\mathbf{R_K}},j_{\mathbf{R_K}} \in [d]^{|\mathbf{R_K}|}$ and $x = y = k \in [d]$, we have
\begin{equation*}
    \begin{split}
        &\sum_{r \in I}\sum_{\alpha,\beta = 1}^{m_r} \sum_{a=1}^{dim(V_r)}\sum_{i_{\mathbf{K-1}}  \in [d]^{|\mathbf{K-1}|}} c_{\alpha,\beta}^{(r)} \, p_{P_{\pi}(i_{\mathbf{R_K}},k,i_{\mathbf{K-1}})}^{r,a,\alpha} \, \left(p_{P_{\pi}(j_{\mathbf{R_K}},k,i_{\mathbf{K-1}})}^{r,a,\beta}\right)^{*}\\
        =&\frac{1}{d}\sum_{r \in I}\sum_{\alpha,\beta = 1}^{m_r} \sum_{a=1}^{dim(V_r)}\sum_{i_{\mathbf{K-1}}  \in [d]^{|\mathbf{K-1}|}} \Biggl[\sum_{z \in [d]}c_{\alpha,\beta}^{(r)} \, p_{P_{\pi} (i_{\mathbf{R_K}},z,i_{\mathbf{K-1}})}^{r,a,\alpha} \, \left(p_{P_{\pi}(j_{\mathbf{R_K}},z,i_{\mathbf{K-1}})}^{r,a,\beta}\right)^{*} \Biggr]. 
    \end{split}
\end{equation*}
Together with the condition that
\begin{equation*}
    \begin{split}
    \operatorname{Tr}(C)
    &=\; \sum_{r \in I} \sum_{\alpha,\beta = 1}^{m_r} \sum_{a=1}^{\dim(V_r)}
    c_{\alpha,\beta}^{(r)}\, \tr\bigl(\ket{V_{r,a,\alpha}}\bra{V_{r,a,\beta}}\bigr)\\
    &=\; \sum_{r \in I} \sum_{\alpha = \beta = 1}^{m_r} \sum_{a = 1}^{\dim(V_r)} c_{\alpha,\beta}^{(r)}\\
    &=\; \sum_{r \in I} \dim(V_{r}) \tr(C^{(r)}) = d^{t+1},
    \end{split}
\end{equation*}
we get the constraints for optimization variables $c_{\alpha,\beta}^{(r)}$.

\subsubsection{Simplification for the target function}
In this section, we will deal with the target function, the main part is to compute the following equation:
\begin{equation*}
C \ast |U\rangle\!\rangle\!\langle\!\langle U|^{\otimes t} = \operatorname{Tr}_{2\cdots(2t+1)}(C^{T_{2\cdots(2t+1)}}(I \otimes |U\rangle\!\rangle\!\langle\!\langle U|^{\otimes t}_{2\cdots(2t+1)} \otimes I))
\end{equation*}
where $I = I_d$ and recall that
\begin{equation*}
|U\rangle\!\rangle = (I \otimes U)|I\rangle\!\rangle = (I \otimes U)(\sum_{i=1}^{d} |e_{i}\rangle \otimes |e_i\rangle).
\end{equation*}
By the equation \eqref{eq:form-of-C}, we can derive that
\begin{equation*}
    \begin{split}
    C^{T_{2\cdots(2t+1)}} = \; &\sum_{r \in I}\sum_{\alpha,\beta = 1}^{m_r} \sum_{a=1}^{\dim(V_r)}
    \sum_{\substack{i_{1},\cdots,i_{2t+2} \in [d] \\[0.4em] j_{1},\cdots,j_{2t+2} \in [d]}} 
    c_{\alpha,\beta}^{(r)}\, p_{i_{1},i_{2},\cdots,i_{2t+2}}^{r,a,\alpha}\, \Big(p_{j_{1},j_{2},\cdots,j_{2t+2}}^{r,a,\beta}\Big)^{*}\\[1mm]
    &|e_{i_{\pi(1)}} \otimes e_{j_{\pi(2)}} \otimes \cdots \otimes e_{j_{\pi(2t+1)}} \otimes e_{i_{\pi(2t+2)}} \rangle \\ 
    &\langle e_{j_{\pi(1)}} \otimes e_{i_{\pi(2)}} \otimes \cdots \otimes e_{i_{\pi(2t+1)}} \otimes e_{j_{\pi(2t+2)}}|.
    \end{split}
\end{equation*}
Next, we can also get that
\begin{equation*}
    \begin{split}
        &I \otimes |U\rangle\!\rangle\!\langle\!\langle U|^{\otimes t}_{2\cdots(2t+1)} \otimes I \\
         = \sum_{k_{1},\cdots,k_{t+2} \in [d]} &|e_{k_1} \otimes e_{k_2} \otimes Ue_{k_2} \otimes \cdots \otimes e_{k_{t+1}} \otimes Ue_{k_{t+1}} \otimes e_{k_{t+2}} \rangle\\
        &\langle e_{k_1} \otimes e_{k_2} \otimes Ue_{k_2} \otimes \cdots \otimes e_{k_{t+1}} \otimes Ue_{k_{t+1}} \otimes e_{k_{t+2}}|.
    \end{split}  
\end{equation*}
When we compute the multiplication of these two matrices, note that
\begin{equation*}
U|e_{i_k}\rangle = \sum_{m=1}^{d} U_{m,i_{k}}|e_{m}\rangle,
\end{equation*}
where $U_{i,j}$ represents the element on row $i$ and column $j$ of $U$. Then
\begin{equation*}
    \begin{aligned}
        \big\langle
        &e_{j_{\pi(1)}} \otimes e_{i_{\pi(2)}} \otimes \cdots
        \otimes e_{i_{\pi(2t+1)}} \otimes e_{j_{\pi(2t+2)}}
        \big|
        \\
        &
        e_{k_1} \otimes e_{k_2} \otimes Ue_{k_2} \otimes \cdots
        \otimes e_{k_{t+1}} \otimes Ue_{k_{t+1}} \otimes e_{k_{t+2}}
        \big\rangle
        \\
        &=
        \delta_{k_{1},\, j_{\pi(1)}}\,
        \delta_{k_{t+2},\, j_{\pi(2t+2)}}
        \prod_{l=2}^{t+1}
        \delta_{k_l,\, i_{\pi(2l-2)}}\,
        \langle e_{i_{\pi(2l-1)}},\, Ue_{k_l} \rangle
        \\
        &=
        \delta_{k_{1},\, j_{\pi(1)}}\,
        \delta_{k_{t+2},\, j_{\pi(2t+2)}}
        \prod_{l=2}^{t+1}
        \delta_{k_l,\, i_{\pi(2l-2)}}\,
        U_{i_{\pi(2l-1)},\, i_{\pi(2l-2)}}
    \end{aligned}
\end{equation*}
from which we can get the following equation
    \begin{align*}
        &
        C^{T_{2\cdots(2t+1)}}\!
        \left(
        I \otimes |U\rangle\!\rangle\!\langle\!\langle U|^{\otimes t}_{2\cdots(2t+1)} \otimes I
        \right)
        \\
        =\, &
        \sum_{r \in I}\
        \sum_{\alpha,\beta = 1}^{m_r}\
        \sum_{a=1}^{\dim(V_r)}\
        \sum_{\substack{
                i_{1},\ldots,i_{2t+2} \in [d] \\
                j_{1},\ldots,j_{2t+2} \in [d]
        }}
        \sum_{m_{1},\ldots,m_{t}=1}^{d}c_{\alpha,\beta}^{(r)} p_{i_{1},\ldots,i_{2t+2}}^{r,a,\alpha}\Bigl(p_{j_{1},\ldots,j_{2t+2}}^{r,a,\beta}\Bigr)^{*}
        \\[-2mm]
        &\times
        \prod_{l=2}^{t+1} U_{i_{\pi(2l-1)},\,i_{\pi(2l-2)}}|e_{i_{\pi(1)}} \otimes e_{j_{\pi(2)}} \otimes e_{j_{\pi(3)}}
        \otimes \cdots \otimes e_{j_{\pi(2t+1)}} \otimes
        e_{i_{\pi(2t+2)}} \rangle
        \\
        &\times
        \langle
        e_{j_{\pi(1)}}
        \otimes e_{i_{\pi(2)}}
        \otimes U_{m_{1},i_{\pi(2)}} e_{m_1}
        \otimes \cdots \otimes
        e_{i_{\pi(2t)}} \otimes
        U_{m_{t},i_{\pi(2t)}} e_{m_t}
        \otimes e_{j_{\pi(2t+2)}}
        |.
    \end{align*}
Therefore, we can derive that
\begin{equation*}
    \begin{aligned}
        C \ast |U\rangle\!\rangle\!\langle\!\langle U|^{\otimes t} = &\sum_{r \in I} \sum_{\alpha, \beta = 1}^{m_r} \sum_{a = 1}^{\dim(V_r)} \sum_{\substack{
                i_{1},\ldots,i_{2t+2} \in [d] \\
                j_{1},\ldots,j_{2t+2} \in [d]\\
                i_{2k} = j_{2k}, \, k=1,\ldots,t
        }} c_{\alpha,\beta}^{(r)} \, p_{P_{\pi}(i_{1},\cdots,i_{2t+2})}^{r,a,\alpha} \left(p_{P_{\pi}(j_{1},\cdots,j_{2t+2})}^{r,a,\beta}\right)^*\\
    & \prod_{l=2}^{t+1} U_{i_{2l-1},i_{2l-2}}(U_{j_{2l-1},i_{2l-2}})^{*} |e_{i_1}\rangle \langle e_{j_1}| \otimes \ket{e_{i_{2t+2}}} \bra{e_{j_{2t+2}}}. 
    \end{aligned}
\end{equation*}
Note that we have the following relationship
\begin{equation*}
     \tr_{F}\!\Big[
    F\big(C \ast |U\rangle\!\rangle\!\langle\!\langle U|_{\boldsymbol{IO}}^{\otimes t}\big)^{\intercal}F
    (O_{F}^{\intercal} \otimes I_P)\Big] =\tr_{F}\!\Big[
    \big(C \ast |U\rangle\!\rangle\!\langle\!\langle U|_{\boldsymbol{IO}}^{\otimes t}\big)^{\intercal}(I_{P} \otimes O_{F}^{\intercal})\Big]
\end{equation*}
where on the left side the system order is $\mathcal{H}_{F} \otimes \mathcal{H}_{P}$ while on the right side it is $\mathcal{H}_{P} \otimes \mathcal{H}_{F}$. That is, the partial trace on the left-hand side is taken over the first subsystem, whereas on the right-hand side it is taken over the second subsystem. Then we define the operator
\begin{equation*}
\begin{aligned}
    S^{\pi}(U,O) = &\sum_{r \in I} \sum_{\alpha, \beta = 1}^{m_r} \sum_{a = 1}^{\dim(V_r)} \sum_{\substack{
            i_{1},\ldots,i_{2t+2} \in [d] \\
            j_{1}, \ldots, j_{2t+2} \in [d]\\
            i_{2k} = j_{2k}, k=1,\ldots,t
    }} c_{\alpha,\beta}^{(r)} \, p_{P_{\pi}(i_{1},\cdots,i_{2t+2})}^{r,a,\alpha} \left(p_{P_{\pi}(j_{1},\cdots,j_{2t+2})}^{r,a,\beta}\right)^*\\
    & \prod_{l=2}^{t+1} U_{i_{2l-1},i_{2l-2}}(U_{j_{2l-1},i_{2l-2}})^{*}
     \operatorname{Tr}\!\left(|e_{j_{2t+2}}\rangle \langle e_{i_{2t+2}}|\,O^{\intercal}\right) \, |e_{j_{1}}\rangle \langle e_{i_{1}}|.
\end{aligned}
\end{equation*}
Now we are ready to transform the constraints and the target function in the SDP optimization problem into constraints and expressions on the variables $c_{\alpha,\beta}^{(r)}$ as
\begin{equation*}
\begin{gathered}
    \min \quad \int_{U} \Bigl\| S^{\pi}(U,O) - UOU^{\dagger} \Bigr\|_{F}\, d\mu_{H}(U), \\[2mm]
    \begin{aligned}
        \textbf{(C1)}:\quad & \sum_{r \in I}\sum_{\alpha,\beta = 1}^{m_r} \sum_{a=1}^{dim(V_r)}\sum_{i_{\mathbf{K-1}}  \in [d]^{|\mathbf{K-1}|}} c_{\alpha,\beta}^{(r)} \, p_{P_{\pi}(i_{\mathbf{R_K}},x,i_{\mathbf{K-1}})}^{r,a,\alpha} \left(p_{P_{\pi}(j_{\mathbf{R_K}},y,i_{\mathbf{K-1}})}^{r,a,\beta}\right)^{*} = 0., \\[1mm]
        &\forall\, K \in \{2,4,\ldots,2t+2\},\ \forall\, i_{\mathbf{R_K}},j_{\mathbf{R_K}}\in [d]^{|\mathbf{R_K}|},\ \forall\, x,y\in [d]\ \text{with}\ x\neq y, \\[3mm]
        \textbf{(C2)}:\quad &\sum_{r \in I}\sum_{\alpha,\beta = 1}^{m_r} \sum_{a=1}^{dim(V_r)}\sum_{i_{\mathbf{K-1}}  \in [d]^{|\mathbf{K-1}|}} c_{\alpha,\beta}^{(r)} \, p_{P_{\pi}(i_{\mathbf{R_K}},k,i_{\mathbf{K-1}})}^{r,a,\alpha}\left(p_{P_{\pi}(j_{\mathbf{R_K}},k,i_{\mathbf{K-1}})}^{r,a,\beta}\right)^{*}\\
        =\frac{1}{d}&\sum_{r \in I}\sum_{\alpha,\beta = 1}^{m_r} \sum_{a=1}^{dim(V_r)}\sum_{i_{\mathbf{K-1}}  \in [d]^{|\mathbf{K-1}|}} \Biggl[\sum_{z \in [d]}c_{\alpha,\beta}^{(r)} \, p_{P_{\pi}(i_{\mathbf{R_K}},z,i_{\mathbf{K-1}})}^{r,a,\alpha}\left(p_{P_{\pi}(j_{\mathbf{R_K}},z,i_{\mathbf{K-1}})}^{r,a,\beta}\right)^{*} \Biggr], \\[1mm]
        & \quad \forall\, K \in \{2,4,\ldots,2t+2\},\ \forall\, i_{\mathbf{R_K}},j_{\mathbf{R_K}} \in [d]^{|\mathbf{R_K}|},\ \forall\, k\in [d], \\[3mm]
        \textbf{(C3)}:\quad &
        \begin{aligned}[t]
            \operatorname{Tr}(C)
            = \sum_{r \in I} \dim(V_{r})\,\operatorname{Tr}\!\bigl(C^{(r)}\bigr) = d^{t+1},
        \end{aligned}
        \\[3mm]
        \textbf{(C4)}:\quad & C^{(r)} = [c_{\alpha,\beta}^{(r)}]_{1 \le \alpha ,\beta \le m_r} \geq 0,\quad \forall\, r \in I.
    \end{aligned}
\end{gathered}
\end{equation*}
Next, we explain how to transform the multiple summation expression into an equivalent block-matrix formulation, and how to pre-process the data in matrix $Q_{O,t}$ to effectively exploit the information it contains. For fixed $K \in \{1,2,\cdots,2t+2\}$, $r \in I$ and $\pi \in S_{2t+2}$, we define
\begin{equation*}
    \begin{aligned}
    \operatorname{Supp}_{r,K}^{\pi}(i_{\mathbf{R_K}},x) := \bigl\{&(a,i_{\mathbf{K-1}}) \in [\dim(V_r)] \times [d]^{K-1}|\\
     &\exists \alpha, s.t.\,\, p_{P_{\pi}(i_{\mathbf{R_K}},x,i_{\mathbf{K-1}})}^{r,a,\alpha} \neq 0\bigr\}
    \end{aligned}
\end{equation*}
for every $(i_{\mathbf{R_K}},x) \in [d]^{|\mathbf{R_K}|} \times [d]$ and define
\begin{equation*}
    \operatorname{S}_{r,K}^{\pi} := \bigl\{(i_{\mathbf{R_K}},x) \in [d]^{|\mathbf{R_K}|} \times [d] | \,\, \operatorname{Supp}_{r,K}^{\pi}(i_{\mathbf{R_K}},x) \neq \emptyset\bigr\}.
\end{equation*}
Then for every $\big((i_{\mathbf{R_K}},x),(j_{\mathbf{R_K}},y)\big) \in \operatorname{S}_{r,K}^{\pi} \times \operatorname{S}_{r,K}^{\pi}$, we define the following matrix $M_{K}^{r,\pi}(i_{\mathbf{R_K}},x,j_{\mathbf{R_K}},y) \in \mathbb{C}^{m_{r} \times m_{r}}$: 
\begin{equation}\label{def:MKR}
[M_{K}^{r,\pi}(i_{\mathbf{R_K}},x,j_{\mathbf{R_K}},y)]_{\beta,\alpha} = \sum_{
    \substack{
        (a,i_{\mathbf{K-1}}) \in \operatorname{Supp}_{r,K}^{\pi}(i_{\mathbf{R_K}},x) \\
        \cap\ \operatorname{Supp}_{r,K}^{\pi}(j_{\mathbf{R_K}},y)
    }
} p_{P_{\pi}(i_{\mathbf{R_K}},x,i_{\mathbf{K-1}})}^{r,a,\alpha} \left(p_{P_{\pi}(j_{\mathbf{R_K}},y,i_{\mathbf{K-1}})}^{r,a,\beta}\right)^{*}
\end{equation}
where if 
$
\operatorname{Supp}_{r,K}^{\pi}(i_{\mathbf{R_K}},x) \cap \operatorname{Supp}_{r,K}^{\pi}(j_{\mathbf{R_K}},y) = \emptyset,
$
we will set
\begin{equation*}
    M_{K}^{r,\pi}(i_{\mathbf{R_K}},x,j_{\mathbf{R_K}},y) = 0.
\end{equation*}
Next we denote the set that
\begin{equation}\label{def:SK}
    \begin{aligned}
    S_{K}^{\pi} := \bigcup_{r \in I} \bigl\{&((i_{\mathbf{R_K}},x),(j_{\mathbf{R_K}},y)) \in S_{r,K}^{\pi} \times S_{r,K}^{\pi}|\\  &\operatorname{Supp}_{r,K}^{\pi}(i_{\mathbf{R_K}},x) \cap \operatorname{Supp}_{r,K}^{\pi}(j_{\mathbf{R_K}},y) \neq \emptyset\bigr\}, 
    \end{aligned}
\end{equation}
then the conditions $\textbf{(C1)}$ and $\textbf{(C2)}$ are equivalent to
\begin{equation*}
    \begin{aligned} &\sum_{r \in I} \tr \big[C^{(r)}    M_{K}^{r,\pi}(i_{\mathbf{R_K}},x,j_{\mathbf{R_K}},y)\big] = \frac{\delta_{x,y}}{d} \sum_{z \in [d]} \sum_{r \in I} \tr \big[C^{(r)} M_{K}^{r,\pi}(i_{\mathbf{R_K}},z,j_{\mathbf{R_K}},z)\big]\\
        & \quad \forall K \in \{2,4,\cdots,2t+2\}, \forall \big((i_{\mathbf{R_K}},x),(j_{\mathbf{R_K}},y)\big) \in S_{K}^{\pi}
    \end{aligned}
\end{equation*}
where if $\big((i_{\mathbf{R_K}},x),(j_{\mathbf{R_K}},y)\big) \notin \operatorname{S}_{r_0,K}^{\pi} \times \operatorname{S}_{r_0,K}^{\pi}$ for some $r_{0} \in I$, we also set 
\begin{equation*}
    M_{K}^{r_0,\pi}(i_{\mathbf{R_K}},x,j_{\mathbf{R_K}},y) = 0.
\end{equation*}
In the similar way, we will deal with the target function as below. For each fixed $r,a, \alpha,\pi$ and $i \in [d]$, we introduce the following notation
\begin{equation*}
    N(r,a,\alpha,\pi,i) := \{(i_{2},\cdots,i_{2t+2}) \in [d]^{2t+1}| \,\, p_{P_{\pi}(i,i_{2},\cdots,i_{2t+2})}^{r,a,\alpha} \neq 0\}.
\end{equation*}
Then for each $(j_{1},i_{1}) \in [d] \times [d]$, we construct $M_{U}^{r,\pi}(j_{1},i_{1}) \in \mathbb{C}^{m_{r} \times m_r}$ as
\begin{equation}\label{def-MUO}
    \begin{aligned}
    [M_{U,O}^{r,\pi}(j_{1},i_{1})]_{\beta,\alpha} = &\sum_{a = 1}^{\dim(V_r)} \sum_{\substack{
            (i_{2},\ldots,i_{2t+2}) \in N(r,a,\alpha,\pi,i_{1}) \\
            (j_{2}, \ldots, j_{2t+2}) \in N(r,a,\beta,\pi,j_{1})\\
            i_{2k} = j_{2k}, k=1,\ldots,t
    }} p_{P_{\pi}(i_{1},\cdots,i_{2t+2})}^{r,a,\alpha} \left(p_{P_{\pi}(j_{1}, \cdots, j_{2t+2})}^{r,a,\beta}\right)^*\\
    & \prod_{l=2}^{t+1} U_{i_{2l-1},i_{2l-2}}(U_{j_{2l-1},i_{2l-2}})^{*}
    \operatorname{Tr}\!\left(|e_{j_{2t+2}}\rangle \langle e_{i_{2t+2}}|\,O^{\intercal}\right).
    \end{aligned}
\end{equation}
Then we will have the following equation
\begin{equation*}
    S^{\pi}(U,O) = \sum_{r \in I} \sum_{i_{1},j_{1} \in [d]} \tr\big[C^{(r)}M_{U,O}^{r,\pi}(j_1,i_1)\big] |e_{j_1}\rangle \langle e_{i_1}|.
\end{equation*}
Hence we can get the following simplified SDP written in block-matrix form
\begin{theorem}
    {For any $d,t \in \mathbb{N}^+$, let $O$ be some fixed $d$-dimensional observable, then the SDP \eqref{formulation-SDP} for general $t$-query shadow inversion of $d$-dimensional unitary under $O$ in the setting of sequential quantum combs is equivalent to the following SDP:}
    \begin{equation*}
        \begin{gathered}
        \min_{\{C^{(r)}\}_{r \in I}}\quad \int_{U} \Bigl\| \sum_{r \in I} \sum_{i_{1},j_{1} \in [d]} \tr\big[C^{(r)}M_{U,O}^{r,\pi}(j_1,i_1)\big] |e_{j_1}\rangle \langle e_{i_1}| - UOU^{\dagger} \Bigr\|_{F}\, d\mu_{H}(U), \\[2mm]
            \begin{aligned}
               \text{s.t.} &\sum_{r \in I} \tr \big[C^{(r)} M_{K}^{r,\pi}(i_{\mathbf{R_K}},x,j_{\mathbf{R_K}},y)\big] = \frac{\delta_{x,y}}{d} \sum_{z \in [d]} \sum_{r \in I} \tr \big[C^{(r)} M_{K}^{r,\pi}(i_{\mathbf{R_K}},z,j_{\mathbf{R_K}},z)\big]\\
                & \quad \forall K \in \{2,4,\cdots,2t+2\}, \forall \big((i_{\mathbf{R_K}},x),(j_{\mathbf{R_K}},y)\big) \in S_{K}^{\pi}, \\[3mm]
                &\operatorname{Tr}(C)
                    = \sum_{r \in I} \dim(V_{r})\,\operatorname{Tr}\!\bigl(C^{(r)}\bigr) = d^{t+1},
                \\[3mm]
                & C^{(r)} = [c_{\alpha,\beta}^{(r)}]_{1 \le \alpha ,\beta \le m_r} \geq 0,\quad \forall\, r \in I
            \end{aligned}
        \end{gathered}
        \end{equation*}
    where $M_{K}^{r,\pi}(i_{\mathbf{R_K}},x,j_{\mathbf{R_K}},y)$, $S_K^{\pi}$ and $M_{U,O}^{r,\pi}(j_1,i_1)$ are defined in \eqref{def:MKR}, \eqref{def:SK} and \eqref{def-MUO} respectively.
\end{theorem}
\noindent Next we will directly give the results about the simplified SDP for parallel situation. Let $P_{\sigma}$ be any permutation operator of $\sigma \in S_{2t+2}$ that maps the tensor factors from the grouped ordering
\begin{equation*}
(P, I_1, \ldots, I_t, O_{1}, \ldots, O_t, F)
\end{equation*}
to the grouped ordering
\begin{equation*}
(P, O_1, \ldots, O_t, I_1, \ldots, I_t, F),
\end{equation*}
i.e., all output systems come before all input systems (while the relative order within each group is irrelevant). For example, we can take $\sigma$ as follows
\begin{equation*}
    \sigma : 
    \begin{cases}
        1 \;\mapsto\; 1, \\[6pt]
        j \;\mapsto\; t+j, & \text{for } j = 2, \dots, t+1, \\[6pt]
        k \;\mapsto\; k-t, & \text{for } k = t+2, \dots, 2t+1, \\[6pt]
        2t+2 \;\mapsto\; 2t+2.
    \end{cases}
\end{equation*}
Then for fixed $K \in \{1,2,\cdots,2t+1\}$,$r \in I$ and $\sigma \in S_{2t+2}$, we define 
\begin{equation*}
    \begin{aligned}
\widetilde{\operatorname{Supp}}_{r,K}^{\sigma}(i_{\mathbf{R_K}},i_{2t+1-K}) := \bigl\{&(a,i_{2K-2t-1}) \in [\dim(V_r)] \times [d]^{2K-2t-1}|\\
     &\exists \alpha, s.t.\,\, p_{P_{\sigma}(i_{\mathbf{R_K}},i_{2t+1-K},i_{2K-2t-1})}^{r,a,\alpha} \neq 0\bigr\}
    \end{aligned}
\end{equation*}
for every $(i_{\mathbf{R_K}},i_{2t+1-K}) \in [d]^{|\mathbf{R_K}|} \times [d]^{2t+1-K}$ and define
\begin{equation*}
\widetilde{\operatorname{S}}_{r,K}^{\sigma} := \bigl\{(i_{\mathbf{R_K}},i_{2t+1-K}) \in [d]^{|\mathbf{R_K}|} \times [d]^{2t+1-K} | \,\, \widetilde{\operatorname{Supp}}_{r,K}^{\sigma}(i_{\mathbf{R_K}},i_{2t+1-K}) \neq \emptyset\bigr\}.
\end{equation*}
Then for every $\big((i_{\mathbf{R_K}},i_{2t+1-K}),(j_{\mathbf{R_K}},j_{2t+1-K})\big) \in \widetilde{\operatorname{S}}_{r,K}^{\sigma} \times \widetilde{\operatorname{S}}_{r,K}^{\sigma}$, we define the following matrix $\widetilde{M}_{K}^{r,\sigma}(i_{\mathbf{R_K}},i_{2t+1-K},j_{\mathbf{R_K}},j_{2t+1-K}) \in \mathbb{C}^{m_{r} \times m_{r}}$: 
\begin{equation}\label{def:MKR2}
\begin{aligned}
&[\widetilde{M}_{K}^{r,\sigma}(i_{\mathbf{R_K}},i_{2t+1-K},j_{\mathbf{R_K}},j_{2t+1-K})]_{\beta,\alpha}\\
= &\sum_{
    \substack{
        (a,i_{2K-2t-1}) \in \widetilde{\operatorname{Supp}}_{r,K}^{\sigma}(i_{\mathbf{R_K}},i_{2t+1-K}) \\
        \cap\ \widetilde{\operatorname{Supp}}_{r,K}^{\sigma}(j_{\mathbf{R_K}},j_{2t+1-K})
    }
} p_{P_{\sigma}(i_{\mathbf{R_K}},i_{2t+1-K},i_{2K-2t-1})}^{r,a,\alpha} \left(p_{P_{\sigma}(j_{\mathbf{R_K}},j_{2t+1-K},i_{2K-2t-1})}^{r,a,\beta}\right)^{*}
\end{aligned}
\end{equation}
where if 
$
\widetilde{\operatorname{Supp}}_{r,K}^{\sigma}(i_{\mathbf{R_K}},i_{2t+1-K}) \cap \widetilde{\operatorname{Supp}}_{r,K}^{\sigma}(j_{\mathbf{R_K}},j_{2t+1-K}) = \emptyset,
$
we will set
\begin{equation*}
    \widetilde{M}_{K}^{r,\sigma}(i_{\mathbf{R_K}},i_{2t+1-K},j_{\mathbf{R_K}},j_{2t+1-K}) = 0.
\end{equation*}
Next we denote the set 
\begin{equation}\label{def:SK2}
    \begin{aligned}
    \widetilde{S}_{K}^{\sigma} := \bigcup_{r \in I} \bigl\{&((i_{\mathbf{R_K}},i_{2t+1-K}),(j_{\mathbf{R_K}},j_{2t+1-K})) \in \widetilde{S}_{r,K}^{\sigma} \times \widetilde{S}_{r,K}^{\sigma}|\\  &\widetilde{\operatorname{Supp}}_{r,K}^{\sigma}(i_{\mathbf{R_K}},i_{2t+1-K}) \cap \widetilde{\operatorname{Supp}}_{r,K}^{\sigma}(j_{\mathbf{R_K}},j_{2t+1-K}) \neq \emptyset\bigr\}.
    \end{aligned}
\end{equation}
If $((i_{\mathbf{R_K}},i_{2t+1-K}),(j_{\mathbf{R_K}},j_{2t+1-K})) \notin \widetilde{S}_{r_0,K}^{\sigma} \times \widetilde{S}_{r_0,K}^{\sigma}$ for some $r_0 \in I$, we will set $\widetilde{M}_{K}^{r_0,\sigma}(i_{\mathbf{R_K}},i_{2t+1-K},j_{\mathbf{R_K}},j_{2t+1-K}) = 0$. Moreover, for each $(j_{1},i_{1}) \in [d] \times [d]$, we construct $\widetilde{M}_{U}^{r,\sigma}(j_{1},i_{1}) \in \mathbb{C}^{m_{r} \times m_r}$ as
\begin{equation}\label{def-MUO2}
    \begin{aligned}
    [\widetilde{M}_{U}^{r,\sigma}(j_{1},i_{1})]_{\beta,\alpha} = &\sum_{a = 1}^{\dim(V_r)} \sum_{\substack{
            (i_{2},\ldots,i_{2t+2}) \in N(r,a,\alpha,\sigma,i_{1}) \\
            (j_{2}, \ldots, j_{2t+2}) \in N(r,a,\beta,\sigma,j_{1})\\
            i_{k} = j_{k}, k= 2,\ldots,t+1
    }} p_{P_{\sigma}(i_{1},\cdots,i_{2t+2})}^{r,a,\alpha} \left(p_{P_{\sigma}(j_{1}, \cdots, j_{2t+2})}^{r,a,\beta}\right)^*\\
    & \prod_{l=2}^{t+1} U_{i_{t+l},i_{l}}(U_{j_{t+l},i_{l}})^{*}
\operatorname{Tr}\!\left(|e_{j_{2t+2}}\rangle \langle e_{i_{2t+2}}|\,O^{\intercal}\right).
    \end{aligned}
\end{equation}
Then we will give the results about the simplified SDP for parallel quantum comb.
\begin{theorem}
    {For any $d,t \in \mathbb{N}^+$, let $O$ be some fixed $d$-dimensional observable, then the SDP for general $t$-query shadow inversion of $d$-dimensional unitary under $O$ in the setting of parallel quantum combs is equivalent to the following SDP:}
    \begin{equation*}
        \begin{gathered}
        \min_{\{C^{(r)}\}_{r \in I}}\quad \int_{U} \Bigl\| \sum_{r \in I} \sum_{i_{1},j_{1} \in [d]} \tr\big[C^{(r)}\widetilde{M}_{U,O}^{r,\sigma}(j_1,i_1)\big] |e_{j_1}\rangle \langle e_{i_1}| - UOU^{\dagger} \Bigr\|_{F}\, d\mu_{H}(U), \\[2mm]
            \begin{aligned}
               \text{s.t.} &\sum_{r \in I} \tr \big[C^{(r)} \widetilde{M}_{t+1}^{r,\sigma}(i_{\mathbf{R_{t+1}}},i_{t},j_{\mathbf{R_{t+1}}},j_{t})\big] = \frac{\delta_{i_{t},j_{t}}}{d^{\intercal}} \sum_{k_t \in [d]^{\intercal}} \sum_{r \in I} \tr \big[C^{(r)} \widetilde{M}_{t+1}^{r,\sigma}(i_{\mathbf{R_{t+1}}},k_t,j_{\mathbf{R_{t+1}}},k_t)\big]\\
               & \quad \forall \big((i_{\mathbf{R_{t+1}}},i_t),(j_{\mathbf{R_{t+1}}},j_t)\big) \in \widetilde{S}_{t+1}^{\sigma}, \\[3mm]
               &\sum_{r \in I} \tr \big[C^{(r)} \widetilde{M}_{2t+1}^{r,\sigma}(i_{\mathbf{R_{2t+1}}},j_{\mathbf{R_{2t+1}}})\big] = \frac{\delta_{i_{\mathbf{R}_{2t+1}},j_{\mathbf{R}_{2t+1}}}}{d} \sum_{k_{\mathbf{R}_{2t+1}} \in [d]} \sum_{r \in I} \tr \big[C^{(r)} \widetilde{M}_{2t+1}^{r,\sigma}(k_{\mathbf{R_{2t+1}}},k_{\mathbf{R_{2t+1}}})\big]\\
                & \quad \forall \big(i_{\mathbf{R_{2t+1}}},j_{\mathbf{R_{2t+1}}}\big) \in \widetilde{S}_{2t+1}^{\sigma}, \\[3mm]
                &\operatorname{Tr}(C)
                    = \sum_{r \in I} \dim(V_{r})\,\operatorname{Tr}\!\bigl(C^{(r)}\bigr) = d^{t+1},
                \\[3mm]
                & C^{(r)} = [c_{\alpha,\beta}^{(r)}]_{1 \le \alpha ,\beta \le m_r} \geq 0,\quad \forall\, r \in I.
            \end{aligned}
        \end{gathered}
        \end{equation*}
        where $\widetilde{M}_{K}^{r,\sigma}(i_{\mathbf{R_K}},i_{2t+1-K},j_{\mathbf{R_K}},j_{2t+1-K})$, $\widetilde{S}_{K}^{\sigma}$ for $K \in \{1,2,\cdots,2t+1\}$ and $\widetilde{M}_{U,O}^{r,\sigma}(j_1,i_1)$ are defined in \eqref{def:MKR2}, \eqref{def:SK2} and \eqref{def-MUO2} respectively.
\end{theorem}
\begin{remark}
    In fact, the Schur matrix $Q_{O,t}$ we construct by the Young tableau method introduced in Section~\ref{g-r-t} is real and hence we can omit the conjugation of the coefficients $(p_{i_{1},\cdots,i_{2t+2}}^{r,a,\alpha})^*$ for any $r \in I,a \in \dim(V_r), \alpha \in m_r$ and $(i_{1},\cdots,i_{2t+2}) \in [d]^{2t+2}$. That is, we can use $p_{i_{1},\cdots,i_{2t+2}}^{r,a,\alpha}$ instead of $(p_{i_{1},\cdots,i_{2t+2}}^{r,a,\alpha})^*$.
\end{remark}

{For any $d,t \in \mathbb{N}^+$, the size of the original variable block in the SDP for general $t$-query shadow inversion of $d$-dimensional unitary under $O$ grows as $d^{4t+4}$. By recognizing symmetry property and block-diagonalizing the variable through group representation, the original single-block is replaced by a set of smaller blocks \( C^{(r)} \) of sizes \( m_r \times m_r \). Consequently, the total number of variables is reduced to
\begin{equation*}
\sum_{r} m_r^2 = m \, I_{t+1}(d) \, J_t(d)
\end{equation*}
by Proposition~\ref{num-ub}. This reduction offers an exponential advantage for large $d$. For example, when \( d=6 \), \( t=3 \), \( m=2 \), and \( (l_1,l_2)=(3,3) \), we can compute that \( I_4(6)=24 \) and \( J_3(6)=48 \), leading to
\begin{equation*}
\sum_{r} m_{r}^2 = 2 \times 24 \times 48 = 2304
\end{equation*}
while the original size is $6^{16}$. This dramatic compression not only reduces dimensionality but also offers substantial practical benefits: the block-diagonal structure transforms each iteration from a large-scale matrix decomposition into multiple smaller, parallelizable decompositions, significantly lowering both memory and computational costs.}

\end{document}